\documentclass[onecolumn,amsmath,amssymb,nofootinbib,superscriptaddress]{revtex4-2}

\usepackage[utf8]{inputenc}
\usepackage[margin=1in]{geometry}
\usepackage{amsmath, amsthm, amssymb,amscd, mathrsfs, amsfonts, mathtools,tikz-cd,pgfplots}
\usepackage{appendix}
\usepackage{graphicx}
\usepackage{relsize}
\usepackage{mathtools}
\usepackage{dsfont}
\usepackage{float}
\usepackage{soul}
\usepackage{quantikz}
\usepackage{afterpage}
\usetikzlibrary{shapes, backgrounds}

\tikzset{
  qaoaBlock/.style={
    draw,
    rounded corners=2pt,
    minimum width=1.4cm,
    minimum height=0.8cm,
    align=center,
    font=\small,
  },
  problem/.style={qaoaBlock, fill=green!30},
  mixer/.style={qaoaBlock, fill=blue!30},
}
\usetikzlibrary{decorations.pathmorphing}
\usetikzlibrary{arrows.meta}
\usetikzlibrary{patterns}
\usepackage{chngcntr}
\usepackage[bookmarks=true,
bookmarksnumbered=true,
breaklinks=true,
pdfstartview=FitH,
hyperfigures=false,
plainpages=false,
naturalnames=true,
colorlinks=true,
linkcolor=blue,       
citecolor=blue,       
urlcolor=blue,        
pagebackref=true,
pdfpagelabels]{hyperref}
\usepackage{ORCIDinREVTeX}
\theoremstyle{definition}
\newtheorem{thm}{Theorem}[section]
\newtheorem{prop}[thm]{Proposition}
\newtheorem{lem}[thm]{Lemma}

\newtheorem{obs}{Observation}
\newtheorem{defn}[thm]{Definition}
\newtheorem{cor}[thm]{Corollary}
\newtheorem{rmk}[thm]{Remark}
\newtheorem{conj}[thm]{Conjecture}
\newtheorem{ex}[thm]{Example}

\begin{document}
\title{Reductions of QAOA Induced by Classical Symmetries: Theoretical Insights and Practical Implications }
\author{Boris Tsvelikhovskiy} 
\orcid{0000-0003-0798-7218}
\email{borist@ucr.edu}
\affiliation{Department of Mathematics, University of California, Riverside, CA 92521, USA} 

\author{Bao Bach} 
\affiliation{Department of Computer and Information Sciences, University of Delaware, Newark, DE 19716, USA} 

\author{Jose Falla}
\orcid{0000-0001-9918-2198}
\affiliation{Department of Physics and Astronomy, University of Delaware, Newark, DE 19716, USA} 

\author{Ilya Safro} 
\orcid{0000-0001-6284-7408}
\email{isafro@udel.edu}
\affiliation{Department of Computer and Information Sciences, University of Delaware, Newark, DE 19716, USA} 
\affiliation{Department of Physics and Astronomy, University of Delaware, Newark, DE 19716, USA}


\begin{abstract}


The performance of the Quantum Approximate Optimization Algorithm (QAOA) is closely tied to the structure of the dynamical Lie algebra (DLA) generated by its Hamiltonians, which determines both its expressivity and trainability. In this work, we show that classical symmetries can be systematically exploited as a design principle for QAOA. Focusing on the MaxCut problem with global bit-flip symmetry, we analyze reduced QAOA instances obtained by fixing a single variable and study how this choice affects the associated DLAs. We show that the structure of the DLAs
can change dramatically depending on which variable is held fixed.
In particular, we construct explicit examples where the dimension
collapses from exponential to quadratic, uncovering phenomena
that do not appear in the original formulation.
 
Numerical experiments on asymmetric graphs indicate that such reductions often produce DLAs of much smaller dimension, suggesting improved trainability. We also prove that any graph can be embedded into a slightly larger one (requiring only quadratic overhead) 
such that the standard reduced DLA coincides with the free reduced DLA, in most cases implying exponential dimension and irreducibility on the Hilbert space for reduced QAOA instances. These results establish symmetry-aware reduction as a principled tool for designing expressive and potentially trainable QAOA circuits.
\end{abstract}

\maketitle

\section{Introduction}

Variational quantum algorithms (VQAs) \cite{CABB}, which belong to the class of hybrid quantum--classical methods, are widely viewed as promising candidates for demonstrating quantum advantage in areas such as combinatorial optimization and machine learning across a broad range of applications~\cite{herman2023quantum,outeiral2021prospects,joseph2023quantum}. VQAs employ parameterized quantum circuits whose variational parameters are iteratively updated using classical optimization routines. A prominent example, which we focus on in this paper, is the \emph{Quantum Approximate Optimization Algorithm} (QAOA)~\cite{QAOA}, a variational method specifically designed to solve combinatorial optimization problems. In the context of  optimization problems on graphs, QAOA has been applied to many NP-hard problems, including MaxCut~\cite{QAOA}, community detection~\cite{shaydulin2019network}, and graph partitioning~\cite{ushijima2021multilevel}. These problems are typically formulated by mapping them onto classical spin-glass Hamiltonians (e.g., Ising models) and seeking to minimize the corresponding energy, which is itself an NP-hard optimization task. The potential realization of practical quantum advantage using VQAs ultimately depends on multiple factors, including qubit quality, quantum gate and circuit fidelity, achievable circuit depth, hardware connectivity, and the effectiveness of the classical optimization procedures used to identify high-quality variational parameters, to mention just a few.


Specifically, the key idea underlying QAOA is to encode a classical objective function into a Hermitian operator, the \emph{problem Hamiltonian} $H_P$, such that its ground state corresponds to an optimal solution of the original problem. The algorithm alternates between the evolution generated by $H_P$ and that of a second operator, the \emph{mixer Hamiltonian} $H_M$, which induces transitions between computational basis states. These transitions are essential for exploring the solution space and avoiding localization in suboptimal configurations. 
Despite its conceptual simplicity, 
QAOA faces significant practical challenges that hinder its scalability. Similar to many other variational algorithms, a major bottleneck is the classical optimization of the variational parameters, which becomes increasingly difficult as the circuit depth $p$ grows. The optimization landscape can be highly nonconvex, leading to substantial classical overhead that may diminish the potential quantum advantage~\cite{ZWCHL} if not completely cancel it (as even theoretically it is not expected to be exponential). An especially problematic phenomenon is the appearance of \emph{barren plateaus}, where the variance of the cost-function gradients decreases exponentially with system size, effectively stalling parameter training~\cite{mcclean2018barren,LTWS}.

A number of approaches have been proposed to mitigate this phenomenon and thereby accelerate training. One direction focuses on improved parameter initialization strategies. For example, the Beinit framework~\cite{beta} addresses barren plateaus by initializing variational parameters using a data-dependent beta distribution, whose hyperparameters are inferred from the input data, and further introduces controlled perturbations during gradient descent to prevent convergence to flat regions of the loss landscape. Complementary to initialization-based methods, parameter transferability has been shown to significantly reduce optimization complexity in QAOA. In particular, several approaches~\cite{galda2023similarity,brandao2018fixed,falla2024graph} demonstrate that optimal QAOA parameters cluster around specific values determined by local graph properties, enabling near-optimal parameters learned on smaller or structurally similar graphs to be reused for larger instances, thereby bypassing regions associated with barren plateaus. Additional strategies include using reinforcement learning based initializations for deep circuits \cite{peng2025breaking}, introducing properly engineered Markovian dissipation after each variational circuit layer \cite{sannia2024engineered} in order to mitigate barren plateaus and restore efficient trainability and many others. 
Collectively, these methods indicate that careful control of initialization, circuit structure, and optimization can somewhat help to  alleviate barren plateaus in practice.

A recently developed and powerful approach to analyzing variational quantum algorithms, including QAOA, is through the study of their associated \emph{dynamical Lie algebras} (DLAs). Given a set of Hermitian generators (Hamiltonians), the DLA is defined as the Lie algebra generated under commutation. A central insight of this framework is that the dimension and structure of the DLA are closely tied to fundamental properties of the optimization landscape, including expressivity and the possible onset of barren plateaus~\cite{FHCKYHSP,LJGCC,LCSMCC,RBSKMLC}. 
Beyond QAOA, DLAs play a crucial role in quantum machine learning (QML). In supervised QML, one seeks a parametrized quantum circuit, generated by a finite set of Hamiltonians, that approximates a target function while generalizing to unseen data. In this setting, DLAs provide rigorous criteria for trainability and characterize regimes in which barren plateaus are unavoidable~\cite{GLCCS,LTWS,WHSU}. Moreover, the DLA framework offers valuable tools for assessing classical simulability and identifying efficiently simulable subclasses of variational circuits~\cite{CLG,GLCCS}.

While substantial progress has been made in understanding QAOA from algorithmic and performance perspectives, the detailed structure of the DLAs associated with QAOA remains comparatively underexplored. Two recent works~\cite{ASYZ1,KLFCCZ} analyzed the DLAs arising from QAOA with the standard Pauli-$X$ mixer, but restricted attention to specific families of graphs, such as paths, cycles, and complete graphs, primarily in the context of the MaxCut problem. A notable exception is the recent work~\cite{MYAZ}, which considers Erdős--Rényi random graphs. There it is shown that for $G(n,p)$ with edge probability $p=0.5$, the QAOA-MaxCut DLA is, with probability at least $1 - e^{-\Omega(n)}$, the direct sum of either one or two simple Lie algebras, each of dimension $\Theta(4^n)$ (see Theorem~3 therein).

Different QAOA mixers are also considered in the literature. In particular, a complete characterization of the of the DLAs associated with Grover-mixer QAOA (GM-QAOA) was proposed recently in \cite{TNB}. In~\cite{kordonowy2025lie}, the authors describe the DLAs characterization of the  constrained XY-mixer~\cite{HWORVB}, such as the Hamming-distance mixer. The analysis is still restricted to a specific class of mixer connectivity (all-to-all or ring). Nevertheless, the standard $X$-mixer version of QAOA remains the most widely studied and implemented variant, making a systematic investigation of its associated DLAs both theoretically and practically important.

The primary objective of the present work is to advance this understanding by studying the DLAs associated with standard-mixer QAOA through the lens of symmetry reduction. We focus on optimization problems defined over binary strings of length $n$, a broad class that includes many problems of practical interest. Most such objective functions are invariant under a global bit-flip symmetry, corresponding classically to the simultaneous inversion of all bits. From a classical standpoint, this symmetry allows one to fix the value of a single bit without loss of information, as solutions of the reduced problem can be immediately lifted to solutions of the original one. While exploiting symmetries reduction in QAOA is not a new methodology \cite{tsvelikhovskiy2023symmetries,shaydulin2021classical,bravyi2020obstacles,tsvelikhovskiy2026equivariant}, using them in combination with DLA analysis is a novel and promising direction. 

At first glance, this symmetry reduction appears to offer little advantage beyond a trivial reduction in problem size. One of the main contributions of this paper is to show that, at the quantum level, the situation is considerably richer. We demonstrate that symmetry reduction can lead to nontrivial changes in the structure and dimension of the associated dynamical Lie algebras, thereby affecting the effective Hilbert space explored by QAOA and the resulting algorithmic dynamics. This reveals new structural features of QAOA that are invisible from a purely classical perspective and highlights the importance of symmetry considerations in the analysis of variational quantum algorithms.

\subsection{Main Results}

We consider optimization problems defined on binary strings of length $n$ whose objective functions are invariant under the global bit-flip symmetry. For such problems, we associate $n$ \emph{reduced optimization problems} obtained by fixing the value of a single bit and restricting the objective function to the subset of strings satisfying this constraint. In this way, we obtain $n$ reduced instances, each defined on $n-1$ bits.

The primary object of our study is the family of dynamical Lie algebras
generated by the Hamiltonians defining different variants of the
QAOA associated with these reduced instances of the MaxCut problem.  The \emph{standard dynamical Lie algebra} is generated by the Pauli-$X$ mixer acting on all $n$ qubits together with the problem Hamiltonian corresponding to the original objective function. For each reduced problem, corresponding to fixing one of the bit values, we define the associated \emph{standard reduced dynamical Lie algebra} as the Lie algebra generated by the Pauli-$X$ mixer acting on the remaining $n-1$ qubits and the problem Hamiltonian corresponding to the restricted objective function.

\begin{enumerate}
    \item We establish deterministic sufficient conditions on a graph $\Gamma$ and a chosen vertex $v\in V(\Gamma)$ under which $\mathfrak{g}^v_{\Gamma,\mathrm{std}}$, the standard reduced dynamical Lie algebra associated with $(\Gamma,v)$, coincides with the free reduced dynamical Lie algebra. The key structural criterion is formulated in Theorem~\ref{thm:ParitySeparationImpliesLocalX}, which shows that a suitable separation of vertices by parity profiles along
distance-increasing paths from $v$ implies the containment
\begin{equation}
\label{eq:ERContainment}
\mathfrak{g}^v_{\Gamma,\mathrm{std}}
\;\supseteq\;
\mathfrak{g}_{\Gamma_v,\mathrm{free}}.
\end{equation}
This containment, in turn, yields equality of the standard and free reduced dynamical Lie algebras:
\begin{equation}
\label{eq:EREquality}
\mathfrak{g}^v_{\Gamma,\mathrm{std}}
=
\mathfrak{g}^v_{\Gamma,\mathrm{free}},
\end{equation}

\item For any connected graph $\Gamma$ on $n$ vertices and any choice of vertex $v \in V(\Gamma)$, we present an explicit construction of an extended graph  $\widehat{\Gamma}_v$ with $O(n^2)$ vertices and $O(n^2)$ edges such that the standard and free reduced dynamical Lie algebras associated with the pair $(\widehat{\Gamma}_v,v)$ coincide. 

Crucially, this extension preserves the structure of the underlying optimization problem: optimal MaxCut solutions on the original graph $\Gamma$ are obtained from optimal solutions on $\widehat{\Gamma}_v$ by restriction to the original vertex set (see Theorem~\ref{thm:GraphExtensionEmbedding}). 

\item If the reduced graph $\Gamma_v$ is connected and neither bipartite nor a cycle, we show that the free reduced dynamical Lie algebra 
$\mathfrak{g}^v_{\Gamma,\mathrm{free}}$ is isomorphic to the simple Lie algebra  $\mathfrak{su}(W_v)$, that is, to the full algebra of traceless skew-Hermitian operators acting on the $2^{n-1}$-dimensional reduced Hilbert space $W_v$ (see Corollary~\ref{thm:ReducedDLAIsFullUnitary}).

\item We identify a family of graphs $\mathcal{O}_{m_1,\dots,m_k}$ (see Figure~\ref{fig:OctopusGraph}) for which the standard DLA of the full graph grows exponentially in the number of vertices, while there exists a distinguished vertex $v$  for which the standard reduced DLA has only quadratic growth. This establishes a dramatic separation between the algebraic complexity of the full and reduced systems, illustrating how symmetry reduction at a carefully chosen vertex can substantially alter the DLA structure.

\item We show that the Hilbert space associated with any reduced QAOA instance forms an irreducible representation of the free reduced dynamical Lie algebra $\mathfrak{g}^v_{\Gamma,\mathrm{free}}$.
Consequently, whenever the equality of the standard and free reduced DLAs holds (see \eqref{eq:ReducedDLAEquality}), the same Hilbert space is also an irreducible representation of the standard reduced DLA
$\mathfrak{g}^v_{\Gamma,\mathrm{std}}$.

    \item As another application of our results, we establish a sufficient condition for the DLA containment in \eqref{eq:ERContainment} to hold for connected acyclic graphs
(Theorem~\ref{thm:ReducedDLAsForTrees}).  This condition depends only on the parity-degree profiles of leaf vertices along their unique paths to a distinguished vertex $v$ and can be verified in time linear in the number of vertices (see Remark \ref{rmk:LinearTimeCheckTrees}). In addition, we present an explicit infinite family of acyclic graphs whose standard reduced dynamical Lie algebras have dimensions that grow exponentially in the number of vertices (see Example~\ref{ex:ExpFamilyTrees}).

    
    \item We perform a numerical study of the dimensions of the standard and reduced dynamical Lie algebras for asymmetric graphs on six and seven vertices (see Fig.~\ref{AsymGraphsSixNodes} and Table~\ref{tab:SevenNodesDimensions}). 
Across all instances examined, we observe that symmetry reduction can lead to a substantial decrease in the dimension of the associated dynamical Lie algebra.  These findings motivate the conjecture that for any asymmetric graph, there exists a choice of reduction vertex for which the dimension of the corresponding standard reduced dynamical Lie algebra is strictly smaller than that of the original DLA (Conjecture~\ref{mainConj}).

Determining the exact dimensions of standard and reduced dynamical Lie algebras for QAOA applied to MaxCut on large asymmetric graphs quickly becomes computationally intractable. 
To probe larger systems, we therefore exploit the established relationship between the dimension of the dynamical Lie algebra and the variance of the QAOA loss function. 
As shown in~\cite{RBSKMLC}, the variance is inversely related to the DLA dimension; in particular, an increased variance is indicative of a smaller dynamical Lie algebra.

Using this proxy, we compute the loss-function variance for both standard and reduced QAOA circuits on a collection of connected asymmetric graphs with $11$--$15$ vertices. 
In all cases studied, the reduced QAOA exhibits a consistently higher variance than the corresponding unreduced instance, providing further numerical support for Conjecture~\ref{mainConj}.

Finally, we observe that artificially attaching a leaf to an asymmetric graph that initially has no leaves often leads to a further reduction in the effective DLA dimension. 
This effect manifests as an increase in the variance of the gradient, despite the concomitant increase in the dimension of the underlying Hilbert space, suggesting that local structural modifications can have a pronounced impact on the expressivity of the associated QAOA dynamics.
    
    \item We show that the behavior of standard and reduced dynamical Lie algebras for Grover-mixer QAOA is markedly different and fully predictable. In this case, both the original and all reduced dynamical Lie algebras are isomorphic to
    \(
        \mathfrak{su}(r) \oplus \mathfrak{u}(1) \oplus \mathfrak{u}(1),
    \)
    where $r$ denotes the number of distinct values attained by the objective function.
    \item Finally, we study the dimensions of standard and reduced dynamical Lie algebras for several specific families of graphs. For star graphs $K_{1,n}$, we show that the standard reduced dynamical Lie algebra corresponding to the central vertex is isomorphic to the $3$-dimensional Lie algebra $\mathfrak{su}(2)$ for all $n$. In contrast, numerical computations for $n \leq 11$  indicate that the dimension of the corresponding standard (unreduced) dynamical Lie algebra grows with $n$. 

    For cycle graphs $C_n$, as well as for the boundary vertices of path graphs $P_n$ with $n \geq 5$, we find that the dimensions of the corresponding standard reduced dynamical Lie algebras exceed those of the standard dynamical Lie algebras.
\end{enumerate}

To summarize, we show that, under explicit and efficiently verifiable graph-theoretic conditions, the standard reduced dynamical Lie algebra matches the free reduced DLA. These conditions depend on structural properties of parity profiles along distance-increasing paths from the chosen reduction vertex. Whenever they are satisfied, the reduced QAOA dynamics generated by the standard ansatz are already as expressive as those generated by an unconstrained (multi-angle) ansatz on the reduced Hilbert space. 

We further show that this phenomenon is not exceptional: for any connected graph and any choice of reduction vertex, one can explicitly construct an extended graph with only quadratic overhead in vertices for which the standard and free reduced dynamical Lie algebras coincide. This extension preserves optimal MaxCut solutions under restriction, demonstrating that maximal dynamical expressivity can always be enforced without altering the underlying optimization landscape. In parallel, our numerical experiments on asymmetric graphs reveal a complementary effect: symmetry reduction can substantially decrease DLA dimension, suggesting the possibility of a refined control over the expressive capacity of QAOA via the choice of appropriate reduction. Furthermore, we exhibit a class of graphs where the standard DLA dimension grows exponentially with the number of vertices, while a single-vertex reduction produces a standard reduced DLA of merely quadratic dimension, demonstrating the significant potential for dynamical simplification. This dual behavior highlights symmetry reduction as a flexible design tool whose dynamical consequences can be predicted from graph structure.

From a broader perspective, our results identify dynamical Lie algebras as a precise structural lens for understanding symmetry reduction in variational quantum algorithms. This viewpoint yields direct algorithmic consequences. Since DLA dimension is tightly linked to the variance of the QAOA loss landscape, our analysis predicts when symmetry reduction should mitigate barren plateaus and improve trainability. Extensive numerical experiments on asymmetric graphs confirm this prediction: vertex reductions that lower the DLA dimension consistently lead to  higher gradient variance at depth, without degrading solution quality. This enables a practical workflow in which variance-based diagnostics identify advantageous reductions even when exact DLA computation is intractable, providing a principled, symmetry-aware preprocessing strategy for QAOA variations.  

\subsection{Structure of the Paper}

The paper is organized as follows.

In Section~\ref{sec:PreprocessingReducedQAOA}, we begin by setting up classical preprocessing techniques for general combinatorial optimization problems and recalling the standard formulation of the Quantum Approximate Optimization Algorithm. We then introduce reduced formulations of QAOA in the specific context of the MaxCut problem, which serve as the foundation for the subsequent analysis.

Section~\ref{sec:StandardReducedDLA} is devoted to the study of dynamical Lie algebras associated with both standard and multi-angle QAOA. We introduce the corresponding reduced DLAs and discuss their basic structural properties.

In Section~\ref{sec:DistinguishedElts}, we construct a collection of distinguished elements in the reduced DLAs, consisting of uniform sums of Pauli-$X$ operators supported on certain explicit subsets of vertices (see Theorems~\ref{thm:DistDLAElements}, and \ref{thm:CoolDLAElements} for precise statements). 
We show that these elements belong to the standard reduced DLA and establish sufficient conditions under which the standard reduced DLA contains the free DLA associated with the reduced graph. 
In addition, for any connected graph and any choice of reduction vertex, we construct an explicit extension with $O(n^2)$ vertices and $O(n^2)$ edges such that, for the resulting reduced problem, the standard and free reduced dynamical Lie algebras coincide.  The main results of this section are Theorems~\ref{thm:GraphExtensionEmbedding} and \ref{thm:ParitySeparationImpliesLocalX}.


Section~\ref{sec:DLASonSomeGraphs} contains a collection of explicit results and comparisons for standard and reduced DLAs associated with specific graph families, including acyclic graphs, paths, cycles, and star graphs. We also present numerical computations of DLA dimensions for asymmetric graphs on $6$ and $7$ vertices, as well as numerical estimates for selected asymmetric graphs on $11$--$15$ vertices.

Sections~\ref{sec:ReducedGMQAOA} and~\ref{sec:HilbSpaces} are devoted to a complete description of reduced DLAs for Grover-mixer QAOA and to the action of reduced DLAs on reduced Hilbert spaces, respectively. A rigorous treatment of the Hilbert space decomposition is deferred to Appendix~\ref{sec:ProofsHilbSpaceDecomp}.

Finally, in Section~\ref{sec:ConnectionLiterature}, we place our results in the context of existing literature and outline several directions for future research.

Detailed proofs and technical arguments that would otherwise interrupt the main exposition are collected in Appendix~\ref{sec:Technicalities}.

\subsection{Notation}

The following notation is used throughout the paper:

\begin{itemize}
\item \( \mathcal{N}_{v,j} \): the set of vertices at distance $j$ from \( v \);
\item \(
\mathcal{C}^v_{j,\mathbf{a}} \subseteq \mathcal{N}_{v,j}
\):
the subset of vertices $w \in \mathcal{N}_{v,j}$ for which there exists a path
\[
w_0 = v - w_1 - \cdots - w_j = w
\]
such that $\deg(w_s)\equiv a_s \pmod{2}$ for each $1\le s\le j$ for  a given binary sequence $\mathbf{a}:=(a_1,a_2,\dots,a_j)\in\mathbb{Z}_{2}^j$;

    \item $\Gamma_v$ is the graph obtained from graph $\Gamma$ by removing the vertex $v$. It is the graph with vertex and edge sets given by
\[
V(\Gamma_v) = V(\Gamma) \setminus \{v\}, 
\qquad 
E(\Gamma_v) = E(\Gamma) \setminus \{ (v,w) \mid w \in \mathcal{N}_v \},
\]

    
    \item \( \mathfrak{g}_{\Gamma, \text{free}} \): the free dynamical Lie algebra, generated by the elements 
    \[
    \left\{ iX_w \mid w \in V(\Gamma) \right\} \quad \text{and} \quad \left\{ iZ_w Z_{w'} \mid (ww') \in E(\Gamma) \right\};
    \]
    
    \item \( \mathfrak{g}_{\Gamma, \text{std}} \): the standard dynamical Lie algebra, generated by the elements 
    \[
     X:=i\sum\limits_{w \in V} X_w \quad \text{and} \quad Z:=i\sum\limits_{(ww') \in E} Z_w Z_w';
    \]

     \item \( \mathfrak{g}^v_{\Gamma, \text{free}} \): the free reduced Lie algebra, generated by the elements 
    \[
     \{ i X_w \mid w \in V(\Gamma_v) \}, \quad
  \{ i Z_w Z_w' \mid (w w') \in E,\; w,w' \neq v \}, \quad
  \{ i Z_w \mid (v w) \in E \};
    \]
    
    \item \( \mathfrak{g}^v_{\Gamma, \text{std}} \): the standard reduced dynamical Lie algebra, generated by the elements 
    \[
    \mathcal{X}_{\widehat{v}} := i\sum\limits_{w \in V(\Gamma_v)} X_i \quad \text{and} \quad  \mathcal{Z}_{\widehat{v}} :=  i\sum\limits_{(w,w') \in E(\Gamma_v) } Z_w Z_{w'} 
    + i\sum\limits_{(v,w) \in E(V(\Gamma))} Z_w.
    \]

\end{itemize}

\section{Classical Preprocessing and Reduced QAOA}
\label{sec:PreprocessingReducedQAOA}

Let $\mathbb{B}^n := \{0,1\}^n$ denote the set of all binary strings of length $n$, and let $\mathcal{S}_{2^n}$ denote the symmetric group acting on the $2^n$ elements of $\mathbb{B}^n$.
  A significant class of optimization problems entails finding  elements in $\mathbb{B}^n$ on which a given function 
\begin{equation}\label{eq:optfunc}
F: \mathbb{B}^n \rightarrow \mathbb{R}
\end{equation}
(i.e., the optimization objective with possible constraints) is either minimized or maximized.

If the function $F(x)$ is invariant with respect to a permutation $g \in \mathcal{S}_{2^n}$, i.e., $F(g^{-1}(x)) = F(x)$ for any $x \in \mathbb{B}^n$, then $g$ is called a \textit{symmetry} of $F$. Let   $G \subseteq \mathcal{S}_{2^n}$ be a subgroup of symmetries of $F$. Then the set $\mathbb{B}^n$ can be expressed as a disjoint union of $G$-orbits: 
$$\mathbb{B}^n = \bigsqcup_{j=1}^{m} \mathcal{O}_j.$$ 

Rather than working directly with the original combinatorial optimization problem on $\mathbb{B}^n$, we may instead pass to the quotient induced by the group action.
Specifically, consider the quotient set 
\begin{equation}
\label{eq:B_G}
     \mathbb{B}_G := \mathbb{B}^n / G,
\end{equation}
  
where two elements of $\mathbb{B}^n$ are equivalent if and only if they belong to the same $G$-orbit. 
Since the objective function $F$ is constant along $G$-orbits, 
we obtain a well-defined function 
\begin{equation}
    \label{eq: redfunction}
     \widetilde{F} : \mathbb{B}_G \to \mathbb{R}, \qquad 
   \widetilde{F}(\mathcal{O}_j) := F(x), 
\end{equation}
where $x\in \mathcal{O}_j$ is any representative of the equivalence class. 
This yields the reduced combinatorial optimization problem of minimizing $\widetilde{F}$ on $\mathbb{B}_G$. 

For clarity of exposition, we will henceforth assume that the action of the group $G$ on $\mathbb{B}^n$ is \emph{free}, that is, no non-identity element of $G$ fixes any point of $\mathbb{B}^n$.
In this case, all $G$-orbits have cardinality equal to $|G|$, 
and therefore the quotient has size
\[
   |\mathbb{B}_G| = \frac{2^n}{|G|}.
\]

\begin{rmk}
The solutions of the original problem, i.e., the minima (maxima) of $F$ on $\mathbb{B}^n$, 
are in $|G|$-to-$1$ correspondence with the solutions of the reduced problem, 
namely the minima (maxima)  of $\widetilde{F}$ on $\mathbb{B}_G$. 
Equivalently, the solutions of the original problem are precisely the $G$-orbits 
of the solutions of the reduced problem.
\end{rmk}

\begin{ex}
    Let \( \Gamma  \) be an undirected graph with vertex set \( V = \{1, 2, \ldots, n\} \). The \emph{MaxCut} problem asks for a partition of \(V\) into two disjoint
sets such that the number of edges crossing the partition is maximized.

We encode a cut by a binary string
\(
(x_1,\ldots,x_n)\in\mathbb{B}^n,
\)
where \(x_i=0\) or \(1\) indicates the side of the cut to which vertex \(i\) belongs. An edge \((i,j)\in E\) contributes \(1\) to the cut if and only if its endpoints lie on opposite sides with the corresponding objective function is
\begin{equation}
    \label{eq: maxcutobjfunction}
     F(x_1, \ldots, x_n) = \sum\limits_{(i,j)\in E} \left[(1 - x_i)x_j + x_i(1 - x_j)\right],
\end{equation}
which evaluates to \(1\) precisely when \(x_i \neq x_j\), and to \(0\) otherwise. Thus, \(F(x_1,\ldots,x_n)\) counts the total number of edges crossing the cut defined
by \((x_1,\ldots,x_n)\). The goal of the MaxCut problem is therefore to find a binary string
\((x_1,\ldots,x_n)\in\mathbb{B}^n\) that maximizes the objective function. The function \(F\) is invariant under the natural \(\mathbb{Z}_2\)-action that flips all bits simultaneously,
\begin{equation}
(x_1,\ldots,x_n)\longmapsto (1-x_1,\ldots,1-x_n),
    \label{eq:Z_2action}
\end{equation}
To eliminate this global symmetry, we fix one vertex
\(s\in\{1,\ldots,n\}\) and impose the condition \(x_s=0\).
This selects a unique representative from each \(\mathbb{Z}_2\)-orbit.
For any representative
\[
(x_1,\ldots,x_{s-1},0,x_{s+1},\ldots,x_n)\in\mathbb{B}_G,
\]
the reduced objective function becomes
\begin{equation}
    \label{eq:maxcutredobjfunction}
    \widetilde{F}(x_1, \ldots, x_{s-1}, 0, x_{s+1}, \ldots, x_n) 
    = \sum\limits_{\substack{(i,j)\in E \\ i,j \neq s}} \Big[(1 - x_i)x_j + x_i(1 - x_j)\Big] 
    \;+\; \sum\limits_{(s,j)\in E} x_j.
\end{equation}
\end{ex}

\subsection{QAOA for the Original Problem}
Variational Quantum Algorithms (VQAs) form a prominent class of hybrid quantum-classical algorithms designed to address problems in quantum chemistry, machine learning, and combinatorial optimization (see \cite{CABB} for a review). 
In a VQA, a parameterized quantum circuit is optimized by a classical algorithm to minimize (or maximize) a loss function, often chosen as the expectation value of a Hamiltonian operator. 

One of the most influential VQAs is the Quantum Approximate Optimization Algorithm (QAOA), first introduced in \cite{QAOA}. Below we provide an overview in a language tailored to the needs of this paper. For a more comprehensive treatment, the reader is referred to \cite{BBCC, QAOA, HWORVB}.

Let $W$ be a vector space of dimension $2^n$ with computational basis $\{\ket{x}\}_{x\in \mathbb{B}^n}$, indexed by elements of $\mathbb{B}^n$. A Hamiltonian $H_P$ is said to \emph{represent} a function $F:\mathbb{B}^n \to \mathbb{R}$ if
\begin{equation}
    \label{eq:problemHamiltonian}
    H_P \ket{x} = F(x) \ket{x} \quad \text{for all } x \in \mathbb{B}^n.
\end{equation}

 The quantum formulation of the optimization problem associated with \(F\) is obtained via the following classical-to-quantum correspondence:

\begin{itemize}
	\item Set $\mathbb{B}^n\rightsquigarrow$ Hilbert space $W=(\mathbb{C}^2)^{\otimes n}$;
	\item Objective function $F\rightsquigarrow$ Linear operator $H_P$ (problem Hamiltonian) acting on $W$;
	\item Minima of $F$ on $\mathbb{B}^n\rightsquigarrow$ Lowest energy states  of $H_P$ in $W$;
    \item Action of $G$ on $\mathbb{B}^n$ that preserves $F\rightsquigarrow$ Unitary representation of $G$ on $W$
    that commutes with $H_P$.
\end{itemize}

A central component of QAOA is the so-called \emph{mixer Hamiltonian}, denoted $H_M$. The negative of this operator, $-H_M$, is required to have a unique ground state $\ket{\xi} \in W$. 
The fundamental idea behind QAOA is to gradually deform 
$H_M$ into $H_P$ through a series of quantum transformations, in such a way that the ground state at each stage is mapped to a ground state of the next.

The process begins with the preparation of the ground state $\ket{\xi}$ of the Hamiltonian $-H_M$. 
Subsequently, one applies an alternating sequence of unitary operators generated by the problem Hamiltonian $H_P$ and the mixer Hamiltonian $H_M$. 
Each layer consists of the action of $e^{-i \gamma_j H_P}$ followed by $e^{-i \beta_j H_M}$, where the real parameters 
\[
\boldsymbol{\beta}=(\beta_1,\dots,\beta_p), \qquad \boldsymbol{\gamma}=(\gamma_1,\dots,\gamma_p),
\] 
with $\beta_i \in [0,2\pi)$ and $\gamma_j \in [0,\pi)$ are collections of classically optimized parameters that control the evolution times under $H_M$ and $H_P$, respectively. Here, $p\in \mathbb{N}$ denotes the QAOA depth, i.e., the number of alternating layers. 
The overall transformation operator is
\begin{equation}
U(\boldsymbol{\beta},\boldsymbol{\gamma})
  := e^{-i \beta_p H_M} e^{-i \gamma_p H_P} \cdots e^{-i \beta_1 H_M} e^{-i \gamma_1 H_P}.
\label{qaoa-chain}
\end{equation}
At the end of the protocol, the observable $H_P$ is measured
on the state
\[
\ket{\psi(\boldsymbol{\beta},\boldsymbol{\gamma})} 
   := U(\boldsymbol{\beta},\boldsymbol{\gamma}) \ket{\xi}.
\]
The expectation value of $H_P$ with respect to this state serves as the cost function used to guide the classical optimization of the parameters $(\boldsymbol{\beta}, \boldsymbol{\gamma})$. 
A schematic overview of the algorithm is given in Figure~\ref{QAOAcircuit}.

\begin{center}
\begin{figure}[h]
\begin{tikzpicture}
\node[draw, rounded corners, fill=blue!30, minimum width=6.5cm, minimum height=4cm, inner sep=0.5cm, anchor=center, label=below:{Ansatz circuit}] (blueframe) at (0.2,0) {};

\node[draw, rounded corners, fill=red!30, minimum width=1cm, minimum height=4cm, inner sep=0.5cm, anchor=center] (redframe) at (4.2,0) {};

\node (circ) {
\begin{quantikz} 
\lstick{\ket{0}} & \gategroup[4,steps=1,style={rounded corners,fill=black!20,draw}, label style={yshift=-0.75in}]{\scriptsize$U_\xi$} & & \gategroup[4,steps=2, style={rounded corners,fill=green!20,draw}, label style={yshift=-.75in}]{\scriptsize$U_P(\gamma_1)$} & & & \gategroup[4,steps=2,style={rounded corners,fill=yellow!20,draw}, label style={yshift=-0.75in}]{\scriptsize$U_M(\beta_1)$} & & \cdots&  \gategroup[4,steps=2, style={rounded corners,fill=green!20,draw}, label style={yshift=-0.75in} ]{\scriptsize$U_P(\gamma_p)$} & & & \gategroup[4,steps=2,style={rounded corners,fill=yellow!20,draw}, label style={yshift=-0.75in}]{\scriptsize$ U_M(\beta_p)$} &  & &\meter{} \\ 
 \lstick{\ket{0}} & & & & & & & & \cdots & & & & & & &\meter{}\\
 \lstick{\vdots} \\
 \lstick{\ket{0}} & & & & & & & &  \cdots & & & & & & & \meter{} \end{quantikz}
};
\draw[<-]  (-1.9,-2.2) -- (-1.9,-3);
\draw[-]  (-1.9,-3) -- (4.1,-3) node[midway,below,yshift=-0.1cm,black]{Update of parameters $(\boldsymbol{\beta}, \boldsymbol{\gamma})$};
\draw[-]  (4.1,-2.2) -- (4.1,-3);
\end{tikzpicture}
\caption{\textbf{Schematic illustration of the QAOA circuit.} The initial state preparation unitary $U_\xi$ with $U_{\xi}(\ket{0}^{\otimes n}) = \ket{\xi}$) is followed by $p$ alternating applications of unitaries $U_P(\gamma_j):=e^{-i \gamma_j H_P}$ and $U_M(\beta_j):=e^{-i \beta_j H_M}$, generated by
the problem Hamiltonian $H_P$ and the mixer Hamiltonian $H_M$, respectively. At the end of the circuit, the state is measured in the computational basis. 
Each measurement outcome $x \in \mathbb{B}^n$ is assigned the value $F(x)$ of the objective function, 
and the empirical mean of these values provides an estimate of 
$\bra{\psi(\boldsymbol{\beta},\boldsymbol{\gamma})} H_P \ket{\psi(\boldsymbol{\beta},\boldsymbol{\gamma})}$. 
This estimate is then used in a classical optimization loop to update the parameters $(\boldsymbol{\beta}, \boldsymbol{\gamma})$ with the goal of minimizing the empirical mean.}
\label{QAOAcircuit}
\end{figure}
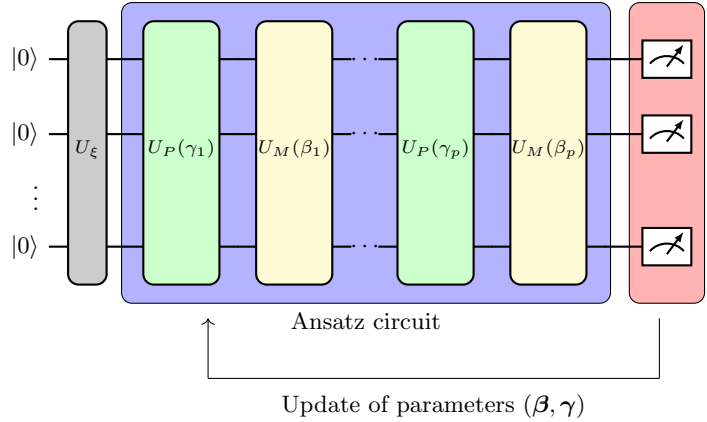
\end{center}

\subsection{Reduced QAOA: Generalities}

The action of $G$ on the set $\mathbb{B}^n$ extends to a linear action on the Hilbert space $W$. Therefore, $W$ is a \textit{representation} of $G$ and can be expressed as a direct sum of isotypic components with respect to this action:
\begin{equation}
\label{generalDecomp}
    W = \bigoplus_i W_i,
\end{equation}
where the index $i$ enumerates nonisomorphic irreducible representations  of $G$ appearing in $W$. In particular, \begin{equation}
    \label{eq:G-invariants}
    W^G:=\{w\in W \mid g\cdot w=w \ \forall g\in G\}
\end{equation}
will denote the trivial subrepresentation (the subspace of $G$-invariants). The actions of $G$ and $H_P$ on $W$ commute: 
\begin{equation}
    \label{eq:problemHamCommutesWithG}
    H_P(g(w)) = g(H_P(w)) \quad \forall w \in W, \quad \forall g \in G.
\end{equation}

\begin{rmk}
  While the dimension of $W^G$ is equal to the number of elements in $\mathbb{B}_G$ (see \eqref{eq:B_G}) there might not be a natural (standard) basis of $W^G$ encoded with the elements for $\mathbb{B}_G$. 

  However, in certain cases (as described in the next section), we can systematically select a representative $x_j \in \mathbb{B}^n$ for each $G$-orbit $\mathcal{O}_j \in \mathbb{B}_G$, allowing us to construct a natural reduced Hilbert space  with the basis $\{\ket{x_j}\}_{\mathcal{O}_j\in \mathbb{B}_G}$.
\end{rmk}

\subsection{Reduced QAOA for Simultaneous Bit Flip Symmetry}
In many optimization problems on the domain of $n$ classical bits, the objective function is invariant under a simultaneous flip of all bits. Exploiting this symmetry, one may consider a \emph{reduced problem} on bit strings of length $n-1$, obtained by fixing the value of one chosen bit to be $0$ or $1$. Each solution to the reduced problem then corresponds to two solutions of the original one (the string itself and its flipped version), so that all solutions to the original problem can be recovered in this way.

A convenient representative of each orbit is chosen by fixing the value of a particular bit $b_j$ with $1 \le j \le n$. 
Recall that the full Hilbert space is
\(
W = (\mathbb{C}^2)^{\otimes n},
\)
with computational basis vectors indexed by bit strings $x \in \{0,1\}^n$.  Fixing the value of the $j^{\text{th}}$ bit to either $0$ or $1$ restricts us to those basis states satisfying $x_j = b_j$.  
Equivalently, the reduced Hilbert space
\[
W_j = (\mathbb{C}^2)^{\otimes (n-1)}
\]
is obtained by contracting the $j^{\text{th}}$ tensor factor of $W$ into the one-dimensional coordinate subspace corresponding to the chosen value, $0$ or $1$. 
Without loss of generality, we assume the fixed bit is $0$.

The reduced problem Hamiltonian $H_{P}^j$ is obtained directly from the original problem Hamiltonian by replacing the Pauli operator $Z_j$ corresponding to the fixed bit with the identity. 
This yields a reduced QAOA defined on $W_j$ with problem Hamiltonian $H_{P}^j$.

\section{From Dynamical to Reduced Dynamical Lie Algebras}
\label{sec:StandardReducedDLA}

Recall that the \emph{standard dynamical Lie algebra} (DLA) associated with a given QAOA instance is defined as the real Lie   algebra generated, under the commutator operation, by the skew-Hermitian operators $iH_M$ and $iH_P$, where $H_M$ and $H_P$ denote the mixer and problem Hamiltonians, respectively. In other words, 
\begin{equation}
    \label{eq: stdDLAdefn}
    \mathfrak{g}_{\mathrm{std}} := \langle\, iH_M,\, iH_P \,\rangle_{\mathrm{Lie}}.
\end{equation}

For comparison, we also consider the corresponding \emph{free dynamical Lie algebra}, defined as the real Lie algebra generated by the individual local terms appearing in the mixer and problem Hamiltonians, $iH_M$ and $iH_P$. Writing
\[
H_M = \sum\limits_{j} M_j,
\qquad 
H_P = \sum\limits_{k} P_k,
\]
as decompositions into local Hermitian operators, the free DLA is given by

\begin{equation}
    \label{eq: freeDLAdefn}
    \mathfrak{g}_{\mathrm{free}} 
:= \big\langle\, \{\, i M_j \,\},\; \{\, i P_k \,\}\big\rangle_{\mathrm{Lie}}.
\end{equation}

We now turn to the central objects of this section: the \emph{reduced dynamical Lie algebras}.  
Given a combinatorial optimization problem (COP) on $n$ binary variables $b_0,\dots,b_{n-1}$, we consider the family of reduced problems obtained by fixing one variable at a time. For each $\ell\in\{0,\dots,n-1\}$, setting $b_\ell=0$ (or equivalently $b_\ell=1$, up to relabeling) yields a reduced COP on $n-1$ variables, which in turn defines a reduced QAOA instance with corresponding mixer and problem Hamiltonians, denoted $H_M^{\ell}$ and $H_P^{\ell}$.

\begin{defn}
The \emph{standard reduced dynamical Lie algebra} associated with the $\ell^{\text{th}}$ reduced problem is defined as
\[
\mathfrak{g}^{\ell}_{\mathrm{std}}
:= \big\langle\, iH_M^{\ell},\, iH_P^{\ell} \big\rangle_{\mathrm{Lie}}.
\]
\end{defn}

Figure~\ref{fig:reduction-schematic} summarizes these relationships schematically. The original COP gives rise to a single QAOA instance with standard DLA $\mathfrak{g}_{\mathrm{std}}$, while each reduced COP produces a corresponding reduced QAOA instance and standard reduced DLA. The arrows indicate the passage from optimization problems to QAOA Hamiltonians and, subsequently, to their associated dynamical Lie algebras.

\begin{center}
    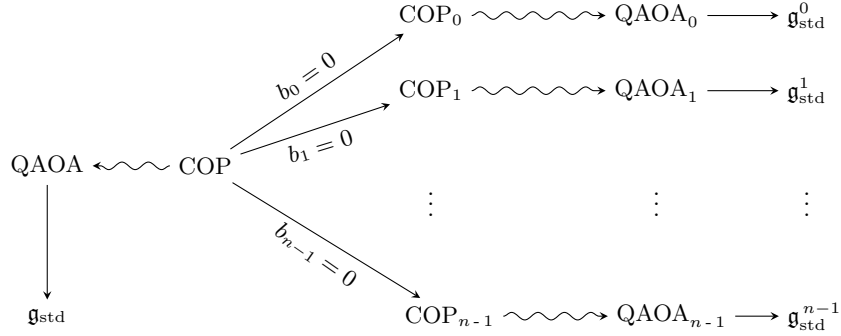
\begin{figure}[h]
        \begin{tikzpicture}[scale=1, >=stealth]

\node (COP0)  at (5,4) {$\mathrm{COP}_0$};
\node (QAOA0) at (8,4) {$\mathrm{QAOA}_0$};
\node (g0)    at (10,4) {$\mathfrak{g}^{0}_{\mathrm{std}}$};

\node (COP1)  at (5,3) {$\mathrm{COP}_1$};
\node (QAOA1) at (8,3) {$\mathrm{QAOA}_1$};
\node (g1)    at (10,3) {$\mathfrak{g}^{1}_{\mathrm{std}}$};

\node (QAOA) at (-0.1,2) {$\mathrm{QAOA}$};
\node (COP)  at (2,2) {$\mathrm{COP}$};
\node        at (10,1.6) {$\vdots$};
\node        at (5,1.6) {$\vdots$};
\node        at (8,1.6) {$\vdots$};

\node (g)        at (-0.1,0) {$\mathfrak{g}_{\mathrm{std}}$};
\node (COPn)     at (5.25,0) {$\mathrm{COP}_{n\operatorname{-}1}$};
\node (QAOAn)    at (8.2,0) {$\mathrm{QAOA}_{n\operatorname{-}1}$};
\node (gn)       at (10.1,0) {$\mathfrak{g}^{\,n-1}_{\mathrm{std}}$};

\draw[->, decorate, decoration={snake, amplitude=0.5mm}]
    (COP0) -- (QAOA0);
\draw[->] (QAOA0) -- (g0);

\draw[->, decorate, decoration={snake, amplitude=0.5mm}]
    (COP1) -- (QAOA1);
\draw[->] (QAOA1) -- (g1);

\draw[->,decorate, decoration={snake, amplitude=0.5mm}]
    (COPn) -- (QAOAn);
\draw[->] (QAOAn) -- (gn);

\draw[->] (QAOA) -- (g);

\draw[->, decorate, decoration={snake, amplitude=0.5mm}]
    (COP) -- (QAOA);

\draw[->] (COP) -- node[rotate=35,above] {$b_0=0$} (COP0);
\draw[->] (COP) -- node[rotate=17,below] {$b_1=0$} (COP1);
\draw[->] (COP) -- node[rotate=-35,below] {$b_{n-1}=0$} (COPn);

\end{tikzpicture}

    \caption{\textbf{Reduction scheme for dynamical Lie algebras.}
An $n$-bit combinatorial optimization problem gives rise to a QAOA instance with standard dynamical Lie algebra $\mathfrak{g}_{\mathrm{std}}$. Fixing a variable $b_\ell$ produces a reduced optimization problem $\mathrm{COP}_\ell$, a corresponding reduced QAOA instance, and an associated standard reduced DLA $\mathfrak{g}^{\ell}_{\mathrm{std}}$.}
\label{fig:reduction-schematic}
    \end{figure}
\end{center}

\begin{ex}
Let $\Gamma = (V,E)$ be a finite graph, and consider the MaxCut problem defined on $\Gamma$. 
The corresponding (standard) mixer and problem Hamiltonians are
\begin{equation}
\label{eq:OriginalHamiltonians}
    H_M = \sum\limits_{w \in V} X_w, 
    \qquad
    H_P = \sum\limits_{(v,w) \in E} Z_v Z_w.
\end{equation}

Then the \emph{free} dynamical Lie algebra is
\begin{equation}\label{eq:FreeDLAgenerators}
\mathfrak{g}_{\Gamma,\mathrm{free}}
= \Big\langle 
    \{ i X_w \mid w \in  V\}, \;
    \{ i Z_w Z_{w'} \mid (w w') \in E \} 
\Big\rangle_{\mathrm{Lie}},
\end{equation}
 while the  \emph{standard} dynamical Lie algebra is
\begin{equation}\label{eq:StandardDLAgenerators}
\mathfrak{g}_{\Gamma,\mathrm{std}}
= \Big\langle  X, Z\Big\rangle_{\mathrm{Lie}},
\end{equation}
with $X:=iH_M$ and $Z:=iH_P$.
If the bit associated with a vertex $v \in V$ is fixed, the reduced mixer and problem Hamiltonians become
\begin{equation}
\label{eq:ReducedHamiltonians}
    H_M^v = \sum\limits_{w \in V \setminus \{v\}} X_w, 
    \qquad
    H_P^v = \sum\limits_{\substack{(w,w') \in E \\ w,w' \neq v}} Z_w Z_{w'} 
    + \sum\limits_{(v,w) \in E} Z_w.
\end{equation}

Consequently, the reduced \emph{free} dynamical Lie algebra is 
\begin{equation}\label{eq:ReducedFreeDLAgenerators}
\mathfrak{g}^{\,v}_{\Gamma,\mathrm{free}}
= \Big\langle 
    \{ i X_w \mid w \in V \setminus \{v\} \}, \;
    \{ i Z_w Z_w' \mid (w w') \in E, \, w,w' \neq v \}, \;
    \{ i Z_w \mid (v w) \in E \} 
\Big\rangle_{\mathrm{Lie}},
\end{equation}
while the reduced \emph{standard} dynamical Lie algebra is the algebra
\begin{equation}\label{eq:ReducedStandardDLAgenerators}
\mathfrak{g}^{\,v}_{\Gamma,\mathrm{std}}
= \Big\langle 
    \mathcal{X}_{\widehat{v}}, \; \mathcal{Z}_{\widehat{v}}
\Big\rangle_{\mathrm{Lie}},
\end{equation}
with generators
\begin{equation}
\mathcal{X}_{\widehat{v}} := iH_M^v 
\text{ and }
\mathcal{Z}_{\widehat{v}} := iH_P^v.    
\end{equation}
\end{ex}

\section{Distinguished Elements in  Reduced DLAs}
\label{sec:DistinguishedElts}

In this section, we establish that certain distinguished elements are contained in  reduced dynamical Lie algebras. These results will be used in subsequent sections to analyze inclusions among DLAs and to derive bounds on their dimensions.


We begin by introducing notation.

\begin{defn}
\label{defn:vertexNSubsets}
Let \( v \) be a vertex in a graph \( \Gamma \), and let \( k \ge 1 \) be an integer. We define

\begin{equation}
\label{eq:DistkEvenOddVertices}
\begin{aligned}
    \mathcal{N}_{v,k}
&:= \{\, w \in V(\Gamma) : \text{the distance (shortest path) between } v \text{ and } w \text{ equals } k \,\},\\
\mathcal{N}_{v,k,\mathrm{even}}
&:= \{ v \in \mathcal{N}_{v,k} \mid \deg(v) \equiv 0 \pmod{2} \},
\\
\mathcal{N}_{v,k,\mathrm{odd}}
&:= \{ v \in \mathcal{N}_{v,k} \mid \deg(v) \equiv 1 \pmod{2} \}.
\end{aligned}
\end{equation}
That is, $\mathcal{N}_{v,k}$ is the set of vertices at graph distance $k$ from $v$, and $\mathcal{N}_{v,k,\mathrm{even}}$ and $\mathcal{N}_{v,k,\mathrm{odd}}$ denote the subsets consisting of vertices of even and odd degree, respectively.
\end{defn}

For each $k \in \mathbb{Z}_{>0}$ and $v \in V(\Gamma)$, we define
\begin{equation}
\label{eq:FixedDistX1element}
\begin{aligned}
\mathcal{X}_{v,k}
&:= i \sum_{w \in \mathcal{N}_{v,k}} X_w, \\
\mathcal{X}_{v,k,\mathrm{even}}
&:= i \sum_{w \in \mathcal{N}_{v,k,\mathrm{even}}} X_w, \\
\mathcal{X}_{v,k,\mathrm{odd}}
&:= i \sum_{w \in \mathcal{N}_{v,k,\mathrm{odd}}} X_w.
\end{aligned}
\end{equation}

These operators represent the sums of Pauli-$X$ operators supported on vertices at graph distance $k$ from $v$, optionally restricted to vertices of even or odd degree.

\begin{thm}
\label{thm:DistDLAElements}
Let $\Gamma = (V,E)$ be a connected graph and let $v \in V(\Gamma)$ be a vertex. 
Assume that for every $j \ge 1$ and every vertex $w \in \mathcal{N}_{v,j}$, the number of neighbors of $w$ in the preceding layer $\mathcal{N}_{v,j-1}$ is odd, i.e.,
\begin{equation}
\label{eq:ParityAssumption}
|\{\, u \in \mathcal{N}_{v,j-1} \mid (u,w) \in E \,\}| \equiv 1 \pmod{2}.
\end{equation}

Then the standard reduced dynamical Lie algebra $\mathfrak{g}^{\,v}_{\Gamma,\mathrm{std}}$ contains the elements $\mathcal{X}_{\widehat{v},k}, \mathcal{X}_{v,k,\mathrm{even}}$ and $\mathcal{X}_{v,k,\mathrm{odd}}$ defined in \eqref{eq:FixedDistX1element} for all $k \in \mathbb{Z}_{>0}$.
\end{thm}

We now introduce a graph-theoretic reduction that plays a central role in our analysis of dynamical Lie algebras.

\begin{defn}
\label{def:ReducedGraph}
Let $\Gamma = (V, E)$ be a graph and let $v \in V(\Gamma)$ be a vertex.  
The \emph{reduced graph at $v$}, denoted by $\Gamma_v$, is the graph obtained from $\Gamma$ by removing the vertex $v$ together with all edges incident to $v$. Equivalently, $\Gamma_v$ is the graph with vertex and edge sets
\begin{equation}
\label{eq:redgraph}
V(\Gamma_v) = V(\Gamma) \setminus \{v\}, 
\qquad 
E(\Gamma_v) = E(\Gamma) \setminus \{ (v,j) \mid j \in \mathcal{N}_{v,1} \}.
\end{equation}
\end{defn}

The following result provides a powerful structural insight. 
It shows that, by extending a given graph in a controlled manner, 
the free dynamical Lie algebra associated with the reduction of the extended graph admits a natural embedding into the standard reduced dynamical Lie algebra of a suitably enlarged graph.  This result is essentially a consequence of Theorem~\ref{thm:DistDLAElements}; a detailed construction and verification are deferred to Appendix~\ref{sec:Technicalities}.

\begin{thm}
\label{thm:GraphExtensionEmbedding}
Let $\Gamma$ be a connected graph on $n$ vertices and let $v \in V(\Gamma)$ be any vertex. 
Denote by
\begin{equation}
    \label{eq:MaxDistToV}
    j(v) := \max_{w \in V(\Gamma)} \operatorname{dist}_{\Gamma}(v,w)
\end{equation}
the eccentricity of $v$ in $\Gamma$.  
Then there exists a graph $\widehat{\Gamma}_v$ satisfying the following properties: 
\begin{itemize}
    \item \emph{Containment.} The graph $\Gamma$ is a subgraph of $\widehat{\Gamma}_v$.
\item \emph{Vertex bound.} The number of vertices satisfies

\begin{equation}
    \label{eq:NumOfVerticesInExtGraph}
    |V(\widehat{\Gamma}_v)|< 1+3(j(v)-1)(n-1) + \frac{(3n-2)(3n-3)}{2}.
\end{equation}

\item \emph{Edge bound.} The number of edges satisfies

\begin{equation}
    \label{eq:NumOfEdgesInExtGraph}
|E(\widehat{\Gamma}_v)|< |E(\Gamma)|+4(n-1)+3(j(v)-1)(n-1) + \frac{(3n-2)(3n-3)}{2}.
\end{equation}

\item \emph{DLA embedding.} The standard reduced dynamical Lie algebra of the extension contains the free dynamical Lie algebra of the reduced extended graph, namely
\[
\mathfrak{g}^{\,v}_{\widehat{\Gamma}_v,\mathrm{std}}
\;\supseteq\;
\mathfrak{g}_{\widehat{\Gamma}_v,\mathrm{free}},
\qquad \text{cf.\ \eqref{eq:ERContainment}.}
\]

\item \emph{Optimization equivalence.} Optimal solutions of the MaxCut problem on $\Gamma$ are obtained from optimal solutions on $\widehat{\Gamma}_v$ by discarding the bits corresponding to the auxiliary vertices $V(\widehat{\Gamma}_v)\setminus V(\Gamma)$.
\end{itemize}
\end{thm}

\begin{rmk}
Since $j(v) \le n-1$ for any connected graph $\Gamma$ on $n$ vertices,
Theorem~\ref{thm:GraphExtensionEmbedding} implies the  bounds
\[
|V(\widehat{\Gamma}_v)| = O(n^2),
\qquad
|E(\widehat{\Gamma}_v)| = O(n^2).
\]

In particular, the number of edges in the extended graph $\widehat{\Gamma}_v$ is at most quadratic in the number of vertices of the original graph $\Gamma$.
Thus, the graph extension required to realize the embedding
\[
\mathfrak{g}_{\widehat{\Gamma}_v,\mathrm{free}}
\;\subseteq\;
\mathfrak{g}^{\,v}_{\widehat{\Gamma}_v,\mathrm{std}}
\]
incurs at most a quadratic overhead in system size.
\end{rmk}




\subsection{Construction of the Extended Graph}

We now describe the construction of the extended graph $\widehat{\Gamma}_v$, which proceeds in three steps (labeled $0–2$).

Step~0 is a preprocessing stage ensuring that all vertices satisfy the parity condition in~\eqref{eq:ParityAssumption}. Steps~1 and~2 then implement the extension procedure, which guarantees that the reduced standard DLA coincides with the free one.

Although the verification of these properties is deferred to Appendix~\ref{sec:Technicalities}, we present the construction here due to its practical relevance.

 \medskip
\noindent\textbf{Step 0. Correcting parity.} For each $j \ge 2$ and every vertex $w \in \mathcal{N}_{v,j}$ that violates the parity assumption \eqref{eq:ParityAssumption}, we select a vertex $w_{j-2} \in \mathcal{N}_{v,j-2}$ such that $\operatorname{dist}(w_{j-2}, w) = 2$. We then extend the graph by attaching the triangular segment shown in Figure~\ref{fig:ExtGamma}.
\begin{figure}[ht]
\centering
\begin{tikzpicture}[scale=0.8]

\node[circle, draw, fill=black, inner sep=2pt] (r8c10) at (8.0,2.4) {}; \node at (8.0,2.8) {$w_{j-2}$};
\node[circle, draw, fill=blue, inner sep=2pt] (r8c12) at (9.6,2.4) {};
\node at (9.6,2.8) {\color{blue}{$w_{j-1}$}};
\node[circle, draw, fill=black, inner sep=2pt] (r8c14) at (11.2,2.4) {};
\node at (11.2,2.8) {$w$};
\node[circle, draw, fill=blue, inner sep=2pt] (r10c12) at (9.6,0.8) {};
\node at (9.6,0.4) {\color{blue}{$w'_{j-1}$}};

\draw[blue] (r8c10) -- (r8c12);
\draw[blue] (r8c10) -- (r10c12);
\draw[blue] (r8c12) -- (r10c12);
\draw[blue] (r8c12) -- (r8c14);

\end{tikzpicture}
\caption{Local modification for parity correction. Vertices and edges belonging to the original graph $\Gamma$ are shown in black, while the auxiliary elements introduced during Step $0$ are highlighted in blue.}
\label{fig:ExtGamma}
\end{figure}

\medskip
\noindent\textbf{Step 1. Equalizing distances.}
For each vertex $w \in V(\Gamma_v)$, attach a path of length 
$j(v)-\operatorname{dist}_\Gamma(v,w)$ to $w$. 
Vertices already at maximal distance $j(v)$ from $v$ remain unchanged. 
Denote the resulting graph by $\Gamma'$.

By construction, every original vertex $w\neq v$ now determines a unique vertex in $\Gamma'$ at distance $j(v)$ from $v$, obtained by following the newly attached path from $w$ to its endpoint. 
Hence there is a natural injective map
\[
\psi : V(\Gamma_v)
\hookrightarrow
V(\Gamma') \cap \mathcal{N}'_{v,j(v)},
\]
where $\mathcal{N}'_{v,j(v)}$ denotes the set of vertices in $\Gamma'$ at distance $j(v)$ from $v$.

\medskip

\noindent\textbf{Step 2. Layered attachment.}
Fix an indexing of the vertices in $\mathcal{N}'_{v,j(v)}$ by 
\[
\{0,1,\dots,|\mathcal{N}'_{v,j(v)}|-1\}.
\]
For each vertex with index $j$, attach an additional path of length $j$. 
Denote the resulting graph by $\widehat{\Gamma}_v$.

The graph $\widehat{\Gamma}_v$ satisfies all properties stated in Theorem~\ref{thm:GraphExtensionEmbedding}; see Appendix~\ref{sec:Technicalities} for verification.

An illustrative example of this procedure is shown in Figures~\ref{fig:graphGamma}, \ref{fig:GammaPrime}, and \ref{fig:GammaHat}, 
which display the original graph, the intermediate extension after Step~1, and the final graph $\widehat{\Gamma}_v$, respectively.

\begin{figure}[ht]
\centering
\begin{tikzpicture}[scale=0.8]

\node[circle, draw, fill=black, inner sep=2pt] (r4c6) at (4.8,5.6) {};
\node[circle, draw, fill=black, inner sep=2pt] (r4c14) at (11.2,5.6) {};

\node[circle, draw, fill=black, inner sep=2pt] (r6c8) at (6.4,4.0) {};
\node[circle, draw, fill=black, inner sep=2pt] (r6c12) at (9.6,4.0) {};

\node[circle, draw, fill=black, inner sep=2pt] (r8c8) at (6.4,2.4) {};
\node[circle, draw, fill=red, inner sep=2pt] (r8c10) at (8.0,2.4) {}; \node at (8.0,2) {$v$};
\node[circle, draw, fill=black, inner sep=2pt] (r8c12) at (9.6,2.4) {};

\node[circle, draw, fill=black, inner sep=2pt] (r10c8) at (6.4,0.8) {};
\node[circle, draw, fill=black, inner sep=2pt] (r10c12) at (9.6,0.8) {};

\draw (r6c8) -- (r4c6);
\draw (r6c12) -- (r4c14);
\draw (r8c8) -- (r10c8);
\draw (r8c10) -- (r6c8);
\draw (r8c10) -- (r6c12);
\draw (r8c10) -- (r8c8);
\draw (r8c10) -- (r8c12);
\draw (r8c10) -- (r10c8);
\draw (r8c10) -- (r10c12);
\draw (r8c12) -- (r10c12);
\draw (r10c8) -- (r10c12);

\end{tikzpicture}
\caption{A sample graph $\Gamma$ with a distinguished vertex $v$.}
\label{fig:graphGamma}
\end{figure}

\begin{figure}[ht]
\centering
\begin{tikzpicture}[scale=0.8]

\node[circle, draw, fill=black, inner sep=2pt] (r4c6) at (4.8,5.6) {};
\node[circle, draw, fill=black, inner sep=2pt] (r4c14) at (11.2,5.6) {};

\node[circle, draw, fill=black, inner sep=2pt] (r6c8) at (6.4,4.0) {};
\node[circle, draw, fill=black, inner sep=2pt] (r6c12) at (9.6,4.0) {};

\node[circle, draw, fill=blue, inner sep=2pt] (r8c6) at (4.8,2.4) {};
\node[circle, draw, fill=black, inner sep=2pt] (r8c8) at (6.4,2.4) {};
\node[circle, draw, fill=red, inner sep=2pt] (r8c10) at (8.0,2.4) {}; \node at (8.0,2) {$v$};
\node[circle, draw, fill=black, inner sep=2pt] (r8c12) at (9.6,2.4) {};
\node[circle, draw, fill=blue, inner sep=2pt] (r8c14) at (11.2,2.4) {};

\node[circle, draw, fill=blue, inner sep=2pt] (r10c6) at (4.8,0.8) {};
\node[circle, draw, fill=black, inner sep=2pt] (r10c8) at (6.4,0.8) {};
\node[circle, draw, fill=black, inner sep=2pt] (r10c12) at (9.6,0.8) {};
\node[circle, draw, fill=blue, inner sep=2pt] (r10c14) at (11.2,0.8) {};

\node at (11.2,0.4) {$0$};
\node at (11.2,2) {$1$};
\node at (11.2,5.2) {$2$};

\node at (4.8,0.4) {$5$};
\node at (4.8,2) {$4$};
\node at (4.8,5.2) {$3$};

\draw (r6c8) -- (r4c6);
\draw (r6c12) -- (r4c14);
\draw (r8c8) -- (r8c6);
\draw (r8c8) -- (r10c8);
\draw (r8c10) -- (r6c8);
\draw (r8c10) -- (r6c12);
\draw (r8c10) -- (r8c8);
\draw (r8c10) -- (r8c12);
\draw (r8c10) -- (r10c8);
\draw (r8c10) -- (r10c12);
\draw (r8c12) -- (r8c14);
\draw (r8c12) -- (r10c12);
\draw (r10c8) -- (r10c6);
\draw (r10c8) -- (r10c12);
\draw (r10c12) -- (r10c14);

\end{tikzpicture}
\caption{The extended graph $\Gamma'$ obtained from $\Gamma$ after Step~1.  Vertices added in this step are highlighted in blue. 
All vertices at maximal graph distance $j(v)=2$ from the distinguished vertex $v$ are labeled, preparing the indexing used in Step~2.}
\label{fig:GammaPrime}
\end{figure}

\begin{figure}[ht]
\centering
\begin{tikzpicture}[scale=0.8]

\node[circle, draw, fill=black, inner sep=2pt] (r1c3) at (2.4,8.0) {};

\node[circle, draw, fill=black, inner sep=2pt] (r2c4) at (3.2,7.2) {};
\node[circle, draw, fill=black, inner sep=2pt] (r2c16) at (12.8,7.2) {};

\node[circle, draw, fill=black, inner sep=2pt] (r3c5) at (4.0,6.4) {};
\node[circle, draw, fill=black, inner sep=2pt] (r3c15) at (12.0,6.4) {};

\node[circle, draw, fill=black, inner sep=2pt] (r4c2) at (1.6,5.6) {};
\node[circle, draw, fill=black, inner sep=2pt] (r4c6) at (4.8,5.6) {};
\node[circle, draw, fill=black, inner sep=2pt] (r4c14) at (11.2,5.6) {};

\node[circle, draw, fill=black, inner sep=2pt] (r5c1) at (0.8,4.8) {};
\node[circle, draw, fill=black, inner sep=2pt] (r5c3) at (2.4,4.8) {};

\node[circle, draw, fill=black, inner sep=2pt] (r6c2) at (1.6,4.0) {};
\node[circle, draw, fill=black, inner sep=2pt] (r6c4) at (3.2,4.0) {};
\node[circle, draw, fill=black, inner sep=2pt] (r6c8) at (6.4,4.0) {};
\node[circle, draw, fill=black, inner sep=2pt] (r6c12) at (9.6,4.0) {};

\node[circle, draw, fill=black, inner sep=2pt] (r7c3) at (2.4,3.2) {};
\node[circle, draw, fill=black, inner sep=2pt] (r7c5) at (4.0,3.2) {};
\node[circle, draw, fill=black, inner sep=2pt] (r7c15) at (12.0,3.2) {};

\node[circle, draw, fill=black, inner sep=2pt] (r8c4) at (3.2,2.4) {};
\node[circle, draw, fill=blue, inner sep=2pt] (r8c6) at (4.8,2.4) {};
\node[circle, draw, fill=black, inner sep=2pt] (r8c8) at (6.4,2.4) {};
\node[circle, draw, fill=red, inner sep=2pt] (r8c10) at (8.0,2.4) {};
\node at (8.0,2) {$v$};
\node[circle, draw, fill=black, inner sep=2pt] (r8c12) at (9.6,2.4) {};
\node[circle, draw, fill=blue, inner sep=2pt] (r8c14) at (11.2,2.4) {};

\node[circle, draw, fill=black, inner sep=2pt] (r9c5) at (4.0,1.6) {};

\node[circle, draw, fill=blue, inner sep=2pt] (r10c6) at (4.8,0.8) {};
\node[circle, draw, fill=black, inner sep=2pt] (r10c8) at (6.4,0.8) {};
\node[circle, draw, fill=black, inner sep=2pt] (r10c12) at (9.6,0.8) {};
\node[circle, draw, fill=blue, inner sep=2pt] (r10c14) at (11.2,0.8) {};

\node at (11.2,0.4) {$0$};
\node at (11.2,2) {$1$};
\node at (11.2,5.2) {$2$};

\node at (4.8,0.4) {$5$};
\node at (4.8,2) {$4$};
\node at (4.8,5.2) {$3$};

\draw (r2c4) -- (r1c3);
\draw (r3c5) -- (r2c4);
\draw (r3c15) -- (r2c16);
\draw (r4c6) -- (r3c5);
\draw (r4c14) -- (r3c15);
\draw (r5c3) -- (r4c2);
\draw (r6c2) -- (r5c1);
\draw (r6c4) -- (r5c3);
\draw (r6c8) -- (r4c6);
\draw (r6c12) -- (r4c14);
\draw (r7c3) -- (r6c2);
\draw (r7c5) -- (r6c4);
\draw (r8c4) -- (r7c3);
\draw (r8c6) -- (r7c5);
\draw (r8c8) -- (r8c6);
\draw (r8c8) -- (r10c8);
\draw (r8c10) -- (r6c8);
\draw (r8c10) -- (r6c12);
\draw (r8c10) -- (r8c8);
\draw (r8c10) -- (r8c12);
\draw (r8c10) -- (r10c8);
\draw (r8c10) -- (r10c12);
\draw (r8c12) -- (r8c14);
\draw (r8c12) -- (r10c12);
\draw (r8c14) -- (r7c15);
\draw (r9c5) -- (r8c4);
\draw (r10c6) -- (r9c5);
\draw (r10c8) -- (r10c6);
\draw (r10c8) -- (r10c12);
\draw (r10c12) -- (r10c14);

\end{tikzpicture}
\caption{The completed graph $\widehat{\Gamma}_v$ obtained after Step~2. 
To each vertex at maximal distance $j(v)=2$ from the distinguished vertex $v$,  an additional path is attached whose length equals its assigned label.}
\label{fig:GammaHat}
\end{figure}

\begin{rmk}
MaxCut solutions on $\Gamma$ are recovered from solution strings on $\widehat{\Gamma}_v$ by deleting the bits corresponding to the auxiliary vertices introduced during the construction.
\end{rmk}

\subsection{Parity-Degree Separation Sequences}

We introduce a refinement of distance layers based on parity patterns of vertex degrees along shortest paths to the fixed node. These subsets will allow us to identify structured elements in the reduced dynamical Lie algebras.

\begin{defn}
\label{defn:vertexSubsets}
Let \( v \) be a vertex in a graph \( \Gamma \) and let $\mathbf{a}:=(a_1,a_2,\dots,a_j)\in\mathbb{Z}_{2}^j$ be a binary sequence. We denote by
\(
\mathcal{C}^v_{j,\mathbf{a}} \subseteq \mathcal{N}_{v,j}
\)
the subset of vertices such that for each vertex $w \in \mathcal{C}^v_{j,\mathbf{a}}$, there exists a path 
\[
w_0 = v - w_1 - \cdots - w_j = w
\]
such that $\deg(w_s)\equiv a_s \pmod{2}$ for each $1\le s\le j$. 
\end{defn}
\begin{thm}
\label{thm:CoolDLAElements}
Let $\Gamma$ be a connected graph and let $v \in V(\Gamma)$ be a vertex. Assume the following:
\begin{enumerate}
\item[(i)] For every $j \ge 1$ and every $w \in \mathcal{N}_{v,j}$,
\[
|\{\, u \in \mathcal{N}_{v,j-1} : (u,w) \in E \,\}| \equiv 1 \pmod{2}.
\]

\item[(ii)] For every $j \ge 1$, every $\mathbf{a} \in \mathbb{Z}_2^j$, and every 
$w \in \mathcal{C}^v_{j,\mathbf{a}}$,
\[
|\{\, u \in \mathcal{C}^v_{j-1,\mathbf{a}'} : (u,w) \in E \,\}| \equiv 1 \pmod{2},
\]
where $\mathbf{a}' = (a_1,\dots,a_{j-1})$.
\end{enumerate} 

Then the standard reduced dynamical Lie algebra $\mathfrak{g}^{\,v}_{\Gamma,\mathrm{std}}$ contains the elements
\begin{equation}
\label{eq:ReducedDLAElts}
X_{\mathcal{C}^v_{j,\mathbf{a}}} := i \sum\limits_{w \in \mathcal{C}^v_{j,\mathbf{a}}} X_w
\end{equation}
for every subset $\mathcal{C}^v_{j,\mathbf{a}}$.
\end{thm}

\begin{ex}
\label{ex:ExpFamilyTrees}
Let $m\ge 3$ be an integer. Define the acyclic graph $\mathcal{T}_m$ as follows.
Start with the path graph $P_m$ on vertices $\{1,2,\dots,m\}$. Attach a copy of the path graph $P_2$ to each vertex $j=1,\dots,m-3$, and attach a single additional edge to the vertex $m-2$.  

Next, let $v-\mathcal{T}_m$ denote the graph obtained from $\mathcal{T}_m$ by adjoining an auxiliary vertex $v$ and connecting it to vertex $1$ of $\mathcal{T}_m$. We view $v-\mathcal{T}_m$ as a rooted tree with root $v$. By construction, \( v \) has degree one, and removing the root \( v \) (together with its incident edge) recovers the original tree \( \mathcal{T}_m \).

The neighborhood layers $\mathcal{N}_{v,j}$ in $\Gamma$ have the following degree parities:
\begin{itemize}
    \item $\mathcal{N}_{v,1}$ consists of a single vertex of degree $1$;
    \item $\mathcal{N}_{v,2}$ consists of two vertices of degrees $1$ and $0$;
    \item for $3\le j\le m-2$, the set $\mathcal{N}_{v,j}$ consists of three vertices of degrees $1$, $0$, and $1$;
    \item $\mathcal{N}_{v,m-1}$ consists of three vertices of degrees $0$, $1$, and $1$;
    \item $\mathcal{N}_{v,m}$ consists of a single vertex of degree $1$.
\end{itemize}

See Figure~\ref{fig:Gamma} for the case $\mathcal{T}_5$, where all vertex degrees are indicated.

It follows that no two vertices lie in the same equivalence class, and hence the induced partition of the vertex set $V(\Gamma)$ consists entirely of singleton sets. By Theorem~\ref{thm:ParitySeparationImpliesLocalX}, the standard reduced dynamical Lie algebra
\(
\mathfrak{g}^v_{v-\mathcal{T}_m,\mathrm{std}}
\)
contains the Pauli operator $iX_w$ for every vertex $w$ of the reduced graph $\mathcal{T}_m$. Consequently, we obtain the inclusion
\begin{equation}
    \label{eq:TreeDLAinclusion}
    \mathfrak{g}_{\mathcal{T}_m,\mathrm{free}}
\;\subseteq\;
\mathfrak{g}^v_{v-\mathcal{T}_m,\mathrm{std}}.
\end{equation}

It is straightforward to verify that when $m=2k$ is even, the tree $\mathcal{T}_m$ is an even-odd bipartite graph. Hence, by Theorem~1 of~\cite{KLFCCZ}, the free DLA associated with $\mathcal{T}_m$ for even $m$ has exponential dimension in $n$, the total number of vertices of $\mathcal{T}_m$. A direct count gives
\begin{equation} 
\label{eq:TmVertexCount}
n = 3(m-3)+2+1+1 = 3m-5. 
\end{equation}

Therefore, the family of graphs $\{v-\mathcal{T}_m\}_{m>3,\; m\ \mathrm{even}}$ provides an explicit class of acyclic graphs for which the dimension of the standard reduced DLA grows exponentially in the number of vertices.

\end{ex}

\begin{figure}[h!]
    \centering
    \begin{tikzpicture}[scale=0.9]
        \node[circle, draw, fill=red, inner sep=2pt] (M) at (2,4) {};
         \node at (2.3,4.1) {$v$};
        \node[circle, draw, fill=blue, inner sep=2pt] (N) at (2,3) {};
        \node at (2.3,3.1) {$1$};
        \node[circle, draw, fill=black, inner sep=2pt] (O) at (3,2) {};
        \node at (3.3,2.1) {$0$};
        \node[circle, draw, fill=black, inner sep=2pt] (P) at (4,1) {};
       \node at (4.3,1.1) {$1$};

        \node[circle, draw, fill=blue, inner sep=2pt] (C) at (1,2) {};
        \node at (1.4,2.1) {$1$};
        \node[circle, draw, fill=blue, inner sep=2pt] (E) at (0,1) {};
        \node at (0.4,1.1) {$1$};
        \node[circle, draw, fill=black, inner sep=2pt] (F) at (2,1) {};
        \node at (2.3,1.1) {$0$};
        \node[circle, draw, fill=blue, inner sep=2pt] (G) at (-1,0) {};
        \node at (-0.6,0.1) {$0$};
        \node[circle, draw, fill=black, inner sep=2pt] (H) at (1,0) {};
        \node at (1.3,0.1) {$1$};
        \node[circle, draw, fill=black, inner sep=2pt] (I) at (3,0) {};
        \node at (3.3,0.1) {$1$};
        \node[circle, draw, fill=blue, inner sep=2pt] (J) at (-2,-1) {};
        \node at (-1.6,-0.9) {$1$};

        \node at (1.7,3.1) {\color{blue}{$1$}};
        \node at (0.7,2.1) {\color{blue}{$2$}};
        \node at (-0.3,1.1) {\color{blue}{$3$}};
        \node at (-1.3,0.1) {\color{blue}{$4$}};
        \node at (-2.3,-0.9) {\color{blue}{$5$}};
        \draw (M) -- (N);
        \draw[blue] (N) -- (C);
        \draw (N) -- (O);
        \draw (O) -- (P);
        \draw[blue] (C) -- (E);
        \draw (C) -- (F);
        \draw (C) -- (I);
        \draw[blue] (E) -- (G);
        \draw (E) -- (H);
        \draw[blue] (G) -- (J);
    \end{tikzpicture}\hspace{0.4in}
      \begin{tikzpicture}[scale=0.9]
          \node[circle, draw, fill=black, inner sep=2pt] (N) at (2,3) {};
        \node[circle, draw, fill=black, inner sep=2pt] (O) at (3,2) {};
         \node[circle, draw, fill=black, inner sep=2pt] (P) at (4,1) {};   
        
        \node[circle, draw, fill=black, inner sep=2pt] (C) at (1,2) {};
        \node[circle, draw, fill=black, inner sep=2pt] (E) at (0,1) {};
        \node[circle, draw, fill=black, inner sep=2pt] (F) at (2,1) {};
        \node[circle, draw, fill=black, inner sep=2pt] (G) at (-1,0) {};
        \node[circle, draw, fill=black, inner sep=2pt] (H) at (1,0) {};
        \node[circle, draw, fill=black, inner sep=2pt] (I) at (3,0) {};
        \node[circle, draw, fill=black, inner sep=2pt] (J) at (-2,-1) {};

          \draw (N) -- (C);
        \draw (N) -- (O);
        \draw (O) -- (P);
        \draw (C) -- (E);
        \draw (C) -- (F);
        \draw (C) -- (I);
        \draw (E) -- (G);
        \draw (E) -- (H);
        \draw (G) -- (J);
    \end{tikzpicture}
  \caption{The graph $v-\mathcal{T}_5$, constructed from the path graph $P_5$ (shown in blue). Each vertex is labeled by its index (left) and the parity of its degree  (right). Also shown is the reduced graph $\mathcal{T}_5$, obtained by removing the vertex $v$ together with its incident edge.}
 
    \label{fig:Gamma}
\end{figure}
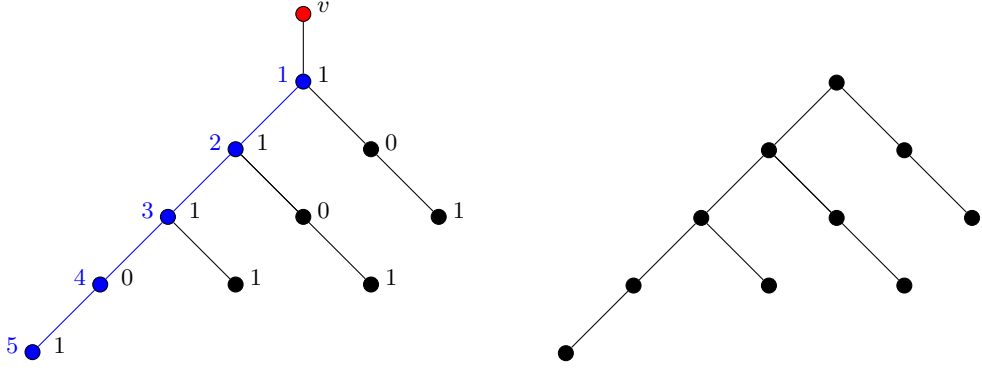

\begin{defn}
\label{defn:vertexSubsetsExtended}
Let \( \Gamma \) be a connected graph and \(v\in V(\Gamma) \) be a vertex.
For integers \( j \ge 0 \) and \( \ell \ge 1 \), and for a binary sequence
\( \mathbf{a}=(a_1,\dots,a_\ell)\in\mathbb{Z}_2^\ell \),
define \( \mathcal{C}^{v}_{j,\ell,\mathbf{a}} \subseteq \mathcal{N}_{v,j} \) as follows: \\
A vertex \( w \in \mathcal{N}_{v,j} \) belongs to \( \mathcal{C}^{v}_{j,\ell,\mathbf{a}} \) if and only if
there exists a vertex \( w' \in \mathcal{N}_{v,j+\ell} \) and a path 
\[
w = w_0 - w_1 - \cdots - w_\ell = w'
\]
such that
\[
\deg(w_s) \equiv a_s \pmod{2}
\qquad \text{for all } s=1,\dots,\ell .
\]
\end{defn}

\begin{thm}
\label{thm:ParitySeparationImpliesLocalX}
Let $\Gamma$ be a connected graph and let $v\in V(\Gamma)$ be a vertex.  
Assume the following:
\begin{enumerate}
\item[(i)] For every $j \ge 1$ and every $w \in \mathcal{N}_{v,j}$,
\[
|\{\, u \in \mathcal{N}_{v,j-1} : (u,w) \in E \,\}| \equiv 1 \pmod{2}.
\]

\item[(ii)] For every $j \ge 1$, every $\mathbf{a} \in \mathbb{Z}_2^j$, and every 
$w \in \mathcal{C}^v_{j,\mathbf{a}}$,
\[
|\{\, u \in \mathcal{C}^v_{j-1,\mathbf{a}'} : (u,w) \in E \,\}| \equiv 1 \pmod{2},
\]
where $\mathbf{a}' = (a_1,\dots,a_{j-1})$.
\item[(iii)] For every $j \ge 1$ and for every pair of distinct vertices 
$w\neq w'\in \mathcal{N}_{v,j}$ 

there exists a binary sequence $\mathbf{a}\in\mathbb{Z}_2^j$ such that
\[
w\in \mathcal{C}^v_{j,\mathbf{a}}
\quad\text{and}\quad
w'\notin \mathcal{C}^v_{j,\mathbf{a}}.
\]

Equivalently,  no two vertices in $\mathcal{N}_{v,j}$ share identical parity profiles along all shortest paths from $v$.
\end{enumerate} 

Then for every vertex $w\in \mathcal{N}_{v,j}$, the corresponding single-site Pauli generator lies in the standard reduced dynamical Lie algebra:
\begin{equation}
\label{eq:LocalXInReducedDLA}
iX_w \;\in\; \mathfrak{g}^v_{\Gamma,\mathrm{std}},
\qquad \forall\, w\in \mathcal{N}_{v,j}.
\end{equation}

If the above condition holds for all $j\ge 1$, then the standard
and free reduced dynamical Lie algebras coincide, namely
\begin{equation}
\label{eq:ReducedDLAEquality}
\mathfrak{g}^v_{\Gamma,\mathrm{std}}
=\mathfrak{g}^v_{\Gamma,\mathrm{free}}.
\end{equation}
\end{thm}

\begin{rmk}
    \label{rmk:ExtensionOfResults}
    The statements of Theorems~\ref{thm:DistDLAElements},
\ref{thm:CoolDLAElements} and ~\ref{thm:ParitySeparationImpliesLocalX} are local in nature, in the sense that they depend only on the availability of a single-site Pauli-$X$ generator at a chosen vertex.
Consequently, all conclusions remain valid with the marked vertex $v$ replaced
by \emph{any} vertex $w\in V(\Gamma)$ satisfying
$iX_w \in \mathfrak{g}_{\Gamma,\mathrm{std}}$.

Likewise, in the reduced setting, all constructions and containment results
continue to hold with $v$ replaced by any vertex
$w\neq v$ for which the standard reduced  dynamical Lie algebra
$\mathfrak{g}^{\,v}_{\Gamma,\mathrm{std}}$ contains the element $iX_w$.
\end{rmk}

\begin{obs}
The constructions in Theorem~\ref{thm:CoolDLAElements} and Remark~\ref{rmk:ExtensionOfResults}
can be combined to systematically isolate individual vertices within a
fixed distance layer and thereby establish the containment of all
single-site Pauli-$X$ operators in the standard reduced dynamical Lie
algebra. This is illustrated in Example \ref{ex:ExpFamilyTrees2}
\end{obs}

\begin{ex}
\label{ex:ExpFamilyTrees2}
In order to show the practical importance of the proposed theorems and how they give the distinguished Pauli-$X$ gates in the reduced DLAs, we give a full analysis of the graph $\Gamma'$ depicted in Figure \ref{fig:GammaEx}. First, we give the full collection of subsets $\mathcal{N}_{v, j}$ (see Definition \ref{defn:vertexNSubsets}): 
\[
\begin{aligned}
    \mathcal{N}_{v,1} &= \{v_{1, 1}\}\\
    \mathcal{N}_{v, 2} &= \{v_{2, 1}, v_{2, 2}\}\\
    \mathcal{N}_{v, 3}& = \{v_{3, 1}, v_{3, 2}, v_{3, 3}, v_{3, 4}\};\\
    \mathcal{N}_{v, 4} &= \{v_{4, 1}, v_{4, 2}, v_{4, 3}, v_{4, 4}\};\\ \mathcal{N}_{v, 5} &= \{v_{5, 1}\}.\\
\end{aligned} 
\]

Next, we list out the subsets of vertices $\mathcal{C}^{v}_{j, \textbf{a}}$ according to Definition \ref{defn:vertexSubsets}  for the graph under consideration: 
\[
\begin{aligned}
\mathcal{C}^{v}_{1, \mathbf{a} = \mathbf{1}} &= \{v_{1, 1}\};\\
\mathcal{C}^{v}_{2, \mathbf{a} = \mathbf{11}} &= \{v_{2, 1}, v_{2, 2}\};\\
\mathcal{C}^{v}_{3, \mathbf{a} = \mathbf{110}} &= \{v_{3, 2}, v_{3, 4}\};\\
\mathcal{C}^{v}_{3, \mathbf{a} = \mathbf{111}} &= \{v_{3, 1}, v_{3, 3}\}; \\
\mathcal{C}^{v}_{4, \mathbf{a}=\mathbf{1101}}&=\{v_{4,3}, v_{4, 4}\}; \\
\mathcal{C}^{v}_{4, \mathbf{a}=\mathbf{1110}} &= \{v_{4,1}\}; \\
\mathcal{C}^{v}_{4, \mathbf{a}=\mathbf{1111}} &= \{v_{4,2}\};\\
\mathcal{C}^{v}_{5, \mathbf{a}=\mathbf{11101}} &= \{v_{5,1}\}.
\end{aligned} 
\]

For the four remaining subsets $\mathcal{C}^{v}_{2, \mathbf{a} = \mathbf{11}}, \mathcal{C}^{v}_{3, \mathbf{a} = \mathbf{110}},\mathcal{C}^{v}_{3, \mathbf{a} = \mathbf{111}}$ and $\mathcal{C}^{v}_{4, \mathbf{a}=\mathbf{1101}}$ of cardinality $2$, we can use the observation from Remark \ref{rmk:ExtensionOfResults} by choosing node $v_{5,1}$ as the marked node for the reduced settings, to partition the pair of vertices according to the respective distances from $v_{5,1}$. 

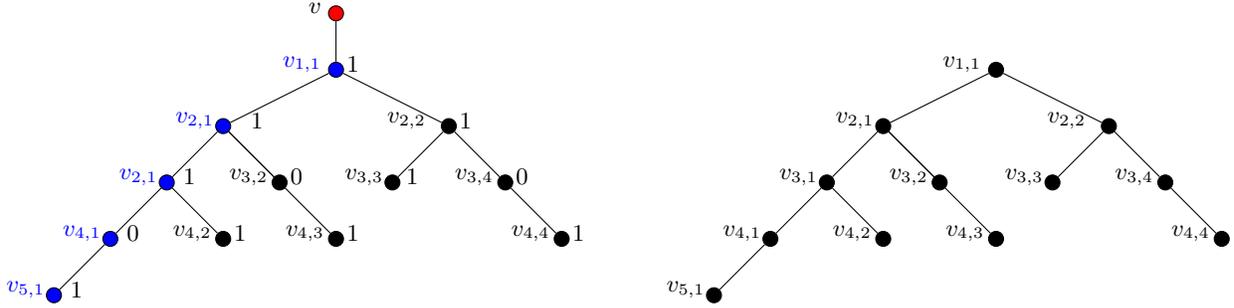
\begin{figure}[h!]
    \centering
    \begin{tikzpicture}[scale=0.75]
        \node[circle, draw, fill=red, inner sep=2pt] (M) at (2,4) {};
         \node at (1.63,4.1) {$v$};
        \node[circle, draw, fill=blue, inner sep=2pt] (N) at (2,3) {};
        \node at (2.3,3.1) {$1$};
        \node at (1.4,3.1) {\color{blue}{$v_{1,1}$}};
        \node[circle, draw, fill=black, inner sep=2pt] (O) at (4,2) {};
        \node at (4.3,2.1) {$1$};
        \node at (3.25,2.1) {$v_{2, 2}$};
        \node[circle, draw, fill=black, inner sep=2pt] (P) at (5,1) {};
        \node at (5.3,1.1) {$0$};  
        \node at (4.45,1.1) {$v_{3, 4}$};
        \node[circle, draw, fill=black, inner sep=2pt] (P1) at (6,0) {};
        \node at (6.3,0.1) {$1$};
        \node at (5.45,0.1) {$v_{4, 4}$};
        \node[circle, draw, fill=black, inner sep=2pt] (O1) at (3,1) {};
        \node at (3.35,1.1) {$1$};
        \node at (2.5,1.1) {$v_{3, 3}$};
        \node[circle, draw, fill=blue, inner sep=2pt] (C) at (0,2) {};
        \node at (0.6,2.1) {$1$};
        \node at (-0.5,2.1) {\color{blue}{$v_{2,1}$}};
        \node[circle, draw, fill=blue, inner sep=2pt] (E) at (-1,1) {};
        \node at (-0.6,1.1) {$1$};
        \node at (-1.5,1.1) {\color{blue}{$v_{3,1}$}};
        \node[circle, draw, fill=black, inner sep=2pt] (F) at (1,1) {};
        \node at (1.3,1.1) {$0$};
        \node at (0.45,1.1) {$v_{3, 2}$};
        \node[circle, draw, fill=blue, inner sep=2pt] (G) at (-2,0) {};
        \node at (-1.6,0.1) {$0$};
        \node at (-2.5,0.1) {\color{blue}{$v_{4,1}$}};
        \node[circle, draw, fill=black, inner sep=2pt] (H) at (0,0) {};
        \node at (0.3,0.1) {$1$};
        \node at (-0.55,0.1) {$v_{4, 2}$};
        \node[circle, draw, fill=black, inner sep=2pt] (I) at (2,0) {};
        \node at (2.3,0.1) {$1$};
        \node at (1.45,0.1) {$v_{4, 3}$};
        \node[circle, draw, fill=blue, inner sep=2pt] (J) at (-3,-1) {};
        \node at (-2.6,-0.9) {$1$};
        \node at (-3.5,-0.9) {\color{blue}{$v_{5,1}$}};
        \draw (M) -- (N);
        \draw (N) -- (C);
        \draw (N) -- (O);
        \draw (O) -- (P);
        \draw (C) -- (E);
        \draw (C) -- (F);
        \draw (C) -- (I);
        \draw (E) -- (G);
        \draw (E) -- (H);
        \draw (G) -- (J);
        \draw (P) -- (P1);
        \draw (O) -- (O1);
    \end{tikzpicture}\hspace{0.3in}
      \begin{tikzpicture}[scale=0.75]

        \node[circle, draw, fill=black, inner sep=2pt] (N) at (2,3) {};
        \node at (1.4,3.1) {$v_{1,1}$};
        \node[circle, draw, fill=black, inner sep=2pt] (O) at (4,2) {};
        \node at (3.25,2.1) {$v_{2, 2}$};
        \node[circle, draw, fill=black, inner sep=2pt] (P) at (5,1) {};
        \node at (4.45,1.1) {$v_{3, 4}$};
        \node[circle, draw, fill=black, inner sep=2pt] (P1) at (6,0) {};
        \node at (5.45,0.1) {$v_{4, 4}$};
        \node[circle, draw, fill=black, inner sep=2pt] (O1) at (3,1) {};
        \node at (2.5,1.1) {$v_{3, 3}$};
        \node[circle, draw, fill=black, inner sep=2pt] (C) at (0,2) {};
        \node at (-0.5,2.1) {$v_{2,1}$};
        \node[circle, draw, fill=black, inner sep=2pt] (E) at (-1,1) {};
        \node at (-1.5,1.1) {$v_{3,1}$};
        \node[circle, draw, fill=black, inner sep=2pt] (F) at (1,1) {};
        \node at (0.45,1.1) {$v_{3, 2}$};
        \node[circle, draw, fill=black, inner sep=2pt] (G) at (-2,0) {};
        \node at (-2.5,0.1) {$v_{4,1}$};
        \node[circle, draw, fill=black, inner sep=2pt] (H) at (0,0) {};
        \node at (-0.55,0.1) {$v_{4, 2}$};
        \node[circle, draw, fill=black, inner sep=2pt] (I) at (2,0) {};
        \node at (1.45,0.1) {$v_{4, 3}$};
        \node[circle, draw, fill=black, inner sep=2pt] (J) at (-3,-1) {};
        \node at (-3.5,-0.9) {$v_{5,1}$};
        \draw (N) -- (C);
        \draw (N) -- (O);
        \draw (O) -- (P);
        \draw (C) -- (E);
        \draw (C) -- (F);
        \draw (C) -- (I);
        \draw (E) -- (G);
        \draw (E) -- (H);
        \draw (G) -- (J);
        \draw (P) -- (P1);
        \draw (O) -- (O1);
    \end{tikzpicture}
    \caption{The graph $\Gamma'$, where each vertex is labeled by its index (left) and by the parity of its degree (right), and the reduced graph $\Gamma'_v$, obtained by removing the vertex $v$ together with all edges incident to it.}

    \label{fig:GammaEx}
\end{figure}
\end{ex}

We conclude this section with a result establishing sufficient conditions for the standard and free reduced dynamical Lie algebra to coincide with the full special unitary algebra on the reduced Hilbert space.

\begin{thm}
\label{thm:FreeReducedFullStructure}
Let $v \in V(\Gamma)$ be such that the reduced graph $\Gamma_v$ is connected and is not bipartite or a cycle. Then the free reduced dynamical Lie algebra $\mathfrak{g}^v_{\Gamma,\mathrm{free}}$ is simple and coincides with the special unitary algebra on the reduced Hilbert space:
\begin{equation}
    \label{eq:FreeReducedIsSU}
    \mathfrak{g}^v_{\Gamma,\mathrm{free}}
    \cong \mathfrak{su}\!\left(W_v\right)
    \cong \mathfrak{su}\!\left(2^{\,n-1}\right).
\end{equation}
\end{thm}

\begin{cor}
\label{thm:ReducedDLAIsFullUnitary}
Let $v \in V(\Gamma)$ be such that the reduced graph $\Gamma_v$ is connected and is neither bipartite nor a cycle. Assume moreover that the containment in \eqref{eq:ERContainment} holds. Then the standard reduced dynamical Lie algebra $\mathfrak{g}^v_{\Gamma,\mathrm{std}}$ is equal to the free reduced dynamical Lie algebra $\mathfrak{g}^v_{\Gamma,\mathrm{free}}$ and, hence, is isomorphic to the special unitary Lie algebra on the reduced Hilbert space $W_v$:
\begin{equation}
    \label{eq:std_free_equality}
    \mathfrak{g}^v_{\Gamma,\mathrm{std}} = \mathfrak{g}^v_{\Gamma,\mathrm{free}}.
\end{equation}
\end{cor}

\section{Original and Reduced DLAs for MaxCut QAOAs on Certain Graphs}
\label{sec:DLASonSomeGraphs}
In this section, we compare the standard and reduced dynamical Lie algebras  associated with the Quantum Approximate Optimization Algorithm  for MaxCut problems on various types of graphs. Our comparison is based on the dimensions of these Lie algebras.

\subsection{Exponential vs.\ Quadratic Growth: The $k$-Armed Spider Graphs}

We start by presenting a family of graphs for which the dimension of the standard dynamical Lie algebra  grows exponentially in the number of vertices, while there exists a distinguished vertex for which the dimension of the corresponding standard \emph{reduced} DLA is only quadratic in the number of vertices. This provides a concrete and explicit separation between the full and reduced settings.

The family under consideration consists of the so-called $k$-armed spider graphs, denoted
\[
\mathcal{O}_{m_1,\dots,m_k},
\]
which were studied in~\cite{MYAZ}. Such a graph is obtained by taking a central vertex $v$ and attaching to it $k$ path graphs
\[
P_{m_1},\dots,P_{m_k},
\]
each joined to $v$ at one of its endpoints. In other words, $\mathcal{O}_{m_1,\dots,m_k}$ consists of $k$ disjoint paths whose initial vertices are identified with the central vertex $v$ (see Figure~\ref{fig:OctopusGraph}).

The reduction of $\mathcal{O}_{m_1,\dots,m_k}$ at the central vertex $v$ removes $v$ and disconnects the graph into a disjoint union of the path graphs
\[
P_{m_1}\sqcup\dots\sqcup P_{m_k}:
\]

\begin{figure}[H]
\begin{center}
    \begin{tikzpicture}[scale=0.8]

\node[circle, draw, fill=black, inner sep=2pt] (r1c11) at (12,5) {};
\node at (12,4.6) {$w_{1,2}$};

\node[circle, draw, fill=black, inner sep=2pt] (r2c9) at (9.6,4) {};
\node at (9.6,3.6) {$w_{1,1}$};

\node[circle, draw, fill=black, inner sep=2pt] (r3c1) at (0,4.2) {};
\node at (0,3.8) {$w_{3,5}$};
\node[circle, draw, fill=black, inner sep=2pt] (r3c3) at (2.4,3.8) {};
\node at (2.4,3.4) {$w_{3,4}$};
\node[circle, draw, fill=black, inner sep=2pt] (r3c4) at (3.6,3.6) {};
\node at (3.6,3.2) {$w_{3,3}$};
\node[circle, draw, fill=black, inner sep=2pt] (r3c5) at (4.8,3.4) {};
\node at (4.8,3) {$w_{3,2}$};
\node[circle, draw, fill=black, inner sep=2pt] (r3c6) at (6,3.2) {};
\node at (5.85,2.85) {$w_{3,1}$};
\node[circle, draw, fill=red, inner sep=2pt] (r3c7) at (7.2,3) {};
\node at (7.2,2.6) {$v$};

\node[circle, draw, fill=black, inner sep=2pt] (r4c6) at (5.2,2.5) {};
\node at (5.2,2.1) {$w_{4,1}$};
\node[circle, draw, fill=black, inner sep=2pt] (r4c9) at (8.8,2) {};
\node at (8.8,1.6) {$w_{2,1}$};

\node[circle, draw, fill=black, inner sep=2pt] (r5c5) at (3.2,2) {};
\node at (3.2,1.6) {$w_{4,2}$};
\node[circle, draw, fill=black, inner sep=2pt] (r5c11) at (10.4,1) {};
\node at (10.4,0.6) {$w_{2,2}$};

\node[circle, draw, fill=black, inner sep=2pt] (r6c4) at (1.2,1.5) {};
\node at (1.2,1.1) {$w_{4,3}$};
\node[circle, draw, fill=black, inner sep=2pt] (r6c5) at (-0.8,1) {};
\node at (-0.7,0.6) {$w_{4,4}$};
\node[circle, draw, fill=black, inner sep=2pt] (r6c13) at (12,0) {};
\node at (12,-0.4) {$w_{2,3}$};


\draw (r3c4) -- (r3c3);
\draw (r3c5) -- (r3c4);
\draw (r3c6) -- (r3c5);
\draw (r3c7) -- (r2c9);
\draw (r3c7) -- (r3c6);
\draw (r3c7) -- (r4c6);
\draw (r3c7) -- (r4c9);
\draw (r4c6) -- (r5c5);
\draw (r4c9) -- (r5c11);
\draw (r5c5) -- (r6c4);
\draw (r6c4) -- (r6c5);
\draw (r3c1) -- (r3c3);
\draw (r1c11) -- (r2c9);
\draw (r5c11) -- (r6c13);

\end{tikzpicture}\vspace{0.3in}

\begin{tikzpicture}[scale=0.8]

\node[circle, draw, fill=black, inner sep=2pt] (r1c11) at (12,5) {};
\node at (12,4.6) {$w_{1,2}$};

\node[circle, draw, fill=black, inner sep=2pt] (r2c9) at (9.6,4) {};
\node at (9.6,3.6) {$w_{1,1}$};

\node[circle, draw, fill=black, inner sep=2pt] (r3c1) at (0,4.2) {};
\node at (0,3.8) {$w_{3,5}$};
\node[circle, draw, fill=black, inner sep=2pt] (r3c3) at (2.4,3.8) {};
\node at (2.4,3.4) {$w_{3,4}$};
\node[circle, draw, fill=black, inner sep=2pt] (r3c4) at (3.6,3.6) {};
\node at (3.6,3.2) {$w_{3,3}$};
\node[circle, draw, fill=black, inner sep=2pt] (r3c5) at (4.8,3.4) {};
\node at (4.8,3) {$w_{3,2}$};
\node[circle, draw, fill=black, inner sep=2pt] (r3c6) at (6,3.2) {};
\node at (5.85,2.85) {$w_{3,1}$};

\node[circle, draw, fill=black, inner sep=2pt] (r4c6) at (5.2,2.5) {};
\node at (5.2,2.1) {$w_{4,1}$};
\node[circle, draw, fill=black, inner sep=2pt] (r4c9) at (8.8,2) {};
\node at (8.8,1.6) {$w_{2,1}$};

\node[circle, draw, fill=black, inner sep=2pt] (r5c5) at (3.2,2) {};
\node at (3.2,1.6) {$w_{4,2}$};
\node[circle, draw, fill=black, inner sep=2pt] (r5c11) at (10.4,1) {};
\node at (10.4,0.6) {$w_{2,2}$};

\node[circle, draw, fill=black, inner sep=2pt] (r6c4) at (1.2,1.5) {};
\node at (1.2,1.1) {$w_{4,3}$};
\node[circle, draw, fill=black, inner sep=2pt] (r6c5) at (-0.8,1) {};
\node at (-0.7,0.6) {$w_{4,4}$};
\node[circle, draw, fill=black, inner sep=2pt] (r6c13) at (12,0) {};
\node at (12,-0.4) {$w_{2,3}$};


\draw (r3c4) -- (r3c3);
\draw (r3c5) -- (r3c4);
\draw (r3c6) -- (r3c5);
\draw (r4c6) -- (r5c5);
\draw (r4c9) -- (r5c11);
\draw (r5c5) -- (r6c4);
\draw (r6c4) -- (r6c5);
\draw (r3c1) -- (r3c3);
\draw (r1c11) -- (r2c9);
\draw (r5c11) -- (r6c13);

\end{tikzpicture}
\end{center}

\caption{Left: the $4$-armed spider graph $\mathcal{O}_{2,3,4,5}$ with central vertex $v$ and $4$ path graphs $P_{2},P_{3},P_{4}$ and $P_{5}$ attached. 
Right: the reduced graph $\mathcal{O}^v_{2,3,4,5}$ obtained by removing the central vertex $v$, resulting in a disjoint union of the $4$ path graphs.}
\label{fig:OctopusGraph}

\end{figure}

Assume now that the lengths $m_1,\dots,m_k$ are pairwise distinct. It was established in~\cite[Corollary~26]{MYAZ} that in this case the standard and free DLAs of the full spider graph coincide, namely
\[
\mathfrak{g}_{\mathcal{O}_{m_1,\dots,m_k},\mathrm{std}}
=
\mathfrak{g}_{\mathcal{O}_{m_1,\dots,m_k},\mathrm{free}}.
\]

Consequently, Theorem~1 of~\cite{KLFCCZ} allows to conclude that $\dim(\mathfrak{g}_{\mathcal{O}_{m_1,\dots,m_k},\mathrm{std}})$ is exponential in $n$. 

Combining Theorem~\ref{thm:DistDLAElements} with Remark~\ref{rmk:ExtensionOfResults}, one obtains directly the analogous statement for the reduced DLAs at the central vertex $v$:
\begin{equation}
\label{eq:ReducedSpiderDecomposition}
\mathfrak{g}^v_{\mathcal{O}_{m_1,\dots,m_k},\mathrm{std}}
=
\mathfrak{g}^v_{\mathcal{O}_{m_1,\dots,m_k},\mathrm{free}}
=
\mathfrak{g}_{P_{m_1},\mathrm{free}}
\oplus
\dots
\oplus
\mathfrak{g}_{P_{m_k},\mathrm{free}}.
\end{equation}

Thus, the reduced algebra decomposes as a direct sum of the free DLAs associated with the individual arms.

We next compute its dimension. It is known that the free DLA of the path graph $P_n$, with an additional single-site $Z$ generator at one of its vertices, is isomorphic to the Lie algebra $\mathfrak{so}(2n+1)$, which has dimension $2n^2+n$ (see Proposition~B3 in~\cite{KLFCCZ}). Applying this to each arm yields
\begin{equation}
\label{eq:ReducedSpiderSO}
\mathfrak{g}^v_{\mathcal{O}_{m_1,\dots,m_k},\mathrm{std}}
=
\mathfrak{so}(2m_1+1)
\oplus
\dots
\oplus
\mathfrak{so}(2m_k+1).
\end{equation}

Consequently,
\[
\dim\!\left(
\mathfrak{g}^v_{\mathcal{O}_{m_1,\dots,m_k},\mathrm{std}}
\right)
=
\sum\limits_{j=1}^k (2m_j^2+m_j)
=
2\sum\limits_{j=1}^k m_j^2
+
\sum\limits_{j=1}^k m_j.
\]

Let $n=\sum\limits_{j=1}^k m_j$ denote the total number of non-central vertices. Then
\begin{equation}
\label{eq:ReducedSpiderDim}
\dim\!\left(
\mathfrak{g}^v_{\mathcal{O}_{m_1,\dots,m_k},\mathrm{std}}
\right)
=
2\sum\limits_{j=1}^k m_j^2 + n
\le
2\left(\sum\limits_{j=1}^k m_j\right)^2 + n
=
2n^2+n.
\end{equation}
This establishes a quadratic upper bound in the system size.

In contrast, for the full (non-reduced) spider graph the standard DLA grows exponentially with the number of vertices, as shown in~\cite{MYAZ}. The $k$-armed spider graphs therefore provide an explicit family in which symmetry reduction at a carefully chosen vertex dramatically alters the algebraic complexity: from exponential growth in the full system to quadratic growth in the reduced one.

\subsection{Acyclic Graphs}

Let $\Gamma$ be a connected acyclic graph, and fix a vertex $v\in V(\Gamma)$ at which we perform the reduction.
Since $\Gamma$ contains no cycles, it can be naturally regarded as a rooted tree with root $v$.
In particular, for every vertex $w\in V(\Gamma)$ there exists a unique simple path
\[
v = w_0 - w_1 - \cdots - w_j = w
\]
connecting $v$ to $w$, where $j=\operatorname{dist}(v,w)$. As a consequence, each vertex $w\in\mathcal{N}_{v,j}$ of the reduced graph $\Gamma_v$ is associated with a
\emph{unique} parity-degree profile
\[
\mathbf{a}=(a_1,\dots,a_j)\in\mathbb{Z}_2^j,
\qquad
a_s \equiv \deg(w_s)\pmod{2},
\]
formed by recording the parities of the degrees of the intermediate vertices along this path.
Equivalently, every vertex $w\in\mathcal{N}_{v,j}$ belongs to exactly one of the subsets
\(
\mathcal{C}^v_{j,\mathbf{a}}
\)
introduced in Definition~\ref{defn:vertexSubsets}.
Thus, for acyclic graphs, the collections $\{\mathcal{C}^v_{j,\mathbf{a}}\}_{\mathbf{a}\in\mathbb{Z}_2^j}$
form a partition of each distance layer $\mathcal{N}_{v,j}$.

This observation allows us to formulate a simple and easily verifiable structural condition that is sufficient for the containment of dynamical Lie algebras in~\eqref{eq:ERContainment} to hold in the acyclic setting. The condition requires that leaf vertices at the same distance from $v$ be
distinguishable by their parity-degree profiles along the unique path to the root.

\begin{thm}
\label{thm:ReducedDLAsForTrees}
Let $\Gamma$ be a connected acyclic graph and let $v\in V(\Gamma)$.
Suppose that for every $j\in\mathbb{Z}_{>0}$, the parity-degree sequences formed by the parities of the degrees of the vertices along the unique paths from $v$ to the \emph{leaf vertices} in $\mathcal{N}_{v,j}$ are pairwise distinct. Then the containment of algebras in~\eqref{eq:ERContainment} holds.
\end{thm}

\begin{rmk}
\label{rmk:LinearTimeCheckTrees}
The parity-degree sequences appearing in Theorem~VI.1 can be computed and
compared efficiently using a breadth-first search (BFS) rooted at the
distinguished vertex $v$.  
Since $\Gamma$ is acyclic, BFS determines the distance layers
$\mathcal{N}_{v,j}$ and uniquely fixes, for each vertex, the simple path
from $v$ to that vertex.

For each leaf $w \in \mathcal{N}_{v,j}$, the associated parity-degree
sequence is obtained by traversing this path and recording the parity of
the degrees of the vertices along it.  
To verify that these sequences are pairwise distinct within each layer
$\mathcal{N}_{v,j}$, one may encode each sequence as a binary string and
store it in a hash table while traversing the leaves at that distance.
A collision in the hash table indicates a violation of the condition in
Theorem~VI.1.

Because each vertex and edge is visited a constant number of times during
the BFS and each parity-degree sequence is processed once, the total time
required to construct all sequences and check their pairwise distinctness
over all distance layers is $O(|V(\Gamma)|)$ for acyclic graphs.
\end{rmk}

\subsection{Asymmetric Graphs}

Next, we shift our focus to asymmetric graphs, i.e., graphs whose automorphism group is trivial. The smallest such graphs have $6$ vertices. Up to isomorphism, there are $8$ distinct asymmetric graphs on $6$ nodes, depicted in Figure~\ref{AsymGraphsSixNodes}.

\begin{figure}[h!]
\label{SixNodeAsymmGraphs}
\begin{center}
\begin{tikzpicture}[scale=1]

\node[circle, fill=black, inner sep=2pt, label={left:$3$}] (A1) at (-1,2) {};
\node[circle, fill=black, inner sep=2pt, label={above:$6$}] (A2) at (0,1) {};
\node[circle, fill=black, inner sep=2pt, label={above:$2$}] (A3) at (2,1) {};
\node[circle, fill=black, inner sep=2pt, label={left:$4$}] (A4) at (-1,0) {};
\node[circle, fill=black, inner sep=2pt, label={below:$1$}] (A5) at (0,-1) {};
\node[circle, fill=black, inner sep=2pt, label={below:$5$}] (A6) at (2,-1) {};
\node  at (1,-2) {$\Gamma_1$};

\draw (A2)--(A1);
\draw (A2)--(A3);
\draw (A2)--(A5);
\draw (A4)--(A2);
\draw (A4)--(A5);
\draw (A5)--(A6);
\draw (A6)--(A3);

\end{tikzpicture}
\hspace{0.5in}\begin{tikzpicture}[scale=1]
    \node[circle, draw, fill=black, inner sep=2pt, label={left:$2$}] (B1) at (-1,2) {};
    \node[circle, draw, fill=black, inner sep=2pt, label={above:$5$}] (B3) at (1,2) {};
    \node[circle, draw, fill=black, inner sep=2pt, label={left:$6$}] (B2) at (0,1) {};
    \node[circle, draw, fill=black, inner sep=2pt, label={left:$3$}] (B4) at (-1,0) {};
    \node[circle, draw, fill=black, inner sep=2pt, label={left:$1$}] (B5) at (2,0) {};
    \node[circle, draw, fill=black, inner sep=2pt, label={left:$4$}] (B6) at (3,-1) {};
\node  at (1,-2) {$\Gamma_2$};
    \draw (B1) -- (B3);
    \draw (B1) -- (B2);
    \draw (B2) -- (B3);
    \draw (B2) -- (B4);
    \draw (B2) -- (B5);
    \draw (B3) -- (B5);
    \draw (B5) -- (B6);
\end{tikzpicture}
\hspace{0.5in}
\begin{tikzpicture}[scale=1]
    \node[circle, draw, fill=black, inner sep=2pt, label={left:$4$}] (D1) at (4,2) {};
    \node[circle, draw, fill=black, inner sep=2pt, label={left:$1$}] (D2) at (3,1) {};
    \node[circle, draw, fill=black, inner sep=2pt, label={left:$2$}] (D3) at (0,0) {};
    \node[circle, draw, fill=black, inner sep=2pt, label={above:$6$}] (D4) at (2,0) {};
    \node[circle, draw, fill=black, inner sep=2pt, label={left:$5$}] (D5) at (1,-1) {};
    \node[circle, draw, fill=black, inner sep=2pt, label={left:$3$}] (D6) at (0,-2) {};
    \node  at (2,-2) {$\Gamma_3$};
    \draw (D2) -- (D1);
    \draw (D3) -- (D4);
    \draw (D3) -- (D5);
    \draw (D4) -- (D2);
    \draw (D5) -- (D4);
    \draw (D6) -- (D5);
\end{tikzpicture}
\vspace{0.4in}

\begin{tikzpicture}[scale=1]

\node[circle, fill=black, inner sep=2pt, label={left:$3$}] (A1) at (-1,2) {};
\node[circle, fill=black, inner sep=2pt,, label={above:$5$}] (A2) at (0,1) {};
\node[circle, fill=black, inner sep=2pt, label={above:$1$}] (A3) at (2,1) {};
\node[circle, fill=black, inner sep=2pt, label={left:$2$}] (A4) at (-1,0) {};
\node[circle, fill=black, inner sep=2pt, label={below:$6$}] (A5) at (0,-1) {};
\node[circle, fill=black, inner sep=2pt, label={below:$4$}] (A6) at (2,-1) {};
\node  at (1,-2) {$\Gamma_4$};

\draw (A2)--(A1);
\draw (A2)--(A3);
\draw (A2)--(A5);
\draw (A4)--(A2);
\draw (A4)--(A5);
\draw (A5)--(A6);
\draw (A6)--(A3);
\draw (A3)--(A5);
\end{tikzpicture}\hspace{0.5in}
\begin{tikzpicture}[scale=1]

\node[circle, fill=black, inner sep=2pt, label={above:$4$}] (C1) at (0,2) {};
\node[circle, fill=black, inner sep=2pt, label={above:$1$}] (C2) at (2,2) {};
\node[circle, fill=black, inner sep=2pt, label={left:$2$}] (C3) at (-1,1) {};
\node[circle, fill=black, inner sep=2pt, label={below:$6$}] (C4) at (0,0) {};
\node[circle, fill=black, inner sep=2pt, label={below:$5$}] (C5) at (2,0) {};
\node[circle, fill=black, inner sep=2pt, label={above:$3$}] (C6) at (4,-1) {};
\node  at (1,-2) {$\Gamma_5$};

\draw (C1)--(C2);
\draw (C1)--(C4);
\draw (C3)--(C1);
\draw (C3)--(C4);
\draw (C4)--(C5);
\draw (C5)--(C2);
\draw (C5)--(C6);
\end{tikzpicture}
\hspace{0.4in}
\begin{tikzpicture}[scale=1]
    \node[circle, draw, fill=black, inner sep=2pt, label={above:$2$}] (A1) at (-1,2) {};
    \node[circle, draw, fill=black, inner sep=2pt, label={above:$4$}] (A3) at (1,2) {};
    \node[circle, draw, fill=black, inner sep=2pt, label={left:$6$}] (A2) at (0,1) {};
    \node[circle, draw, fill=black, inner sep=2pt, label={above:$1$}] (A4) at (2,0) {};
    \node[circle, draw, fill=black, inner sep=2pt, label={above:$3$}] (A5) at (-2,0) {};
    \node[circle, draw, fill=black, inner sep=2pt, label={below:$5$}] (A6) at (0,-1) {};
    \node  at (1,-2) {$\Gamma_6$};
    
    \draw (A1) -- (A3);
    \draw (A1) -- (A2);
    \draw (A2) -- (A3);
    \draw (A3) -- (A4);
    \draw (A2) -- (A4);
    \draw (A2) -- (A5);
    \draw (A6) -- (A4);
    \draw (A6) -- (A5);
\end{tikzpicture}

\vspace{0.5in}
\begin{tikzpicture}[scale=1]

\node[circle, fill=black, inner sep=2pt, label={above:$3$}] (A1) at (3,-2) {};
\node[circle, fill=black, inner sep=2pt, label={above:$4$}] (A2) at (0,1) {};
\node[circle, fill=black, inner sep=2pt, label={above:$1$}] (A3) at (2,1) {};
\node[circle, fill=black, inner sep=2pt, label={left:$2$}] (A4) at (-1,0) {};
\node[circle, fill=black, inner sep=2pt, label={below:$6$}] (A5) at (0,-1) {};
\node[circle, fill=black, inner sep=2pt, label={below:$5$}] (A6) at (2,-1) {};
\node  at (1,-2) {$\Gamma_7$};

\draw (A6)--(A1);
\draw (A2)--(A3);
\draw (A2)--(A5);
\draw (A4)--(A2);
\draw (A4)--(A5);
\draw (A5)--(A6);
\draw (A6)--(A3);
\draw (A3)--(A5);
\end{tikzpicture}\hspace{1in}
\begin{tikzpicture}[scale=1]

\node[circle, fill=black, inner sep=2pt, label={right:$2$}] (A1) at (1,-2) {};
\node[circle, fill=black, inner sep=2pt, label={above:$5$}] (A2) at (0,1) {};
\node[circle, fill=black, inner sep=2pt, label={above:$1$}] (A3) at (2,1) {};
\node[circle, fill=black, inner sep=2pt, label={left:$3$}] (A4) at (-1,0) {};
\node[circle, fill=black, inner sep=2pt, label={below:$6$}] (A5) at (0,-1) {};
\node[circle, fill=black, inner sep=2pt, label={below:$4$}] (A6) at (2,-1) {};
\node  at (-0.5,-2) {$\Gamma_8$};

\draw (A6)--(A1);
\draw (A5)--(A1);
\draw (A2)--(A3);
\draw (A4)--(A2);
\draw (A4)--(A5);
\draw (A5)--(A6);
\draw (A6)--(A3);
\draw (A2)--(A6);
\draw (A3)--(A5);
\end{tikzpicture}
\vspace{0.4in}

\begin{tabular}{|c|c|c|c|c|c|c|c|}
\hline
$i$ &
$\dim\bigl(\mathfrak{g}_{\Gamma_i,\mathrm{std}}\bigr)$ & $\dim\bigl(\mathfrak{g}^{v_1}_{\Gamma_i,\mathrm{std}}\bigr)$ & $\dim\bigl(\mathfrak{g}^{v_2}_{\Gamma_i,\mathrm{std}}\bigr)$ & $\dim\bigl(\mathfrak{g}^{v_3}_{\Gamma_i,\mathrm{std}}\bigr)$ &
$\dim\bigl(\mathfrak{g}^{v_4}_{\Gamma_i,\mathrm{std}}\bigr)$ & $\dim\bigl(\mathfrak{g}^{v_5}_{\Gamma_i,\mathrm{std}}\bigr)$ & $\dim\bigl(\mathfrak{g}^{v_6}_{\Gamma_i,\mathrm{std}}\bigr)$\\
\hline
$1$ & 1984 & 956 & 1020& 1023 & 1021 & 991 & \color{red}{258} \\
\hline
$2$ & 1970 & \color{red}{258} & 995& 1014 & 1023 & 957 & \color{red}{258}\\
\hline
$3$ & 1461 & 258 & 862& 907 & 915 & 258 & \color{red}{73}\\
\hline
$4$ & 1984 & 1011 & 1022 & 1023 & 1008 & \color{red}{258} & 1023\\
\hline
$5$ & 1908 & 971 & 979 & 1023 & 909 & \color{red}{258} & 890 \\
\hline
$6$ & 1984 & 1013 & 1023& 1023 & 1023 & 1023 & 870\\
\hline
$7$ & 1984 & 987 & 964 & 1023 & 1014 & \color{red}{258} & 995\\
\hline
$8$ & 1984 & \color{red}{1023} & \color{red}{1023} & \color{red}{1023}  & \color{red}{1023} & \color{red}{1023} & \color{red}{1023} \\
\hline
\end{tabular}

 \caption{Dimensions of the standard and reduced dynamical Lie algebras for asymmetric graphs on $6$ nodes.}
    \label{AsymGraphsSixNodes}
\end{center}
\end{figure}

Using the computational framework implemented in~\cite{pennylane}, we computed the dimensions of the standard DLAs and the standard reduced DLAs corresponding to each vertex of each asymmetric graph on $6$ nodes. The results are summarized in Figure~\ref{SixNodeAsymmGraphs}. 

\begin{rmk}\label{rmk:dimension_su_32}
By definition, for an $n$-node graph $\Gamma$ and a vertex $v\in V(\Gamma)$, the standard reduced DLA 
\(
\mathfrak{g}^{v}_{\Gamma,\mathrm{std}}
\)
is a Lie subalgebra of 
\[
\mathfrak{su}(2^{n-1}) \;=\; \mathfrak{su}(W_v),
\]
where $W_v$ is a Hilbert space of dimension $2^{n-1}$. 
In particular, for $n=6$ we have
\[
\dim \mathfrak{su}(2^{5}) = 4^{5} - 1 = 1023.
\]
Therefore, any algebra whose dimension entry equals $1023$ in Figure~\ref{SixNodeAsymmGraphs} necessarily coincides with the full algebra $\mathfrak{su}(32)$ (up to Lie algebra isomorphism).
\end{rmk}

Similarly, we computed the dimensions of the standard and reduced dynamical Lie algebras for all $144$ asymmetric graphs (up to isomorphism) on $7$ nodes. The results are presented in Table~\ref{tab:SevenNodesDimensions}. The patterns observed in the $6$- and $7$-vertex cases suggest a consistent trend in how the dimensions of reduced DLAs behave relative to the full standard DLA, and motivate the following conjecture.

\begin{conj}
\label{mainConj}
    Let $\Gamma$ be an asymmetric graph. Then there exists a vertex $v \in V(\Gamma)$ such that the dimension of the corresponding standard reduced dynamical Lie algebra is strictly smaller than that of the full standard DLA:
    \[
        \dim \mathfrak{g}^{\,v}_{\Gamma,\mathrm{std}}
        \;<\;
        \dim \mathfrak{g}_{\Gamma,\mathrm{std}}.
    \]
\end{conj}
This conjecture has been verified for all asymmetric graphs on $6$ and $7$ vertices (see Figure~\ref{SixNodeAsymmGraphs} and Table~\ref{tab:SevenNodesDimensions}).

\begin{rmk}
Conjecture~\ref{mainConj} may fail for certain graphs with nontrivial groups of automorphisms.  For example, it does not hold for cyclic graphs with $5$ or more vertices (see~\eqref{eq:cyclicGraphsDimInequality}).
\end{rmk}

\begin{table}[H]
\centering
\begin{tabular}{|c|c|c|c|c|c|c|c|c|}
\hline
$i$ &
$\dim\bigl(\mathfrak{g}_{\Gamma_i,\mathrm{std}}\bigr)$ & $\dim\bigl(\mathfrak{g}^{v_1}_{\Gamma_i,\mathrm{std}}\bigr)$ & $\dim\bigl(\mathfrak{g}^{v_2}_{\Gamma_i,\mathrm{std}}\bigr)$ & $\dim\bigl(\mathfrak{g}^{v_3}_{\Gamma_i,\mathrm{std}}\bigr)$ &
$\dim\bigl(\mathfrak{g}^{v_4}_{\Gamma_i,\mathrm{std}}\bigr)$ & $\dim\bigl(\mathfrak{g}^{v_5}_{\Gamma_i,\mathrm{std}}\bigr)$ & $\dim\bigl(\mathfrak{g}^{v_6}_{\Gamma_i,\mathrm{std}}\bigr)$
& $\dim\bigl(\mathfrak{g}^{v_7}_{\Gamma_i,\mathrm{std}}\bigr)$\\
\hline
$1$ & 1808 & 477 & 131 & 1122 & 1072 & 1183 & 521 & \color{red}{33}\\
\hline
$2$ & 4400 & 1025& 1979& 1816& 2099& 2303& 843& \color{red}{75} \\
\hline
$3$ & $ > 3260$ & 2045& 2694& 2775& 3203& 2844& 993& \color{red}{911} \\
\hline
$4$ &  $> 3890$ & 3135& 3079& 3668& 3438& 2976& \color{red}{1026}& \color{red}{1026}\\
\hline
$5$ & $> 2953$ & \color{red}{265}& 1929& 2642& 2277& 2047& 1026& 833 \\
\hline
$6$ & $> 4396$ & 2141& 3627& 2606& 3265& 1024& 2284& \color{red}{963} \\
\hline
$7$ & $> 3145$  & 2131& 2503& 2047& 2708& 1584& 2517& \color{red}{675} \\
\hline
$8$ & $> 4272$  & 2538& 3053& 3584& 3227& 1916& 1023& \color{red}{906} \\
\hline
$9$ & $>4244$ & 2007& 2738& 3450& 3475& 3525& 3169& \color{red}{872} \\
\hline
$10$ & $> 4532$ & 2884& 2834& 3297& 3785& 3466& 2653& \color{red}{1024}\\
\hline
$11$ & 2891 & 1615& 1610& 265& 1855& 697& 919& \color{red}{84}\\
\hline
$12$ & $>4212$ & 2623& 3727& 3609& 3364& 2818& 2861& \color{red}{1026} \\
\hline
$13$ &  $> 4742$ & 3746& 3592& 3432 & 4038 & 3547 & 3269 & \color{red}{1027}\\
\hline
$14$ & $> 4422$ & 1819& 2506& 3032& 3698& 2562& \color{red}{1027}& 1891\\
\hline
$15$ & $>4434$ & 1974 & 2968 & 2499 & 3985 & 2503 & \color{red}{1027} & 1788\\
\hline
$16$ & $>4796$ & 3730 & 3149& 3304& 3970& 3044& \color{red}{1027}& 2508 \\
\hline
$17$ & $4320$ & 2609 & 2576& 3615& 3436& 2306& 1012& \color{red}{977}\\
\hline
$18$ & $>4846$ & 2682 & 3712& 2909& 3474& \color{red}{1025}& 2580& 1026 \\
\hline
$19$ & $> 4208$ & 3621 & 2931 & 3404 & 4051 & 1849 & 3417 & \color{red}{1027} \\
\hline
$20$ & $>4716$ & 3391 & 3408 & 3855 & 3397 & 3421 & 3123 & \color{red}{1027} \\
\hline
$21$ & 3733 & 3042 &  2821 & 3676 & 3799 & 2847 & 1026 & \color{red}{935}\\
\hline
$22$ & $>3903$ & 3939 & 3415 & 3300 & 3972 & 2763 & \color{red}{1027} &  2988 \\
\hline
$23$ & $>4735$ & 3612 & 3831 & 3044 & 3840& 2738 & \color{red}{1027} & 2336\\
\hline
$24$ & 2594 & 511 & 532 & 1366 & 1417 & 1500 & 1469& \color{red}{131} \\
\hline
$25$ & $> 3766$ & 1923 & 1027 & 2091 & 2237 & 3098 & 1756 & \color{red}{266} \\
\hline
$26$ & $>3017$ & 826 & 1019 & 1936 & 2191 & 2920 & 2071 & \color{red}{266} \\
\hline
$27$ & $>4358$ & 2547 & \color{red}{1027} & 3194 & 3062 & 3785 & 2117 & 1767 \\
\hline
$28$ & $>4222$ & 2095 & 3467 & 3715 & 3843 & 3167 & 2277 & \color{red}{1809}\\
\hline
$29$ & $>4803$ & 2560 & 1026 & 3587 & 2705 & 3330 & \color{red}{1002} & 2434 \\
\hline
$30$ & $>3858$ & 2141 & 1027 & 1954 & 1816 & 2832 & \color{red}{266} & 2148\\

\hline
\end{tabular}
\caption{Dimensions of the standard and reduced dynamical Lie algebras for asymmetric graphs on $7$ nodes. The minimal dimension values are highlighted in red.}
\label{tab:SevenNodesDimensions}
\end{table}

\subsubsection{Variance-Based Proxies for DLA Dimensions of Standard and Reduced QAOA}

As it appears infeasible to establish the precise dimensions of standard and reduced dynamical Lie algebras corresponding to QAOAs for MaxCut problems on arbitrary asymmetric graphs with a larger number of nodes, we instead exploit the connection between the variance of the QAOA loss function and the dimension of the corresponding DLA. Broadly speaking, the variance is inversely proportional to the DLA dimension; thus, larger variance signals a smaller-dimensional DLA. 

These results are also of independent interest, since analyzing the variance is essential for understanding the emergence of barren plateaus in the optimization landscape \cite{larocca2025barren, RBSKMLC}. The phenomenon of \emph{barren plateaus} refers to the exponential suppression of gradients in variational quantum algorithms, which renders classical training ineffective \cite{mcclean2018barren, larocca2025barren}. A key diagnostic of this behavior is the variance of the loss function over parameter space: when this variance is exponentially small in the number of qubits \( n \), the corresponding gradients also vanish exponentially \cite{arrasmith2021equivalence}.

We now recall this connection more explicitly. The \emph{loss function} in QAOA is defined as the expectation value of the problem Hamiltonian \( H_P \) with respect to the parametrized QAOA state:
\begin{equation}
\begin{aligned}
\ell_{\boldsymbol{\beta}, \boldsymbol{\gamma}} (\rho, H_P) 
    :=&\, \bra{\psi(\boldsymbol{\beta},\boldsymbol{\gamma})} H_P \ket{\psi(\boldsymbol{\beta},\boldsymbol{\gamma})} \\
    =&\, \operatorname{Tr} \!\left[ U(\boldsymbol{\beta}, \boldsymbol{\gamma}) \, \rho \, U^\dagger(\boldsymbol{\beta}, \boldsymbol{\gamma}) \, H_P \right],
\end{aligned}
    \label{LossFnEqn}
\end{equation}
where \( \rho = \ket{\xi}\bra{\xi} \) is the pure initial state and 
\( U(\boldsymbol{\beta}, \boldsymbol{\gamma}) \) is the parameterized QAOA unitary defined in~\eqref{qaoa-chain}.  

Due to Theorems~2 and~3 of~\cite{RBSKMLC}, for sufficiently deep QAOA circuits (i.e., for large enough value of \( p \)), the variance of the loss function over parameters \( (\boldsymbol{\beta}, \boldsymbol{\gamma}) \) can be approximated by the variance taken over the dynamical Lie group \( G = e^{\mathfrak{g}} \). The latter was computed in Theorem~1 of~\cite{RBSKMLC}, which we now recall.

Let the DLA decompose as
\[
  \mathfrak{g} 
  = \mathfrak{g}_1 \oplus \cdots \oplus \mathfrak{g}_k \oplus \mathfrak{z},
\]
into simple compact Lie algebras \( \mathfrak{g}_j \) together with a (possibly trivial) center \( \mathfrak{z} \) 
(this decomposition exists by Proposition~A.1 in~\cite{WKKB1}).  

Then the variance of the loss function over the group \( G \) is given by
\begin{equation}
\label{eq:VarAndDim}
   \operatorname{Var}_{G}
  \bigl[\ell_{\boldsymbol{\beta}, \boldsymbol{\gamma}} (\rho, H_P) \bigr] 
  = \sum\limits_{j=1}^{k} 
    \frac{\mathcal{P}_{\mathfrak{g}_j}(\rho)\, \mathcal{P}_{\mathfrak{g}_j}(H_P)}
         {\dim(\mathfrak{g}_j)}.
\end{equation}
In the above formula,
$\mathcal{P}_{\mathfrak{g}_j}(H) := \operatorname{Tr}\!\bigl[H_{\mathfrak{g}_j}^2\bigr]$
and \( H_\mathfrak{g_j} \) denotes the orthogonal projection of an operator \(H\) onto the subalgebra \( \mathfrak{g}_j \subseteq \mathrm{End}(W)\).  
\begin{rmk}
Two remarks are in order.  
First, formula~\eqref{eq:VarAndDim} depends on the dimensions of the \emph{simple components} \( \mathfrak{g}_j \) rather than on the total dimension of the DLA itself. Thus, the relation between the variance and the full DLA structure can be subtle: the variance captures a weighted contribution from each simple summand rather than a single global quantity.  

Second, the variance of the loss function provides an indicator of whether a given QAOA instance is susceptible to barren plateaus. Therefore, it can be used as an operational diagnostic of the comparative performance of the standard and reduced QAOA circuits for the same optimization problem.
\end{rmk}

Below, we present numerical computations of the variance for both standard and reduced QAOA circuits for the MaxCut problem on a variety of connected asymmetric graphs with $7$ and $14$ vertices. The variance of the gradient is computed as the mean over 400 initializations of the parametrized circuit corresponding to a MaxCut instance, with each initialization assigning randomized variational parameters between $0$ and $\pi$, with a total of $2p$ parameters, where $p$ is the circuit depth. All gradients were computed for a range of circuit depths using Pennylane's QNode with the \textit{autograd} interface and a \textit{backprop} differentiation method, which resulted in the fastest computation of the gradient.

\afterpage{\clearpage
    \begin{figure}
        \centering
        \includegraphics[width=2\paperwidth, height=0.8\paperheight, keepaspectratio]{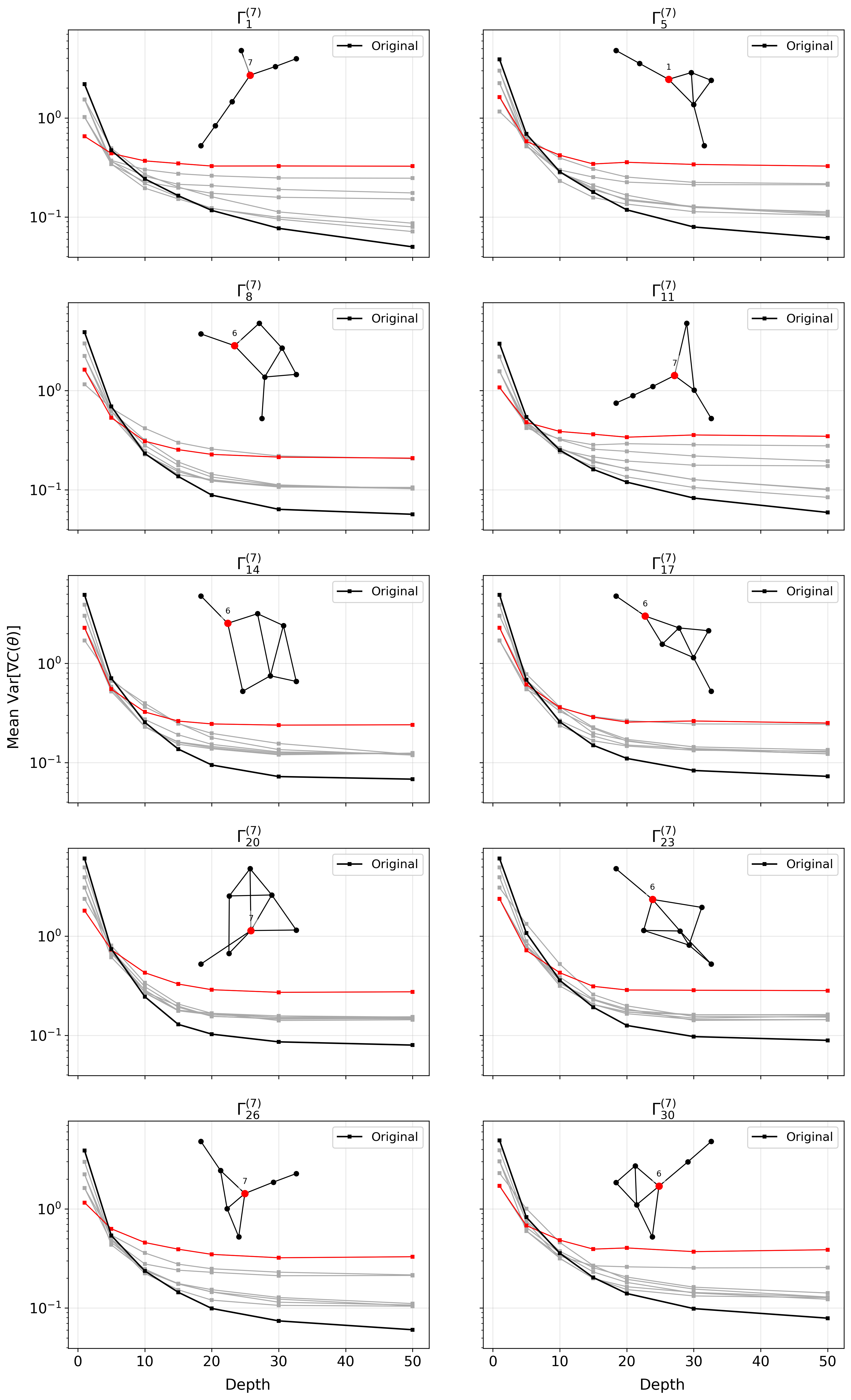}
        \caption{Mean of variance for selected 7-node asymmetric graphs ($\Gamma^{(7)}$). The red curve shows the vertex reduction that results in the largest variance of the gradient at depth 50, with the corresponding vertex highlighted in red and labeled on the graph inset.}
        \label{fig:7_node_variances}
    \end{figure}
    \clearpage
}

Figure \ref{fig:7_node_variances} shows the variance of the gradient for a select number of 7-node asymmetric graphs. Upon comparison to Table \ref{tab:SevenNodesDimensions}, we see that there is good correspondence between the DLA dimensions and the magnitude of the variance in terms of their inverse relationship. In many cases, there is a specific vertex reduction (shown in red) that simplifies the problem significantly compared with the full (standard) problem DLA dimensions. Furthermore, for any vertex reduction (shown in gray) there is a corresponding increase in the variance of the gradient compared to the original problem, for large enough depths. For a circuit depth of $1$, in all cases, we see a larger value for the variance in the standard problem. We also note that around a circuit depth of $5$, the variance of the gradient for the original problem and that for the reduced version cross paths. This behavior, while noteworthy, is beyond the scope of this article. 

In practice, computing the variance of the gradient allows us to determine a good candidate for vertex reduction, especially for cases where computing the DLA dimensions becomes intractable. To demonstrate this practical application we compute the variance for ten 14-node asymmetric graphs. 

\afterpage{\clearpage
    \begin{figure}
        \centering
        \includegraphics[width=2\paperwidth, height=0.8\paperheight, keepaspectratio]{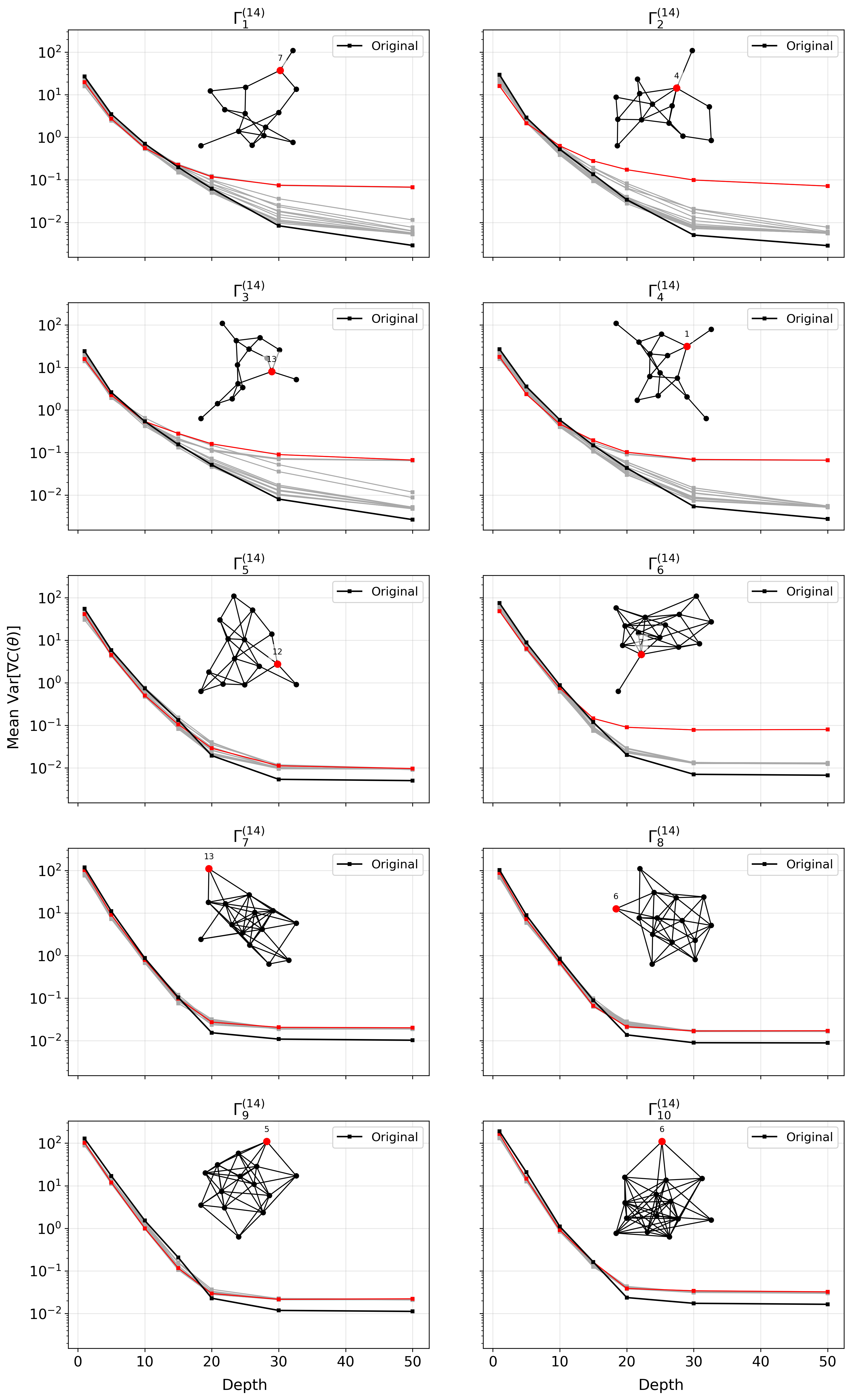}
        \caption{Mean of variance for ten 14-node asymmetric graphs ($\Gamma^{(14)}$). The red curve shows the vertex reduction that results in the largest variance of the gradient at depth 50, with the corresponding vertex highlighted in red and labeled on the graph inset.}
        \label{fig:14_node_variances}
    \end{figure}
    \clearpage
}

Figure \ref{fig:14_node_variances} shows the mean of the variance of the gradient for increasing circuit depth on ten 14-node graphs. The vertex reduction that results in the largest variance is highlighted in red. Here we notice that as the graphs become denser\footnote{These 14-node graphs were generated using a range of edge-creation probabilities, resulting in asymmetric graphs of increasing edge density.}, most vertex reductions result in a small difference in the variance of the gradient with respect to the original problem. Nevertheless, in all cases, any vertex reduction results in a larger variance of the gradient compared with the standard problem without any loss of quality in the solution.

\begin{obs}
\label{leaf_obs}
    In general, we see that the vertex reduction that results in the largest magnitude (at depth 50) of the variance of the gradient corresponds to a vertex that contains a leaf. This can be evidenced in Figure \ref{fig:added_leafs}: when a leaf is added artificially to a vertex in the graph, there is an increase in the variance of the gradient compared to the original problem instance (no leaves).   
\end{obs}

\begin{figure}[h!]
    \centering
    \includegraphics[width=\linewidth, keepaspectratio]{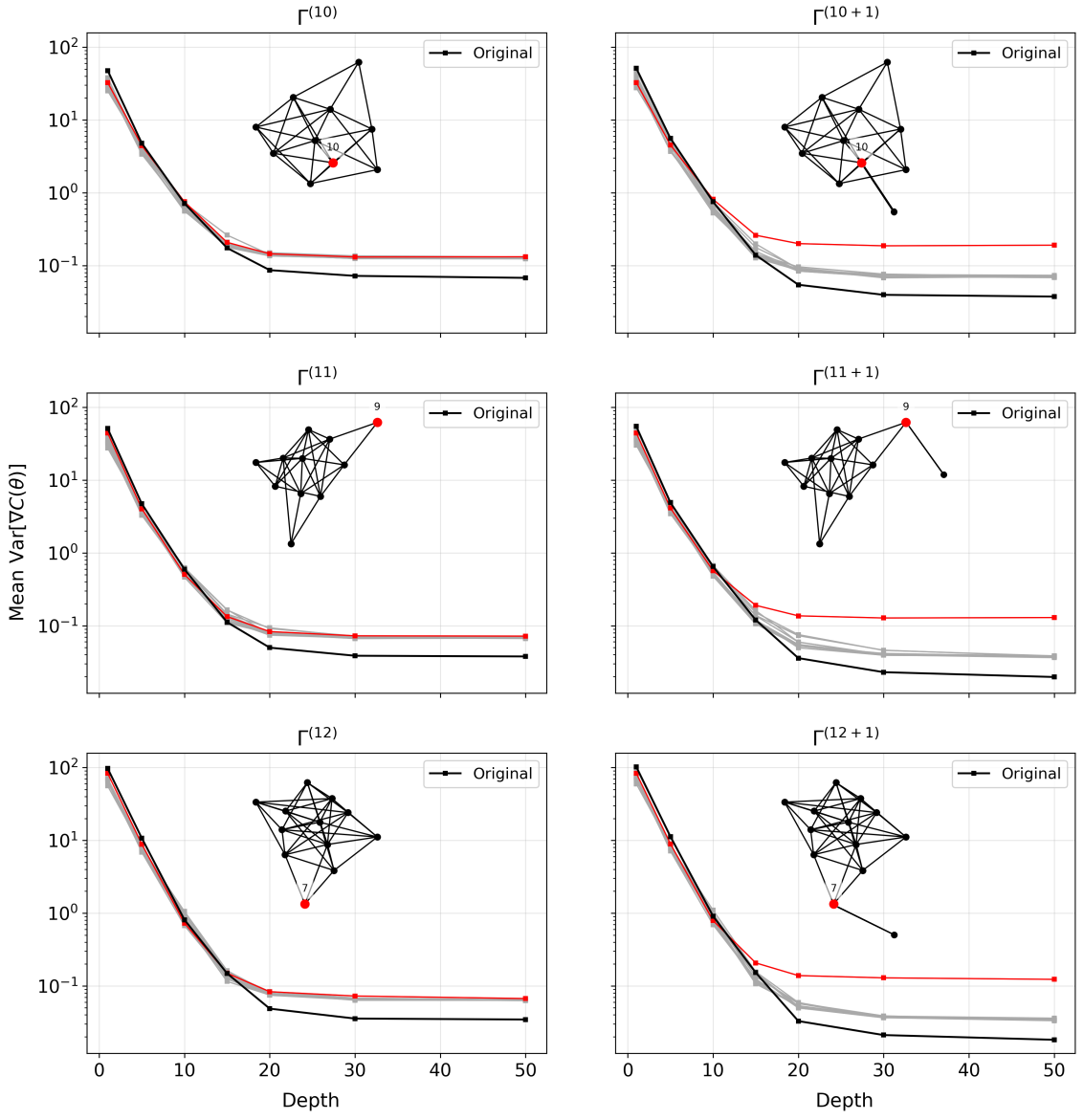}
    \caption{Effect of artificially introducing a leaf. The graphs on the left column ($\Gamma^{(n)}$) contain no leaves and any vertex reduction results in similar magnitudes in the variance of the gradient. The graphs on the right column ($\Gamma^{(n+1)}$) show an increase in the variance of the gradient when a leaf is artificially introduced to a vertex. 
    }
    \label{fig:added_leafs}
\end{figure}

\begin{obs}
\label{leaf_improvement}
    Remarkably, as the variance of the gradient increases with the graph size (black curve in Figure \ref{fig:added_leafs}) for the original problem instances, there is an improvement in the variance with the artificial addition of a leaf for the reduced case. Furthermore, since a leaf is added to the fixed vertex, the quality of the solution is not jeopardized. Similarly, we see that vertices with similar vertex connectivity result in similar magnitudes of the variance of the gradient upon reduction.
\end{obs}


Refer to Appendix D for additional numerical results for graphs ranging from $11 - 15$ nodes. Code availability for all numerical results can be found here: \url{https://github.com/joseluisfalla/QAOA_Reductions_by_Classical_Symmetries.git}.

\subsection{Star Graphs} 
\label{sec:star_graph}
Consider the star graph $K_{1,n}$ consisting of a central vertex,
denoted by $0$, and $n$ leaf vertices labeled
$1,\ldots,n$. Each leaf vertex is connected to the central
vertex by a single edge, and there are no edges among the leaves.
See Figure~\ref{starGraph} for an illustration in the case $n=5$.

\begin{figure}[H]
\begin{center}
\begin{tikzpicture}[scale=1.1]
  \node[circle, draw, fill=black, inner sep=2pt, label={left:$0$}] (0) at (0,0) {};
  
  \foreach \i in {1,...,5} {
    \node[circle, draw, fill=black, inner sep=2pt, label={\i}] (\i) at ({72*(\i-1)}:2) {};
    \draw (0) -- (\i);
  }
\end{tikzpicture}
\end{center}
\caption{The star graph $K_{1,5}$ with center vertex $0$ and outer vertices $1,\dots,5$.}
\label{starGraph}
\end{figure}
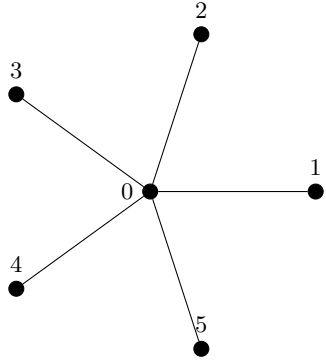

The standard dynamical Lie algebra associated to $K_{1,n}$, denoted 
$\mathfrak{g}_{K_{1,n},\mathrm{std}}$, is generated by the two operators
\[
   X \;=\; iX_0 + iX_1 + \cdots + iX_n,
   \qquad
   Z \;=\; iZ_0 Z_1 + \cdots + iZ_0 Z_n.
\]
Although we did not obtain a closed-form expression for 
$\dim( \mathfrak{g}_{K_{1,n},\mathrm{std}})$ as a function of $n$, 
we computed the dimensions numerically for several values of $n$.
These values are presented in Table~\ref{tab:starDimensions}.

\begin{table}[H]
\centering
\begin{tabular}{|c|c|c|c|c|c|c|c|c|c|c|}
\hline
$n$ & 2 & 3& 4& 5& 6 &7 &8 & 9 & 10 & 11  \\
\hline
$\dim\bigl(\mathfrak{g}_{K_{1,n},\mathrm{std}}\bigr)$ & 9 & 24 & 33 & 66 & 81 & 138 & 161 & 247 & 281 & 400  \\
\hline
\end{tabular}
\caption{Numerically computed dimensions of the standard dynamical Lie algebra 
$\mathfrak{g}_{K_{1,n},\mathrm{std}}$ for several values of $n$.}
\label{tab:starDimensions}
\end{table}

\begin{rmk}
The \emph{orbit Lie algebra} $\mathfrak{g}_{\Gamma,\mathrm{orb}}$ of a graph $\Gamma$ is defined as the Lie algebra generated by the $\operatorname{Aut}(\Gamma)$–orbit sums of the single-qubit and two-qubit Pauli operators, namely
\[
\left\{ i\sum\limits_{g \in \operatorname{Aut}(\Gamma)} X_{g(w)} \;\middle|\; w \in V(\Gamma) \right\}
\qquad\text{and}\qquad
\left\{ i\sum\limits_{g \in \operatorname{Aut}(\Gamma)} Z_{g(w)} Z_{g(w')} \;\middle|\; (ww') \in E(\Gamma) \right\}
\]
(see Section VI in \cite{KLFCCZ}).

 The automorphism group of the star graph $K_{1,n}$ is the group $S_n$, which acts by permuting the $n$ leaf vertices, while fixing the central vertex. Thus there are two vertex–orbits: the singleton orbit of the center, and the orbit formed by the leaves. Consequently, the corresponding orbit-sum generators are
\[
iX_0, \qquad i\sum\limits_{j=1}^n X_j,
\]
and for the edges (all of which are of the form $(0,i)$), the unique orbit-sum operator
\[
i\sum\limits_{j=1}^n Z_0 Z_j.
\]

When $n$ is even, the degrees of the vertices modulo $2$ satisfy:
\[
\deg(0) = n \equiv 0 \pmod{2}, \qquad 
\deg(j) = 1 \equiv 1 \pmod{2} \quad \text{for all } j \ge 1.
\]
Thus the center and the leaves lie in \emph{distinct degree-parity classes}. By Theorem~\ref{thm:CoolEltsInStdDLA}, this implies that the standard dynamical Lie algebra $\mathfrak{g}_{K_{1,n},\mathrm{std}}$ contains all orbit generators of $\mathfrak{g}_{K_{1,n},\mathrm{orb}}$. Hence, for even $n$, we obtain
\[
\mathfrak{g}_{K_{1,n},\mathrm{std}}
  \;=\;
\mathfrak{g}_{K_{1,n},\mathrm{orb}}.
\]
For odd values of $n$, this inclusion may fail. A concrete example is $K_{1,3}$, where
\[
\dim\bigl(\mathfrak{g}_{K_{1,3},\mathrm{std}}\bigr)=24,
\qquad
\dim\bigl(\mathfrak{g}_{K_{1,3},\mathrm{orb}}\bigr)=26.
\]
\end{rmk}

The standard reduced  dynamical Lie algebra at the central vertex $v=0$, 
denoted $\mathfrak{g}^0_{K_{1,n},\mathrm{std}}$, is generated by the elements
\begin{equation}
   \label{eq:GeneratorsOfReducedStarGraph}
    \mathcal{X}_{\widehat{0}} \;:=\; iX_1 + \cdots + iX_n, 
   \qquad 
   \mathcal{Z}_{\widehat{0}} \;:=\; iZ_1 + \cdots + iZ_n,
\end{equation}
and is isomorphic to the three-dimensional simple Lie algebra $\mathfrak{su}(2)$. 

\begin{rmk}
The standard reduced DLA $\mathfrak{g}^0_{K_{1,n},\mathrm{std}}$ is thus independent of $n$ and has dimension $3$, in stark contrast to the full standard DLA $\mathfrak{g}_{K_{1,n},\mathrm{std}}$, whose dimension grows with $n$ (see Table~\ref{tab:starDimensions}).  

While dimensions of $\mathfrak{g}_{K_{1,n},\mathrm{std}}$ can be computed explicitly for small $n$, a closed-form expression for its dimension, as well as a uniform description of its decomposition into simple Lie algebra components and its center, remains nontrivial. This highlights the simplifying effect of the vertex-reduction procedure on the DLA structure.
\end{rmk}

\subsection{Path Graphs}

Consider the path graph $P_n$, illustrated in Figure~\ref{pathGraph}.

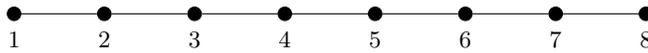
\begin{figure}[h!]
\begin{center}
\def\nn{8} 
\begin{tikzpicture}[scale=1, every node/.style={circle, draw, fill=black, inner sep=1.8pt}]
  \foreach \i in {1,...,\nn} {
    \node (v\i) at ({1.2*(\i-1)},0) {};
  }
  \foreach \i in {1,...,\nn} {
    \pgfmathtruncatemacro\ip{int(\i+1)}
    \ifnum\ip>\nn\else
      \draw (v\i) -- (v\ip);
    \fi
  }
  \foreach \i in {1,...,\nn} {
    \node[draw=none, fill=none] at ({1.2*(\i-1)},-0.35) {$\i$};
  }
\end{tikzpicture}
\caption{The path graph $P_8$ with $8$ vertices.}
\label{pathGraph}
\end{center}
\end{figure}

It was shown in \cite{KLFCCZ} that the free dynamical Lie algebra associated with $P_n$ satisfies
\[
\mathfrak{g}_{P_n,\mathrm{free}} \cong \mathfrak{so}(2n),
\]
and therefore
\[
\dim\!\left(\mathfrak{g}_{P_n,\mathrm{free}}\right) = 2n^2 - n.
\]
Moreover, the standard and orbit dynamical Lie algebras coincide and are isomorphic to the unitary Lie algebra,
\[
\mathfrak{g}_{P_n,\mathrm{std}}
=
\mathfrak{g}_{P_n,\mathrm{orb}}
\cong
\mathfrak{u}(n),
\]
as established in Theorem~1 and Lemma~A.6 of \cite{KLFCCZ}. In particular,
\[
\dim\!\left(\mathfrak{g}_{P_n,\mathrm{std}}\right)
=
\dim\!\left(\mathfrak{g}_{P_n,\mathrm{orb}}\right)
=
n^2.
\]

\medskip

Let $v = k \in V(P_n)$ be a fixed vertex.  
Then 
Remark~\ref{rmk:ExtensionOfResults} together with
Theorem~\ref{thm:DistDLAElements} imply the following containments:
\begin{equation}
\label{eq:PathContainments}
\mathfrak{g}^{\,v}_{P_n,\mathrm{std}}
\;\supseteq\;
\begin{cases}
\mathfrak{g}_{P_{n-1},\mathrm{free}},
& k = 1 \text{ or } k = n, \\[6pt]
\mathfrak{g}_{P_{(n-1)/2},\mathrm{free}},
& k = \frac{n+1}{2}, \; n \equiv 1 \pmod{2}, \\[6pt]
\mathfrak{g}_{P_{k-1},\mathrm{free}}
\;\oplus\;
\mathfrak{g}_{P_{n-k},\mathrm{free}},
& \text{otherwise}.
\end{cases}
\end{equation}

\medskip

In particular, if $v$ is a leaf vertex (i.e., $k=1$ or $k=n$), then
\[
\mathfrak{g}^{\,v}_{P_n,\mathrm{std}}
\;\supseteq\;
\mathfrak{g}_{P_{n-1},\mathrm{free}},
\]
and therefore
\[
\dim\!\left(\mathfrak{g}^{\,v}_{P_n,\mathrm{std}}\right)
\;\ge\;
\dim\!\left(\mathfrak{g}_{P_{n-1},\mathrm{free}}\right)
=
2(n-1)^2 - (n-1).
\]

For all $n \ge 5$, one verifies that
\[
2(n-1)^2 - (n-1) > n^2,
\]
which shows that the dimension of the standard reduced dynamical Lie algebra at a leaf vertex strictly exceeds the dimension of the standard dynamical Lie algebra of the original path graph.


\subsection{Cycle Graphs}

We now consider the cycle graphs \(C_n\), illustrated in Figure~\ref{cycleGraph}.

\begin{figure}[h!]
\begin{center}
\def\nn{7} 
\begin{tikzpicture}[scale=1.2, every node/.style={circle, draw, fill=black, inner sep=1.6pt}]
  \foreach \i in {1,...,\nn} {
    \node (v\i) at (90+\i*360/\nn:2) {};
  }
  \foreach \i in {1,...,\nn} {
    \pgfmathtruncatemacro\ip{int(mod(\i, \nn) + 1)}
    \draw (v\i) -- (v\ip);
  }
  \foreach \i in {1,...,\nn} {
    \node[draw=none, fill=none] at (90+\i*360/\nn:2.3) {$\i$};
  }
\end{tikzpicture}
\caption{The cycle graph \(C_7\) with \(7\) vertices.}
\label{cycleGraph}
\end{center}
\end{figure}
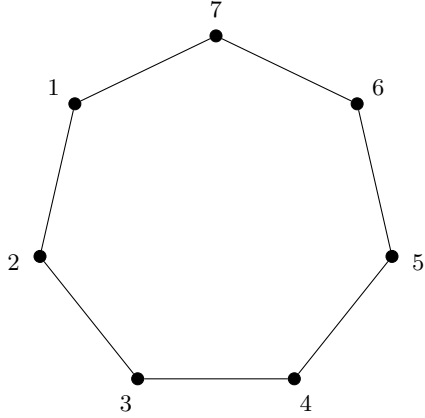

It follows from Theorem~1 in \cite{KLFCCZ} that the free standard dynamical Lie algebra associated to a cycle graph satisfies
\[
\mathfrak{g}_{C_n,\mathrm{free}} \cong \mathfrak{so}(2n) \oplus \mathfrak{so}(2n).
\]

In addition, Theorem~33 in \cite{ASYZ1} gives an explicit decomposition:
\[
\mathfrak{g}_{C_n,\mathrm{std}} 
\cong 
\underbrace{\mathfrak{su}(2) \oplus \cdots \oplus \mathfrak{su}(2)}_{n - 1 \text{ copies}}
\;\oplus\; \mathfrak{u}(1)^{\oplus 2}.
\]
In particular,
\[
\dim\!\bigl(\mathfrak{g}_{C_n,\mathrm{std}}\bigr)
=3(n-1)+2 = 3n - 1.
\]

According to Lemma~8 in \cite{KLFCCZ}, the natural Lie algebra \(\mathfrak{g}_{C_n,\mathrm{nat}}\) associated to the cycle graph also satisfies
\[
\dim(\mathfrak{g}_{C_n,\mathrm{nat}})=3n-1.
\]

We always have the containments
\[
\mathfrak{g}_{C_n,\mathrm{std}}
\;\subseteq\;
\mathfrak{g}_{C_n,\mathrm{orb}}
\;\subseteq\;
\mathfrak{g}_{C_n,\mathrm{nat}}.
\]
Since all three Lie algebras have the same dimension, they must coincide. Therefore,
\[
\mathfrak{g}_{C_n,\mathrm{std}}
=\mathfrak{g}_{C_n,\mathrm{orb}}
=\mathfrak{g}_{C_n,\mathrm{nat}}
\;\cong\;
\underbrace{\mathfrak{su}(2) \oplus \cdots \oplus \mathfrak{su}(2)}_{n - 1 \text{ copies}}
\;\oplus\; \mathfrak{u}(1)^{\oplus 2}.
\]

   We turn our attention to the reduced DLAs. Let $v \in C_n$ be any vertex of the cycle graph $C_n$. Consider the corresponding standard reduced  dynamical Lie algebra $\mathfrak{g}^v_{C_n,\mathrm{std}}$. By Lemma \ref{lem:ReducedSubalgebras}, this algebra contains the standard DLA of the path graph obtained by removing $v$ and its two incident edges from $C_n$, that is,
\[
\mathfrak{g}_{P_{n-1},\mathrm{std}} \subseteq \mathfrak{g}^v_{C_n,\mathrm{std}}.
\]
As a consequence, we have the dimension bound
\[
\dim\big(\mathfrak{g}^v_{C_n,\mathrm{std}}\big) \ge \dim\big(\mathfrak{g}_{P_{n-1},\mathrm{std}}\big) = (n-1)^2,
\]
which exceeds the dimension of the standard DLA of the cycle graph:
\[
\dim\big(\mathfrak{g}_{C_n,\mathrm{std}}\big) = 3n-1,
\]
for all $n \ge 5$, i.e.

\begin{equation}
    \label{eq:cyclicGraphsDimInequality}
\dim\big(\mathfrak{g}^v_{C_n,\mathrm{std}}\big)>\dim\big(\mathfrak{g}_{C_n,\mathrm{std}}\big) \text{ for } n>5
\end{equation}

\section{Reduced DLAs for the Grover-mixer QAOA}
\label{sec:ReducedGMQAOA}

In this subsection we compare the reduced and standard dynamical Lie algebras arising in the variant of QAOA that employs the 
\emph{Grover mixer}, defined as the  projector onto the initial state:
\begin{equation}\label{eq:GM}
    G_M \;:=\; \ket{\xi}\bra{\xi}
    \;=\;
    H^{\otimes n} \ket{0\cdots 0}\bra{0\cdots 0} H^{\otimes n},
\end{equation}
where $\ket{\xi} = \ket{+\cdots+}$ is taken to be the uniform superposition of all computational basis states.

We briefly recall the setup and refer the reader to~\cite{BE} and~\cite{TNB} for a more detailed discussion. We consider a general combinatorial optimization problem of minimizing an objective function $F$ of the form given in~\eqref{eq:optfunc}.  The variant of QAOA, employing the mixer Hamiltonian $G_M$ from \eqref{eq:GM} was introduced in~\cite{BE} and is known as the \emph{Grover-mixer QAOA} (GM-QAOA).

Notice that  the range of the reduced objective function $\widetilde{F}(x)$ coincides with the range of $F(x)$ (each value appears with multiplicity reduced by a factor of two). Consequently, by Theorem~III.1 of~\cite{TNB}, the dynamical Lie algebra $\mathfrak{g}_\xi$ generated by the GM-QAOA Hamiltonians, together with all of its bit-reduced versions $\mathfrak{g}^b_\xi$ (obtained by fixing any bit $b$ to the value $0$ or $1$), are all isomorphic to
\begin{equation}
    \label{eq:GroverDLAdim}
    \mathfrak{su}(r) \oplus \mathfrak{u}(1) \oplus \mathfrak{u}(1),
\end{equation}
    
where $r$ denotes the number of distinct objective values attained by $F(x)$ and $\widetilde{F}(x)$.

We now look at the situation for QAOA implementing Grover search as in~\cite{jiang2017near}. Consider the case of searching for a single marked element $\omega \in \mathbb{B}^n$. The problem Hamiltonian is the rank-one projector
\[
    H_P = \ket{\omega}\bra{\omega},
\]
while the mixer $H_M$ is chosen to be the standard $X$-mixer $B=\sum\limits_{j=1}^n X_j$. These generate the \emph{standard dynamical Lie algebra}
\[
    \mathfrak{g}_{\mathrm{std}}
    := 
    \langle iB,\; iH_P\rangle_{\mathrm{Lie}}.
\]
By Corollary~V.3 of~\cite{TNB}, this algebra is isomorphic to

\begin{equation}
    \label{eq:GroverDLAdim2}
    \mathfrak{su}(n+1) \oplus \mathfrak{u}(1) \oplus \mathfrak{u}(1),
\end{equation}
while all corresponding standard reduced algebras $\mathfrak{g}^b_{\mathrm{std}}$ (obtained by restricting the bit $b$ to a fixed bit value) satisfy
\[
    \mathfrak{g}^b_{\mathrm{std}}
    \;\cong\;
    \mathfrak{su}(n) \oplus \mathfrak{u}(1) \oplus \mathfrak{u}(1).
\]

\begin{rmk}
The above results highlight a distinctive feature of the Grover-mixer version of QAOA: 
its standard reduced dynamical Lie algebras are isomorphic to the original ones. In contrast, for QAOA with the standard $X$-mixer, bit reductions can alter the structure of the DLA in many different ways. 
\end{rmk}

\section{Action on Hilbert Spaces}
\label{sec:HilbSpaces}

The original Hilbert space $W$, as well as the reduced spaces $W_j$, admit decompositions into isotypic components under the action of the respective dynamical Lie algebras. The dimensions of these components are central to the analysis of classical simulability~\cite{CLG,GLCCS} and constitute the primary object of study in this section.

\begin{thm}
\label{DecompThm}
  Let $\Gamma$ be a connected graph on $n$ vertices.  Then the Hilbert space $W$ admits a decomposition
\begin{equation}
    \label{eq:EvenOddSpaces}
    W = W_{\mathrm{even}} \oplus W_{\mathrm{odd}},
\end{equation}
where $W_{\mathrm{even}}$ is spanned by all computational basis states that contain an even number of $\ket{-}$ factors, and $W_{\mathrm{odd}}$ is spanned by those containing an odd number of $\ket{-}$ factors.

Moreover, these subspaces have equal dimension,
\[
\dim W_{\mathrm{even}} = \dim W_{\mathrm{odd}} = 2^{n-1}.
\]
This decomposition yields a direct sum of two irreducible representations of the dynamical Lie algebra $\mathfrak{g}_{\Gamma,\mathrm{free}}$.
\end{thm}

\begin{rmk}
    A similar result in \cite{KLFCCZ} was previously established for \textit{archetypal} graphs, connected graphs that are neither bipartite nor cycles.
\end{rmk}

\begin{rmk}
\label{rmk:GlobalParity}
    The representations \( W_{\text{even}} \) and \( W_{\text{odd}} \) appearing in Theorem \ref{DecompThm} correspond to the trivial and sign representations of the group \( \mathbb{Z}_2 \), respectively. Here, \( \mathbb{Z}_2 \) acts via the global Pauli-\(X\) operator \( X_1 \cdots X_n \), which flips all qubits simultaneously. 
\end{rmk}

Theorem~\ref{DecompThm} admits the following analogue for reduced DLAs.  

\begin{thm}
    Let $\Gamma$ be a connected graph and $v$ a vertex in $\Gamma$. Then the reduced Hilbert space $W_v$ is an irreducible representation of the Lie algebra $\mathfrak{g}^{\,v}_{\Gamma,\mathrm{free}}$.
\label{DecompThmReduced}
\end{thm}



\section{Connection to the Literature} 
\label{sec:ConnectionLiterature}

We begin by relating the results of this work to a central conjecture proposed in \cite{KLFCCZ} (see Conjecture~1 therein).  
In that work, the authors introduced the notion of the \emph{natural dynamical Lie algebra}, designed to respect the intrinsic symmetries of the MaxCut problem.  
Concretely, this algebra is defined as the subalgebra of the free dynamical Lie algebra that commutes with both the action of the graph automorphism group and the global bit-flip symmetry on the Hilbert space $W$.  
A precise definition may be found in Section~VI of \cite{KLFCCZ}.

For our purposes, the key structural fact is that the natural dynamical Lie algebra sits between the standard and free DLAs:
\begin{equation}
    \label{eq:DLAInclusions}
    \mathfrak{g}_{\Gamma, \text{std}} \subseteq \mathfrak{g}_{\Gamma, \text{nat}} \subseteq \mathfrak{g}_{\Gamma, \text{free}}.
\end{equation}

\begin{defn}
    Given a vector $v$ in a representation space $V$ of a Lie algebra $\mathfrak{g}$, the \emph{cyclic representation generated by $v$} is the smallest $\mathfrak{g}$-invariant subspace of $V$ that contains $v$. It is the span of all vectors obtained by applying finite products of elements of $\mathfrak{g}$ to $v$, i.e.,
\[
\mathfrak{g}\cdot v := \mathrm{span} \left\{ x_1x_2\cdots x_k \cdot v \;\middle|\; x_1, \ldots, x_k \in \mathfrak{g},\; k \geq 0 \right\}.
\]

\end{defn}

Let
\begin{equation}
\label{eq:uniformSuperpositionState}
\ket{\xi_n}
:= \underbrace{\ket{+\,+\,\cdots\,+}}_{n}
\end{equation}
denote the uniform superposition over the standard basis states indexed by \(\mathbb{B}^n\).

Denote by $m$ and $M$ the respective dimensions of the cyclic subrepresentations of $W$ generated by $\ket{\xi_n}$ under the actions of
$\mathfrak{g}_{\Gamma,\mathrm{std}}$ and $\mathfrak{g}_{\Gamma,\mathrm{nat}}$.

\begin{rmk}
One motivation for focusing on these subspaces is rooted in the structure of the Quantum Approximate Optimization Algorithm. The state $\ket{\xi_n}$ is the standard choice of initial state for QAOA, and the unitaries appearing in the QAOA circuit are generated by elements of $\mathfrak{g}_{\Gamma,\mathrm{std}}$. Consequently, the entire QAOA evolution starting from $\ket{\xi}$ remains confined to the cyclic subrepresentation up to the final measurement.

In other words, this subspace captures the effective portion of the Hilbert space that is dynamically accessible to QAOA when initialized in the uniform superposition state. Knowing its dimension therefore provides insight into the true degrees of freedom governing the algorithm’s evolution, as opposed to the full ambient Hilbert space.
\end{rmk}

The conjecture proposed in \cite{KLFCCZ} can be reformulated as follows.

\begin{conj}
\label{Conjecture1}
For most graphs $\Gamma$ on $n$ vertices, the quantity
\begin{equation}
\label{eq:Delta}
    \Delta := M - m
\end{equation}

is bounded by a polynomial in $n$.  
Equivalently, the difference between the dimensions of the cyclic representations generated by $\ket{\xi_n}$ for
$\mathfrak{g}_{\Gamma,\mathrm{std}}$ and $\mathfrak{g}_{\Gamma,\mathrm{nat}}$
grows at most polynomially with the system size.
\end{conj}

While a reduced analog of the natural dynamical Lie algebra has not yet been formally defined, a natural candidate would be the (largest with respect to inclusion) subalgebra of
$\mathfrak{g}^v_{\Gamma,\mathrm{free}}$ commuting with the action of the automorphism group of the reduced graph $\Gamma_v$.
This yields an algebra squeezed between the reduced DLAs
$\mathfrak{g}^v_{\Gamma,\mathrm{std}}$ and $\mathfrak{g}^v_{\Gamma,\mathrm{free}}$.

We now explain why, in the reduced setting considered in this work, the analogue of Conjecture~\ref{Conjecture1} holds in a particularly strong form.

By combining the identity
\[
\mathfrak{g}^v_{\widehat{\Gamma}_v,\mathrm{free}}
=
\mathfrak{g}^v_{\widehat{\Gamma}_v,\mathrm{std}},
\]
which holds for a suitable extension $\widehat{\Gamma}_v$ of any graph $\Gamma$ with chosen vertex $v$ (constructed in Section~\ref{sec:DistinguishedElts}),
together with the fact that the reduced Hilbert space $W_v$ carries an irreducible representation of $\mathfrak{g}^v_{\Gamma,\mathrm{free}}$
(Theorem~\ref{DecompThmReduced}), we obtain the following.

The cyclic representations generated by the initial state $\ket{\xi_{n-1}}$ under the actions of
$\mathfrak{g}^v_{\Gamma,\mathrm{std}}$
and
$\mathfrak{g}^v_{\Gamma,\mathrm{free}}$
both coincide with the entire reduced Hilbert space $W_v$.

In particular, the analogue of the quantity $\Delta$ in \eqref{eq:Delta} is equal to zero in the reduced setting.  
Therefore, Conjecture~\ref{Conjecture1} holds for the reduced dynamical Lie algebras considered in this work, associated to the natural extension of any graph $\Gamma$ with a chosen vertex $v$.

\subsection{Implications for optimization landscapes and barren plateaus}
\label{subsec:BarrenPlateaus}

Another important consequence of our results concerns the structure of the optimization landscape associated with variational quantum algorithms.  
In particular, the dimension and representation-theoretic structure of the underlying dynamical Lie algebra  are known to be tightly linked to the possible emergence of barren plateaus, regions of exponentially suppressed gradients in the parameter landscape
\cite{FHCKYHSP,LJGCC,LCSMCC,RBSKMLC}.

Beyond variational optimization, DLAs also play a central role in quantum machine learning (QML).
In supervised QML, one seeks a circuit within a parametrized family, generated by a finite set of Hamiltonians, that approximates a target function and generalizes to unseen data.
The associated DLAs provide rigorous criteria for the presence or absence of barren plateaus
\cite{GLCCS,LTWS,WHSU}, and furnish powerful tools for analyzing questions of classical simulability
\cite{CLG,GLCCS}.

The barren plateau phenomenon refers to the exponential suppression of gradients of the loss function with respect to circuit parameters, rendering classical training methods ineffective
\cite{mcclean2018barren,larocca2025barren}.
A key diagnostic is the variance of the loss function over the parameter space: if this variance decays exponentially in the number of qubits $n$, then gradient magnitudes also vanish exponentially
\cite{arrasmith2021equivalence}.

Recall that for sufficiently large circuit depth $p$, the variance of the loss function can be approximated using the expression in \eqref{eq:VarAndDim}, which depends explicitly on the projections of the initial state and problem Hamiltonian onto irreducible components of the DLA.

Let
\[
\rho_{n-1} := \ket{\xi_{n-1}}\bra{\xi_{n-1}}
\]
denote the reduced uniform superposition state, and define the normalized reduced problem Hamiltonian
\[
\widehat{H}^v_P := \frac{1}{|\lambda_{\max}|}\, H^v_P,
\]
where $\lambda_{\max}$ is the maximal value of the reduced objective function (see \eqref{eq:maxcutredobjfunction}).

Theorem~\ref{thm:ParitySeparationImpliesLocalX}, together with Theorem~1 of \cite{KLFCCZ}, implies that for any pair $(\Gamma, v)$ for which the identity \eqref{eq:ReducedDLAEquality} holds and the reduced graph $\Gamma_v$ is \emph{archetypal}, we have
\[
\mathfrak{g}^v_{\Gamma,\mathrm{std}}
\simeq \mathfrak{su}(2^{n-1}),
\]

is a simple Lie algebra of dimension $4^{\,n-1}-1$. In addition, we compute
\begin{equation}
    \label{eq:ProjBoundRho}
    \mathcal{P}_{\mathfrak{g}^v_{\Gamma,\mathrm{std}}}(\rho_{n-1}) = \operatorname{Tr}(\rho_{n-1}^2)=1
\end{equation}
(see \eqref{eq:VarAndDim}). Moreover, we have the general upper bound
\begin{equation}
    \label{eq:ProjBoundHP}
    \mathcal{P}_{\mathfrak{g}^v_{\Gamma,\mathrm{std}}}(\widehat{H}^v_P) \;\le\; 2^{n-1}.
\end{equation}

Collecting these bounds and applying \eqref{eq:VarAndDim}, we obtain
\begin{equation}
    \label{eq:VarUpperBound}
\operatorname{Var}
\bigl[
\ell_{\boldsymbol{\beta}, \boldsymbol{\gamma}} (\rho_{n-1}, \widehat{H}^v_P)
\bigr]
\;=\;
\frac{
\mathcal{P}_{\mathfrak{g}^v_{\Gamma,\mathrm{std}}}(\rho_{n-1})\,
\mathcal{P}_{\mathfrak{g}^v_{\Gamma,\mathrm{std}}}(\widehat{H}^v_P)
}{\dim(\mathfrak{g}^v_{\Gamma,\mathrm{std}})}
\;\le\;
\frac{2^{n-1}}{4^{n-1}-1}.
\end{equation}

In particular, the variance decays exponentially in $n$, implying the presence of a barren plateau.


\section{Conclusions} 

We investigated the impact of classical bit-flip symmetry reduction on QAOA through the structure of its dynamical Lie algebras. Although symmetry reduction is classically trivial, we showed that it can induce substantial and nontrivial changes in the quantum dynamics of QAOA.  

For QAOA on MaxCut with the standard mixer, we derived explicit  graph-theoretic conditions under which the standard reduced dynamical Lie algebra contains the free dynamical Lie algebra of the reduced problem. We identified sufficient conditions guaranteeing maximal reduced expressivity, proved that this regime can always be enforced via a quadratic graph extension, and exhibited families of graphs where the algebra dimension grows exponentially. By contrasting these results with Grover-mixer QAOA, we demonstrated that the observed behavior is not universal but depends critically on the underlying algebraic structure.  

Overall, our results establish symmetry reduction as a meaningful mechanism for shaping QAOA dynamics and highlight reduced dynamical Lie algebras as a fundamental tool for understanding expressivity and trainability in variational quantum algorithms. Numerically, we show how vertex reductions lead to an increase in the variance of the QAOA loss function gradient and present the idea of artificial leaf addition as a possible technique for further reduction in the reduced DLAs dimensions.

Our findings also underscore the importance of selecting the reduction vertex strategically, according to criteria such as balancing expressivity against the risk of barren plateaus. This points to the potential utility of classical (polynomial-time) algorithms that can compute or estimate the dimensions of standard reduced DLAs, or, ideally, provide simple factor decompositions. Most technical results in this manuscript are constructive and could, in principle, be implemented to guide subsequent algorithm design and analysis in practical QAOA applications.

\section*{Acknowledgments}
We would like to thank Rui Mao for a careful reading of an earlier version of the manuscript and for identifying a counterexample to the corresponding version of Theorem~\ref{thm:DistDLAElements}.
This research was supported in part by NSF Award No. 2427042. 

\bibliographystyle{plain}
\bibliography{References}

\appendix
\section{Structural Preliminaries}
\label{sec:Preliminaries}

The results and constructions presented in this section provide the main preliminary ingredients for the proofs of the statements appearing in Section~\ref{sec:DistinguishedElts}. While these results are developed here in service of the main theorems, they are formulated in a general setting and may be of independent interest.

\subsection{Distinguished Relations and Elements in Standard DLAs}
We begin by deriving several fundamental relations in the standard dynamical Lie algebras. 
Throughout, we adopt the following notation for iterated adjoint actions. 
For elements $A,B$ of a Lie algebra $\mathfrak{g}$ and $k \in \mathbb{N}$, define
\begin{equation}
\label{eq:adjoint-power}
\operatorname{ad}_A^{k}(B) := [A,\operatorname{ad}_A^{k-1}(B)],
\qquad 
\operatorname{ad}_A^{0}(B) := B .
\end{equation}
That is, $\operatorname{ad}_A^{k}(B)$ denotes the $k$-fold adjoint action of $A$ on $B$.

\begin{prop}
The following relations hold in the algebra $\mathfrak{g}_{\Gamma,\text{std}}$: 
\begin{enumerate}
\item $\operatorname{ad}^k_Z(X) =  -2^ki \sum\limits_{v \in V} Y_v \left( \sum\limits_{w \in \mathcal{N}_v} Z_w \right)^k, \quad \text{for $k$ odd,}$
    \item $\operatorname{ad}^k_Z(X) =  -2^ki \sum\limits_{v \in V} X_v \left( \sum\limits_{w \in \mathcal{N}_v} Z_w \right)^k, \quad\quad \text{for $k$ even,}$
    \item $\operatorname{ad}^3_X(Z)+16\operatorname{ad}_X(Z) =0$.
\end{enumerate}
\label{propRelns}
\end{prop}

\begin{proof}
First, we confirm the base case for \( k = 1 \):  
\begin{equation}\label{eq:AdZX}
    \operatorname{ad}_Z(X) = [Z, X] = \left[ i\sum\limits_{(vw) \in E} Z_v Z_w, \sum\limits_{v \in V} iX_v \right] = -2i \sum\limits_{(vw) \in E} (Y_v Z_w + Z_v Y_w) = -2i \sum\limits_{v \in V} Y_v \left( \sum\limits_{w \in \mathcal{N}_v} Z_w \right)
\end{equation}

Next, we observe that the \( Z \)-gates corresponding to any pair of vertices (including coinciding ones) commute. As a result, the adjoint operator \( \operatorname{ad}_Z \) commutes with any polynomial \( f(Z_1, \ldots, Z_n) \).

Substituting the expression for \( \operatorname{ad}_Z(X) \), we proceed as follows:  
\[
\operatorname{ad}^2_Z(X)= [Z, [Z, X]] = \left[Z, -2i \sum\limits_{v \in V} Y_v \left( \sum\limits_{w \in \mathcal{N}_v} Z_w \right)\right].
\]  

Using the aforementioned commutativity, we get:  
\[
-2i \sum\limits_{v \in V} \left( \sum\limits_{w \in \mathcal{N}_v} Z_w \right) [Z, Y_v]=-2 \sum\limits_{v \in V} \left( \sum\limits_{w \in \mathcal{N}_v} Z_w \right) [Z, iY_v].
\]  

Recalling that \( [iZ_wZ_v, iY_v] = 2i Z_wX_v \), we substitute:  
\[
\operatorname{ad}^2_Z(X) = -2 \sum\limits_{v \in V} \left( \sum\limits_{w \in \mathcal{N}_v} Z_w \right)^2(2i X_v) = 4i \sum\limits_{v \in V} X_v \left( \sum\limits_{w \in \mathcal{N}_v} Z_w \right)^2.
\]

The assertion for higher values of \( k \) follows analogously by induction, using the same commutator properties and the fact that \( \operatorname{ad}_Z \) commutes with polynomials in \( Z_1, \ldots, Z_n \).

Finally, we verify property $(3)$. Using \eqref{eq:AdZX}, we get that 
\[
\operatorname{ad}_X(Z) = -\operatorname{ad}_Z(X) = 2 i\sum\limits_{(vw) \in E} (Y_v Z_w + Z_v Y_w).
\]
Next, we compute 
\[
\begin{aligned}
 \operatorname{ad}^2_X(Z) &= 2 \left[\sum\limits_{v \in V} iX_v, \sum\limits_{(vw) \in E} i(Y_v Z_w + Z_v Y_w) \right] = 4 i\sum\limits_{(vw) \in E} (-Z_v Z_w + Y_v Y_w - Z_v Z_w + Y_v Y_w) \\
 &=8i \sum\limits_{(vw) \in E} (Y_v Y_w -Z_v Z_w) 
\end{aligned}
\]
 and establish
\[
\begin{aligned}
 \operatorname{ad}^3_X(Z) &= 8 \left[i\sum\limits_{i \in V} X_v, i\sum\limits_{(vw) \in E} (Y_v Y_w -Z_v Z_w) \right] = -16i \sum\limits_{(vw) \in E} (Z_v Y_w + Y_v Z_w + Y_v Z_w + Z_v Y_w)\\
 &= -32i \sum\limits_{(vw) \in E} (Z_v Y_w + Y_v Z_w) = -16 \operatorname{ad}_X(Z).
 \end{aligned}
\]
\end{proof}

\begin{defn}
\label{defn:BalanceFn}
Let $b \in \mathbb{B}^n$ be a binary string of length $n$, and let 
\[
\{i_1, i_2, \ldots, i_m\} \subseteq \{1,2,\ldots,n\}
\]
be a set of distinct indices. The \emph{balance} of $b$ supported on these indices, denoted by
\[
\operatorname{Bal}_{i_1 \ldots i_m}(b),
\]
is defined as the number of $0$’s minus the number of $1$’s appearing in $b$ at positions $i_1, i_2, \ldots, i_m$. Equivalently,
\[
\operatorname{Bal}_{i_1 \ldots i_m}(b)
:= \bigl|\{j \in \{i_1,\ldots,i_m\} : b_j = 0\}\bigr|
   - \bigl|\{j \in \{i_1,\ldots,i_m\} : b_j = 1\}\bigr|.
\]
\end{defn}

\begin{prop}
    Consider the polynomial 
    \[
    f(Z_{i_1}, Z_{i_2}, \dots, Z_{i_m}) = (Z_{i_1} + Z_{i_2} + \ldots + Z_{i_m})^k.
    \] 
    Then for any standard basis vector \( \ket{b}  \), we have 
\[
f(Z_{i_1}, Z_{i_2}, \dots, Z_{i_m})(\ket{b}) = (\operatorname{Bal}_{i_1 \ldots i_m}(b))^k \ket{b}.
\] 
\label{BalProp}
\end{prop}

\begin{proof}
    Recall that for any \( j \in \{i_1, i_2, \dots, i_m\} \), the action of \( Z_{i_j} \) on the standard basis vector \( \ket{b} \) is given by
    \[
    Z_{i_j}(\ket{b}) = 
    \begin{cases} 
    \ket{b}, & \text{if } b_{i_j} = 0, \\ 
    -\ket{b}, & \text{if } b_{i_j} = 1.
    \end{cases}
    \]
    Therefore, the sum of operators \( Z_{i_1} + Z_{i_2} + \ldots + Z_{i_m} \) acting on \( \ket{b} \) results in
    \[
    (Z_{i_1} + Z_{i_2} + \ldots + Z_{i_m})(\ket{b}) = \operatorname{Bal}_{i_1 \ldots i_m}(b) \ket{b}.
    \]
   This follows because each \( Z_{i_j} \) contributes \( +1 \) if \( b_j = 0 \) and \( -1 \) if \( b_j = 1 \) with $j\in \{i_1,\ldots,i_m\}$.

    Hence, applying \( f(Z_{i_1}, Z_{i_2}, \dots, Z_{i_m}) = (Z_{i_1} + Z_{i_2} + \ldots + Z_{i_m})^k \) to \( \ket{b} \) gives
    \[
    f(Z_{i_1}, Z_{i_2}, \dots, Z_{i_m})(\ket{b}) = (\operatorname{Bal}_{i_1 \ldots i_m}(b))^k \ket{b}.
    \]
    This concludes the proof.
\end{proof}

Let \( f(t) = \sum\limits_{j=1}^n \alpha_j t^j \) be a polynomial with real coefficients, i.e., \( f(t) \in \mathbb{R}[t] \). We associate the operator  
\begin{equation}\label{eq:OpAssocToPoly}
    F(\operatorname{ad}^2_Z) := \sum\limits_{j=1}^n \alpha_j \operatorname{ad}^{2j}_Z
\end{equation}  
to \( f(t) \) and observe that the element \( F(\operatorname{ad}^2_Z)(X) \) resides in the standard dynamical Lie algebra \( \mathfrak{g}_{\Gamma, \text{std}} \).  

Using part $(2)$ of Proposition \ref{propRelns}, we find that  
\[
F(\operatorname{ad}^2_Z)(X) = -i\sum\limits_{j=1}^n 2^{2j} \alpha_j \sum\limits_{v \in V} X_v \left( \sum\limits_{w \in \mathcal{N}_v} Z_w \right)^{2j}.
\]  
This can be rewritten as  
\begin{equation}\label{PolyAdZEqn}
   F(\operatorname{ad}^2_Z)(X) = -i\sum\limits_{v \in V} X_v \sum\limits_{i=1}^n 2^{2j} \alpha_j \left( \sum\limits_{w \in \mathcal{N}_v} Z_w \right)^{2j} = -i\sum\limits_{v \in V} X_v \widehat{f} \left( \left( \sum\limits_{w \in \mathcal{N}_v} Z_w \right)^2 \right),
\end{equation}

where \( \widehat{f}(t) := \sum\limits_{j=1}^n 2^{2j} \alpha_j t^j \).  

Combining this result with Proposition \ref{BalProp}, we have  
\begin{equation}\label{eq:AdZonStandardBasis}
    F(\operatorname{ad}^2_Z)(X)(\ket{b}) = -i\sum\limits_{v \in V} X_v \widehat{f} \left( \operatorname{Bal}_{\mathcal{N}_v}(\ket{b})^2 \right),
\end{equation}

for any standard basis vector \( \ket{b}  \).

Let 
\begin{equation}
    \label{eq:EvenOddVertices}
    \mathcal{V}_{\text{even}}:=\{v\in V(\Gamma)\mid\deg(v)\equiv 0\pmod{2}\} \text{ and } \mathcal{V}_{\text{odd}}:=\{v\in V(\Gamma)\mid\deg(v)\equiv 1\pmod{2}\}
\end{equation}
denote, respectively, the subsets of vertices of even and odd degree in the graph~$\Gamma$.

With this notation in place, we are ready to establish the following result.

\begin{thm}
\label{thm:CoolEltsInStdDLA}
    Let \( \Gamma \) be a connected graph. The standard dynamical Lie algebra \( \mathfrak{g}_{\Gamma, \mathrm{std}} \) contains the following elements:
    \begin{enumerate}
        \item \( \mathcal{X}_{\text{even}} := i\sum\limits_{v \in \mathcal{V}_{\text{even}}} X_v \), 
        \item \( \mathcal{X}_{\text{odd}} := i\sum\limits_{v \in \mathcal{V}_{\text{odd}}} X_v \), 
        \item \( \mathcal{Z}_{\text{even/odd}} := i\sum\limits_{\substack{(vw)\in E\\ v\in \mathcal{V}_{\text{even}}\\w\in \mathcal{V}_{\text{odd}}}} Z_v Z_w \),

    \end{enumerate}
\end{thm}

\begin{proof}
Let \( m \) denote the largest even degree among the vertices of \( \Gamma \).
We introduce the polynomials
\begin{equation}
\label{eq:SuperPoly}
\widehat{o}_m(t) := \sum\limits_{j=1}^{m+1} \alpha_j t^j \in \mathbb{R}[t],
\qquad
\widehat{e}_m(t) := \sum\limits_{j=1}^{m+1} \alpha_j t^j \in \mathbb{R}[t],
\end{equation}
characterized by the interpolation conditions
\[
\widehat{o}_m(k)=
\begin{cases}
0, & k = 0^2,2^2,\dots,m^2,\\
1, & k = 1^2,3^2,\dots,(m+1)^2,
\end{cases}
\qquad
\widehat{e}_m(k)=
\begin{cases}
1, & k = 0^2,2^2,\dots,m^2,\\
0, & k = 1^2,3^2,\dots,(m+1)^2.
\end{cases}
\]
Such polynomials can be constructed explicitly using Lagrange interpolation. Observe that the range of the function \( \operatorname{Bal}_{i_1 \ldots i_k} \) on the set of standard basis vectors is \( \{-k, -k+2, \ldots, k-2, k\} \).  
    Therefore, the union of the ranges of the functions \( \operatorname{Bal}_{i_1 \ldots i_k} \) for \( k \in \{2, 4, \ldots, m\} \) (with \( m \) even) is \( \{-m, -m+2, \ldots, m-2, m\} \) (see Definition \ref{defn:BalanceFn}).  
    Consequently, the union of the ranges of the squared functions \( \operatorname{Bal}_{i_1 \ldots i_k}^2 \) on the set of standard basis vectors is \( \{0^2, 2^2, 4^2, \ldots, m^2\} \).  
    
    Let $\mathcal{O}_m(\operatorname{ad}^2_Z)$ and $\mathcal{E}_m(\operatorname{ad}^2_Z)$ be the operators associated with the polynomials $o_m(t)$ and $e_m(t)$ according to the formula in \eqref{eq:OpAssocToPoly}. Using the formula from the operator expansion in \eqref{PolyAdZEqn}, we get  
    \[
    F_m(\operatorname{ad}^2_Z)(iX_v) = i\left( \sum\limits_{j=1}^m \alpha_j \left( \sum\limits_{w \in \mathcal{N}_v} Z_w \right)^{2j} \right) X_v. 
    \] operator \( F_m(\operatorname{ad}^2_Z)(X) \in \mathfrak{g}_{\Gamma, \text{std}} \).  For any standard basis vector \( b \in \mathcal{B}_{0,1} \), define \( w_{b,v} := \operatorname{Bal}_{\mathcal{N}_v}(b) \), where \( \mathcal{N}_v \) denotes the neighborhood of \( v \). Then, according to Proposition~\ref{BalProp}, we have:
\[
\mathcal{O}_m(\operatorname{ad}^2_Z)(iX_v)(b) = i\widehat{o}_m(w_{b,v}^2) X_v(b).
\]

If the degree of \( v \) is an even number \( k \), then \( w_{b,v} \) ranges over the even numbers between \( 0 \) and \( m \). Therefore, \( \widehat{o}_m(w_{b,v}^2) = 0 \), and the operator \( \mathcal{O}_m(\operatorname{ad}^2_Z)(iX_v) \) acts identically as zero on the Hilbert space \( W \).  

On the other hand, for any vertex \( v \) of odd degree and any \( b \in \mathcal{B}_{0,1} \), the range of \( w_{b,v} \) consists of odd values between $1$ and $m+1$. In this case, \( \widehat{o}_m(w_{b,v}^2)=1 \).

This implies that 
\[
\mathcal{O}_m(\operatorname{ad}^2_Z)(iX_v) = 
\begin{cases}
iX_v, & \text{if \( v \) has odd degree}, \\
0, & \text{if \( v \) has even degree}.
\end{cases}
\]
Therefore, we conclude
\[
\mathcal{O}_m(\operatorname{ad}^2_Z)(X) = i\sum\limits_{v \in \mathcal{V}_{\text{odd}}} X_v = \mathcal{X}_{\text{odd}}\in\mathfrak{g}_{\Gamma, \text{std}}.
\]

Hence, the elements \( \mathcal{X}_{\text{odd}} \) and \( \mathcal{X}_{\text{even}} = X - \mathcal{X}_{\text{odd}} \) both belong to \( \mathfrak{g}_{\Gamma, \text{std}} \).

Next, observe that the operator \( \operatorname{ad}^2_{\mathcal{X}_{\text{odd}}}(\operatorname{ad}^2_{\mathcal{X}_{\text{even}}}(Z)) \) is proportional to \( \mathcal{Z}_{\text{even/odd}} \). Thus, \( \mathcal{Z}_{\text{even/odd}} \) is in \( \mathfrak{g}_{\Gamma, \text{std}} \). Furthermore, since:
\[
\mathcal{Z}_{\text{even/even}} + \mathcal{Z}_{\text{odd/odd}} = Z - \mathcal{Z}_{\text{even/odd}},
\]
it follows that \( \mathcal{Z}_{\text{even/even}} + \mathcal{Z}_{\text{odd/odd}} \) is also in \( \mathfrak{g}_{\Gamma, \text{std}} \).

This completes the proof.
\end{proof}

\subsection{Distinguished Elements in Reduced DLAs}
We begin with two structural lemmas describing elements that necessarily
belong to the standard and free reduced dynamical Lie algebras.
The first shows that in the standard reduced DLA one can separate
the single-site $Z$-terms supported on neighbors of $v$
from the $Z$–$Z$ interactions internal to the reduced graph.
The second lemma identifies natural Lie subalgebras contained in the
free and standard reduced DLAs. In particular, the free reduced DLA contains, for each vertex
adjacent to the reduction vertex $v$, a copy of $\mathfrak{su}(2)$
supported on the corresponding qubit Hilbert space.

\begin{lem}
\label{ZsplittingLemma}
    The standard reduced  dynamical Lie algebra $\mathfrak{g}^{\,v}_{\Gamma,\mathrm{std}}$ 
    contains the elements
    \begin{equation}
        \label{eq:Zsplitting}
        \mathcal{Z}_{\widehat{v},1} := i\sum\limits_{(vw) \in E(\Gamma)} Z_w,
        \qquad 
        \mathcal{Z}_{\widehat{v},2} := i\sum\limits_{(ww') \in E(\Gamma_v)} Z_w Z_{w'}.
    \end{equation}
\end{lem}

\begin{proof}
We compute:
\[
\begin{aligned}
   [\,\mathcal{X}_{\widehat{v}}, \mathcal{Z}_{\widehat{v}}\,] 
   &= -2i \sum\limits_{(ww') \in E(\Gamma_v) } \bigl( Y_w Z_{w'} + Z_w Y_{w'} \bigr) -2i \sum\limits_{(vw) \in E(\Gamma)} Y_w;\\
   \operatorname{ad}^3_{\mathcal{X}_{\widehat{v}}}(\mathcal{Z}_{\widehat{v}}) &=  32i \sum\limits_{(ww') \in E(\Gamma_v) } \bigl( Y_w Z_{w'} + Z_w Y_{w'} \bigr) +8i \sum\limits_{(vw) \in E(\Gamma)} Y_w.
\end{aligned}
\]

From this we isolate the two contributions. Define
\[
\widetilde{Y}_1 :=\frac{1}{24}\cdot\left(\operatorname{ad}^3_{\mathcal{X}_{\widehat{v}}}(\mathcal{Z}_{\widehat{v}})+ 16[\,\mathcal{X}_{\widehat{v}}, \mathcal{Z}_{\widehat{v}}\,] \right)
=i \sum\limits_{(vw) \in E(\Gamma)} Y_w,
\]
and then

\[
-0.5[\,\mathcal{X}_{\widehat{v}}, \widetilde{Y}_1\,] =i\sum\limits_{(vw) \in E} Z_w= \mathcal{Z}_{\widehat{v},1}, 
\]
yields the desired elements $\mathcal{Z}_{\widehat{v},1}$ and $\mathcal{Z}_{\widehat{v},2}$ in $\mathfrak{g}^{\,v}_{\Gamma,\mathrm{free}}$ (see \ref{eq:Zsplitting}).
\end{proof}

\begin{lem}
\label{lem:ReducedSubalgebras}
Let $\Gamma$ be a graph and $v \in V(\Gamma)$ a vertex.  
Then the free reduced dynamical Lie algebra $\mathfrak{g}^{\,v}_{\Gamma,\mathrm{free}}$ contains the following natural subalgebras:
\begin{itemize}

    \item[$(1)$] The free dynamical Lie algebra on the reduced graph $\Gamma_v$ (see Definition \ref{def:ReducedGraph}):
    \[
        \mathfrak{g}_{\Gamma_v,\mathrm{free}}
        = \Big\langle 
            \{\, iX_w \mid w \in V(\Gamma)\setminus\{v\} \,\},\;
            \{\, iZ_w Z_{w'} \mid (w w') \in E(\Gamma),\ w,w' \neq v \,\}
        \Big\rangle_{\mathrm{Lie}},
    \]
    generated by all single-qubit $iX$ terms except the one corresponding to $v$, together with all $Z-Z$ interactions on edges not incident to $v$.

    \item[$(2)$] The direct sum of single-qubit $\mathfrak{su}(2)$ algebras acting on the neighbors of $v$,
    \[
        \mathfrak{su}(2)^{\oplus |\mathcal{N}_v|}
        =
        \Big\langle 
            \{\, iX_j \mid w \in \mathcal{N}_w \,\}, \;
            \{\, iZ_j \mid w \in \mathcal{N}_w \,\}
        \Big\rangle_{\mathrm{Lie}},
    \]
    corresponding to the independent local actions on each vertex adjacent to $v$.
\end{itemize}

Similarly, the standard reduced  dynamical Lie algebra $\mathfrak{g}^{\,v}_{\Gamma,\mathrm{std}}$ contains the following subalgebras:
\begin{itemize}

    \item[$(1)$] The standard dynamical Lie algebra on the reduced graph $\Gamma_v$:
    \[
        \mathfrak{g}_{\Gamma_v,\mathrm{std}}
        = \Big\langle 
           \mathcal{X}_{\widehat{v}},\;
            \,\mathcal{Z}_{\widehat{v},2} \,
        \Big\rangle_{\mathrm{Lie}}.
    \]

    \item[$(2)$] The  $\mathfrak{su}(2)$ subalgebra, arising as
    \[
        \mathfrak{su}(2)=\Big\langle 
           \mathcal{X}_{\widehat{v}},\;
            \,\mathcal{Z}_{\widehat{v},1} \,
        \Big\rangle_{\mathrm{Lie}}.
    \]
    
\end{itemize}
\end{lem}

\begin{proof}
By Lemma~\ref{ZsplittingLemma}, the specified generators indeed belong to the corresponding dynamical Lie algebras.  
It then follows directly that their Lie brackets generate the subalgebras listed above. 
\end{proof}

\section{Proofs and Technical Details: Sections \ref{sec:DistinguishedElts} and \ref{sec:DLASonSomeGraphs}} 
\label{sec:Technicalities}

This section collects the main technical arguments underlying the results of Sections~\ref{sec:DistinguishedElts} and~\ref{sec:DLASonSomeGraphs}. Our main goal is to verify that the elements
\eqref{eq:FixedDistX1element} and \eqref{eq:ReducedDLAElts} are contained in the standard reduced
dynamical Lie algebra
\(
\mathfrak{g}^{\,v}_{\Gamma,\mathrm{std}}.
\)
\subsection{More Distinguished Elements in Reduced DLAs}

We begin by showing that $\mathfrak{g}^{\,v}_{\Gamma,\mathrm{std}}$ contains uniform $X$-type operators supported on distance layers from the distinguished vertex $v$. We then refine this construction to isolate subsets of vertices specified by parity data along distance-increasing paths.

\begin{proof}[Proof of Theorem \ref{thm:DistDLAElements}]
    We proceed by induction on $j$. For the base case $j=1$, we first observe that the degree of any vertex $w \in \mathcal{N}_{v,1}$ is one greater than its degree in the reduced graph $\Gamma_v$. A direct computation shows that
\begin{equation}
\label{eq:X1element}
\begin{aligned}
\mathcal{X}_{\widehat{v},1}&=0.25\operatorname{ad}^{2}_{\mathcal{Z}_{\widehat{v},1}}(\mathcal{X}_{\widehat{v}}),\\
\mathcal{X}_{v,1,\mathrm{even}}&=\mathcal{O}_m(\operatorname{ad}^{2}_{\mathcal{Z}_{\widehat{v},2}})(\mathcal{X}_{\widehat{v},1}),\\
\mathcal{X}_{v,1,\mathrm{odd}}&=\mathcal{X}_{\widehat{v},1}-\mathcal{X}_{v,1,\mathrm{even}},
\end{aligned}
\end{equation}
where $\mathcal{X}_{\widehat{v}}$, $\mathcal{Z}_{\widehat{v},1}$, and $\mathcal{Z}_{\widehat{v},2}$ are defined as in Equations~\eqref{eq:ReducedStandardDLAgenerators} and~\eqref{eq:Zsplitting}. From \eqref{eq:X1element}, it follows that $\mathcal{X}_{\widehat{v},1}$, $\mathcal{X}_{v,1,\mathrm{even}}$, and $\mathcal{X}_{v,1,\mathrm{odd}}$ are elements of the reduced dynamical Lie algebra $\mathfrak{g}^{\,v}_{\Gamma,\mathrm{std}}$, thereby establishing the base case.

Now, assume that $\mathcal{X}_{\widehat{v},j}, \mathcal{X}_{v,j,\mathrm{even}}, \mathcal{X}_{v,j,\mathrm{odd}} \in \mathfrak{g}^{\,v}_{\Gamma,\mathrm{std}}$ for all $j \le n$. To prove the statement for $j = n+1$, we construct an auxiliary operator consisting of two-qubit Pauli-$Z$ terms supported on the vertices in $\mathcal{N}_{v,n} \cup \mathcal{N}_{v,n+1}$:

\begin{equation}
    \label{eq:XLevelnPreparation}
\mathcal{Z}_{n}:=-0.25\operatorname{ad}^{2}_{\mathcal{X}_{\widehat{v},n}}\left(\mathcal{Z}_{\widehat{v},2}\right)=i\sum\limits_{\substack{(ww')\in E\\ w\in \mathcal{N}_{v,n}\\w'\not\in \mathcal{N}_{v,n}}} Z_{w} Z_{w'}-2i\sum\limits_{\substack{(ww')\in E\\ w,w'\in \mathcal{N}_{v,n}}} (Y_{w} Y_{w'}-Z_{w} Z_{w'})
\end{equation}

Utilizing this element and the assumption in \eqref{eq:ParityAssumption}, we obtain $\mathcal{X}_{\widehat{v},n+1}$ via the relation:

\begin{equation}
    \label{eq:XLeveln}
\mathcal{X}_{\widehat{v},n+1}=\mathcal{O}_m(\operatorname{ad}^{2}_{\mathcal{Z}_{n}})(\mathcal{X}_{\widehat{v}}-\mathcal{X}_{\widehat{v},n}).
\end{equation}
Finally, the parity-specific components are recovered following the same logic as the base case:
\begin{equation}
    \label{eq:XLevelnParity}
\begin{aligned}
\mathcal{X}_{v,n+1,\mathrm{odd}}&=\mathcal{O}_m(\operatorname{ad}^{2}_{\mathcal{Z}_{\widehat{v},2}})(\mathcal{X}_{\widehat{v},n+1}),\\
\mathcal{X}_{v,n+1,\mathrm{odd}}&=\mathcal{X}_{\widehat{v},n+1}-\mathcal{X}_{v,n+1,\mathrm{even}}.
\end{aligned}
\end{equation}

This completes the inductive step and hence the proof.
\end{proof}

We proceed with the verifying the properties listed in Theorem~\ref{thm:GraphExtensionEmbedding}. We recall that the graph extension construction was already introduced in Section \ref{sec:DistinguishedElts} of the main text. 
For completeness, we reproduce it here.

\begin{proof}[Proof of Theorem~\ref{thm:GraphExtensionEmbedding}]
By monitoring the change in the number of vertices and edges at each modification step, we verify that $\widehat{\Gamma}_v$ satisfies the theorem's structural requirements.

 \medskip
\noindent\textbf{Step 0. Adjusting parity.} The initial reduced graph $\Gamma_v$ contains $n-1$ vertices. Since each vertex violating the parity assumption \eqref{eq:ParityAssumption} triggers an extension via the gadget shown in Figure~\ref{fig:ExtGamma}, adding two vertices and four edges per instance, the total number of vertices in the resulting graph is bounded by $1 + 3(n-1) = 3n-2$. Similarly, the number of edges is at most $|E(\Gamma)| + 4(n-1)$.

\medskip

\noindent\textbf{Step 1. Uniformizing maximal distance.} 
Let $v\in V(\Gamma)$ be the distinguished vertex and let $j(v)$ denote the maximal distance from $v$ to any vertex of $\Gamma$ (see \eqref{eq:MaxDistToV}). Notice that it is not altered on the previous step. 
For each vertex $w \in V(\Gamma_v)$, we attach to $w$ a path graph of length 
\[
j(v)-\mathrm{dist}(v,w).
\]
If $\mathrm{dist}(v,w)=j(v)$, no modification is performed.

Denote by $\Gamma'$ the resulting graph.  
By construction, every vertex $w \in V(\Gamma_v)$ now either lies at distance $j(v)$ from $v$ in $\Gamma'$, or is the ancestor of a unique leaf at that distance.
Hence there is a bijection
\[
V(\Gamma_v)
\longrightarrow 
\mathcal{N}'_{v,j(v)},
\]
obtained by following the attached path from $w$ to its endpoint.

Since each added vertex contributes exactly one new edge, the total number of added vertices and edges in this step is bounded above by
\[
(j(v)-1)(3n-3).
\]

\medskip

\noindent\textbf{Step 2. Creating distinct distance levels.} 
Label the vertices of $\mathcal{N}'_{v,j(v)}$ from $0$ to $|\mathcal{N}'_{v,j(v)}|-1$. To each vertex with index $j$, attach a path of length $j$. Let $\widehat{\Gamma}_v$ denote the resulting graph.

By construction, the newly added vertices in this step have pairwise distinct distances to $v$.  Moreover, each such vertex is the unique vertex in $\widehat{\Gamma}_v$ at its respective distance from $v$.

\medskip

\noindent\textbf{Vertex and edge bounds.}

We verify that $\widehat{\Gamma}_v$ possesses the properties stated in the theorem. 
Specifically, the total vertex count is bounded by:
\[
\lvert V(\widehat{\Gamma}_v) \rvert 
< 1 + 3(j(v)-1)(n-1) + \sum_{j=0}^{\lvert\mathcal{N}'_{v,j(v)}\rvert-1} j 
< 1 + 3(j(v)-1)(n-1) + \frac{(3n-2)(3n-3)}{2},
\]
where the leading $1$ represents the vertex $v$ and the final inequality follows from the bound $\lvert\mathcal{N}'_{v,j(v)}\rvert < n$.

Similarly, the total number of edges in $\widehat{\Gamma}_v$ satisfies
\[
|E(\widehat{\Gamma}_v)|
<
|E(\Gamma)| + 3(j(v)-1)(n-1) + \sum\limits_{j=0}^{|\mathcal{N}'_{v,j(v)}|-1} j
<
|E(\Gamma)| + 3(j(v)-1)(n-1)+ \frac{(3n-2)(3n-3)}{2}.
\]

\medskip

The vertices added in the second step have pairwise distinct distances to $v$ by construction. 
Moreover, each such vertex is the unique vertex in $\widehat{\Gamma}_v$ at its respective distance from $v$. 
It therefore follows from Theorem~\ref{thm:DistDLAElements} that
\[
iX_w \in \mathfrak{g}^{\,v}_{\widehat{\Gamma}_v,\mathrm{std}}
\]
for each of these vertices.

Applying Remark~\ref{rmk:ExtensionOfResults} together with Theorem~\ref{thm:DistDLAElements} to the vertices in $\mathcal{N}'_{v,j(v)}$, we conclude that
\[
iX_w \in \mathfrak{g}^{\,v}_{\widehat{\Gamma}_v,\mathrm{std}}
\]
for all remaining vertices as well, allowing to establish the required containment of DLAs. 

Finally, we establish a canonical extension-reduction principle between the original and extended graphs. Specifically, every MaxCut configuration on $\Gamma$ admits a unique extension to a configuration on $\widehat{\Gamma}_v$. Conversely, any solution on $\widehat{\Gamma}_v$ restricts uniquely to a solution on $\Gamma$ by discarding the values assigned to the auxiliary vertices. 

Each triangular subgraph introduced in Step $0$ (cf.~Figure~\ref{fig:ExtGamma}) possesses a maximum cut value of $3$. A crucial property of this construction is that this maximum is attained regardless of whether the states assigned to $w_{j-2}$ and $w_j$ are identical or opposite, thereby preserving the original cut constraints. 

This local property, combined with the path attachments in the subsequent steps, allows the MaxCut objective function on $\widehat{\Gamma}_v$ to be decomposed into two independent contributions: the original objective function and a constant term arising from the optimal cuts of the auxiliary edges in $\widehat{\Gamma}_v \setminus \Gamma$. Because the auxiliary components do not impose additional constraints on the states of the original vertices, every solution on $\Gamma$ extends to a unique solution on $\widehat{\Gamma}_v$. Conversely, any MaxCut solution on $\widehat{\Gamma}_v$ reduces to a solution on $\Gamma$ by restriction. This ensures a natural one-to-one correspondence between the set of MaxCut solutions on $\Gamma$ and those on $\widehat{\Gamma}_v$, which completes the proof.
\end{proof}




\begin{proof}[Proof of Theorem \ref{thm:CoolDLAElements}] We argue by induction on $j$, beginning with the base case $j=1$.

Recall that the containment $\mathcal{X}_{\widehat{v},1}\in \mathfrak{g}^{\,v}_{\Gamma,\mathrm{std}}$ was established in Theorem \ref{thm:DistDLAElements}.
Next, proceeding as in the proof of Theorem~\ref{thm:CoolEltsInStdDLA}, we obtain

\begin{equation}
    \label{eq:XLevel1}
\begin{aligned}
X_{\mathcal{C}^v_{1,0}}&=\mathcal{E}_m(\operatorname{ad}^{2}_{\mathcal{Z}_{\widehat{v},2}})(\mathcal{X}_{\widehat{v},1})&=4i \sum\limits_{\substack{w\in\mathcal{N}_{v,1}\\ \deg(w)\equiv 0\pmod 2}} X_w,\\
X_{\mathcal{C}^v_{1,1}}&=\mathcal{X}_{\widehat{v}}-X_{\mathcal{C}^v_{1,0}},
\end{aligned}
\end{equation}

both of which lie in $\mathfrak{g}^{\,v}_{\Gamma,\mathrm{std}}$. This establishes the claim for $j=1$.

Now assume that for some $n\ge 1$ and for every sequence $\mathbf{a}:=(a_1,a_2,\dots,a_j)\in\mathbb{Z}_{2}^j$ with $j\le n$, the elements
\[
X_{\mathcal{C}^v_{n,\mathbf{a}}}
\in 
\mathfrak{g}^{\,v}_{\Gamma,\mathrm{std}}
\]
are contained in the reduced dynamical Lie algebra. We show that the same holds for $j=n+1$.

Let $\mathbf{a}:=(a_1,a_2,\dots,a_{n+1})\in\mathbb{Z}_{2}^{n+1}$ and denote by
$\mathbf{a}':=(a_1,a_2,\dots,a_n)\in\mathbb{Z}_{2}^n$
its truncation. By the induction hypothesis, the element
$X_{\mathcal{C}^v_{n,\mathbf{a}'}}$ belongs to $\mathfrak{g}^{\,v}_{\Gamma,\mathrm{std}}$. We define
\[
X_{\widehat{v},> n}
\;:=\;
\mathcal{X}_{\widehat{v}} - \sum\limits_{k=1}^{n} X_{\widehat{v},k},
\]
which belongs to $\mathfrak{g}^{\,v}_{\Gamma,\mathrm{std}}$ by Theorem~\ref{thm:DistDLAElements}. Using this operator, we introduce the following auxiliary element:

\begin{equation}
    \label{eq:XLevelnPrep}
\begin{aligned}
\mathcal{Z}_{\mathbf{a}'}:&=\frac{1}{16}\cdot\operatorname{ad}^{2}_{X_{\widehat{v},> n}}\left(\operatorname{ad}^{2}_{X_{\mathcal{C}^v_{n,\mathbf{a}'}}}(\mathcal{Z}_{\widehat{v},2})\right)\\
&=\frac{1}{16}\cdot\operatorname{ad}^{2}_{X_{\widehat{v},> n}}\left(-4i\sum\limits_{\substack{(ww')\in E\\ w\in \mathcal{C}^v_{n,\mathbf{a}'}\\w'\not\in \mathcal{C}^v_{n,\mathbf{a}'}}} Z_{w} Z_{w'}+8i\sum\limits_{\substack{(ww')\in E\\ w,w'\in \mathcal{C}^v_{n,\mathbf{a}'}}} (Y_{w} Y_{w'}-Z_{w} Z_{w'})\right)\\
&=i\sum\limits_{\substack{(ww')\in E\\ w\in \mathcal{C}^v_{n,\mathbf{a}'}\\w'\in X_{\widehat{v},> n}}} Z_{w} Z_{w'}.
\end{aligned}
\end{equation}



Finally, we obtain

\begin{equation}
    \label{eq:XLeveln2}
X_{\mathcal{C}^v_{n+1,\mathbf{a}',1}}=\operatorname{ad}^{2}_{\mathcal{Z}_{\mathbf{a}'}}\left(X_{\widehat{v},> n,\operatorname{odd}}\right)\text{ and } X_{\mathcal{C}^v_{n+1,\mathbf{a}',0}}=\operatorname{ad}^{2}_{\mathcal{Z}_{\mathbf{a}'}}\left(X_{\widehat{v},> n,\operatorname{even}}\right)
\end{equation}

This establishes the inductive step and completes the proof.

  \end{proof}

In order to establish the statement in Theorem \ref{thm:ParitySeparationImpliesLocalX}, we begin with a technical lemma that allows us to isolate elements supported on intersections of parity-defined subsets.

\begin{lem}
\label{lem:IntersectionFromParity}
Let $j\ge 1$ and let $\mathbf{a}\neq\mathbf{b}$ be binary sequences  in $\mathbb{Z}_2^j$.
Let
\[
\mathcal{S}_{\mathbf{a}}\subseteq \mathcal{C},
\qquad
\mathcal{S}_{\mathbf{b}}\subseteq \mathcal{C}',
\qquad
\mathcal{S}_{\mathbf{a}}\cap \mathcal{S}_{\mathbf{b}}\neq\varnothing,
\]
where $(\mathcal{C},\mathcal{C}')$ denotes 
$(\mathcal{C}^v_{j,\mathbf{a}},\mathcal{C}^v_{j,\mathbf{b}})$. Assume that
\[
i\sum\limits_{w\in \mathcal{S}_{\mathbf{a}}} X_w,
\qquad
i\sum\limits_{w\in \mathcal{S}_{\mathbf{b}}} X_w
\;\in\;
\mathfrak{g}^v_{\Gamma,\mathrm{std}}.
\]
Then
\begin{equation}
\label{eq:IntersectionX}
i\sum\limits_{w\in \mathcal{S}_{\mathbf{a}}\cap \mathcal{S}_{\mathbf{b}}} X_w
\;\in\;
\mathfrak{g}^v_{\Gamma,\mathrm{std}}.
\end{equation}
\end{lem}

\begin{proof}
We prove the statement in the case 
$\mathcal{C}=\mathcal{C}^v_{j,\mathbf{a}}$ and 
$\mathcal{C}'=\mathcal{C}^v_{j,\mathbf{b}}$; 
the argument for $\mathcal{C}^v_{j,\ell,\mathbf{a}}$ and 
$\mathcal{C}^v_{j,\ell,\mathbf{b}}$ is analogous. For $j=1$, the sequences $\mathbf{a}$ and $\mathbf{b}$ are single binary digits. Since $\mathbf{a}\neq\mathbf{b}$, the sets $\mathcal{C}^v_{1,\mathbf{a}}$ and $\mathcal{C}^v_{1,\mathbf{b}}$ are disjoint, hence
\[
\mathcal{S}_{\mathbf{a}}\cap\mathcal{S}_{\mathbf{b}}=\varnothing,
\]
and the claim is vacuous.

Assume now that $j\ge 2$. We distinguish two cases.

\medskip
\noindent
\textbf{Case 1.} The truncated sequences
\[
\mathbf{a}'=(a_1,\dots,a_{j-1})
\quad\text{and}\quad
\mathbf{b}'=(b_1,\dots,b_{j-1})
\]
are distinct. Define
\[
\mathcal{Z}_{\mathbf{a}'}
:=
-0.25\,\operatorname{ad}^{2}_{X_{\mathcal{C}^v_{j,\mathbf{a}'}}}(\mathcal{Z}_{\widehat{v},2}).
\]
Then the element
\(
\sum\limits_{w\in \mathcal{S}_{\mathbf{a}}\cap\mathcal{S}_{\mathbf{b}}}X_w
\)
can be obtained via the identity

\begin{equation}
\label{eq:IntersectionCase1}
\sum\limits_{w\in \mathcal{S}_{\mathbf{a}}\cap\mathcal{S}_{\mathbf{b}}}X_w
=\mathcal{O}_m
\bigl(\operatorname{ad}^{2}_{\mathcal{Z}_{\mathbf{a}'}}\bigr)
\left(
\sum\limits_{w\in \mathcal{S}_{\mathbf{b}}}X_w
\right),
\end{equation}
and, therefore, lies in $\mathfrak{g}^v_{\Gamma,\mathrm{std}}$.

\medskip
\noindent
\textbf{Case 2.} The truncated sequences coincide:
\[
(a_1,\dots,a_{j-1})=(b_1,\dots,b_{j-1}),
\]
but $a_j\not\equiv b_j \pmod{2}$. Without loss of generality, assume $a_j\equiv 1\pmod{2}$. Then we compute
\begin{equation}
\label{eq:IntersectionCase2}
\begin{aligned}
&
\mathcal{O}_m\bigl(\operatorname{ad}^{2}_{\mathcal{Z}_{\widehat{v},2}}\bigr)
\left(
\sum\limits_{w\in \mathcal{S}_{\mathbf{a}}}X_w
+
\sum\limits_{w\in \mathcal{S}_{\mathbf{b}}}X_w
\right)
-
\sum\limits_{w\in \mathcal{S}_{\mathbf{b}}}X_w
\\
&\qquad
=
\sum\limits_{w\in \mathcal{S}_{\mathbf{b}}}X_w
+
2\sum\limits_{w\in \mathcal{S}_{\mathbf{a}}\cap\mathcal{S}_{\mathbf{b}}}X_w
-
\sum\limits_{w\in \mathcal{S}_{\mathbf{b}}}X_w
\\
&\qquad
=
\sum\limits_{w\in \mathcal{S}_{\mathbf{a}}\cap\mathcal{S}_{\mathbf{b}}}X_w,
\end{aligned}
\end{equation}
which again belongs to $\mathfrak{g}^v_{\Gamma,\mathrm{std}}$.
\end{proof}

\begin{proof}[Proof of Theorem~\ref{thm:ParitySeparationImpliesLocalX}]
Consider the distance shell $\mathcal{N}_{v,j}$. By Theorem~\ref{thm:CoolDLAElements}, for every binary sequence $\mathbf{a}\in\mathbb{Z}_2^j$, the Lie algebra $\mathfrak{g}^{\,v}_{\Gamma,\mathrm{std}}$ contains the element
\[
X_{\mathcal{C}^v_{j,\mathbf{a}}}
=
i\sum\limits_{w\in \mathcal{C}^v_{j,\mathbf{a}}} X_w.
\]

Let $\ell := |\mathcal{N}_{v,j}|$, and enumerate the vertices of $\mathcal{N}_{v,j}$ as $w_1,\dots,w_{\ell}$. If $\mathbf{a}\in\mathbb{Z}_2^j$ is such that $w'\in \mathcal{C}^v_{j,\mathbf{a}}$ while $w''\notin \mathcal{C}^v_{j,\mathbf{a}}$, then the element
\[
X_{\widehat{v},j}-X_{\mathcal{C}^v_{j,\mathbf{a}}}
\]
belongs to $\mathfrak{g}^{\,v}_{\Gamma,\mathrm{std}}$ (by Theorems~\ref{thm:DistDLAElements} and~\ref{thm:CoolDLAElements}) and contains $X_{w''}$ but not $X_{w'}$. Thus, for each parity class $\mathbf{a}$, we have access to uniform $X$-operators supported both on $\mathcal{C}^v_{j,\mathbf{a}}$ and on its complement within $\mathcal{N}_{v,j}$.

For $t\in\{2,\dots,\ell\}$, let $\mathcal{C}^v_{j,t}\subseteq \mathcal{N}_{v,j}$ be the subset such that $\mathcal{C}^v_{j,t}$ contains $w_1$ but not $w_t$. Such subsets can be chosen either directly from the parity classes appearing in the statement of the theorem or from their complements in $\mathcal{N}_{v,j}$.

Applying Lemma~\ref{lem:IntersectionFromParity} iteratively to the intersections
\[
\mathcal{C}^v_{j,2}\cap\mathcal{C}^v_{j,3},
\quad
\mathcal{C}^v_{j,2}\cap\mathcal{C}^v_{j,3}\cap\mathcal{C}^v_{j,4},
\quad
\dots,
\]
and observing that
\[
\mathcal{C}^v_{j,2}\cap\dots\cap\mathcal{C}^v_{j,\ell}=\{w_1\},
\]
we conclude that the element
\(
iX_{w_1}
\)
belongs to $\mathfrak{g}^{\,v}_{\Gamma,\mathrm{std}}$. The same argument applies to each $w_t\in\mathcal{N}_{v,j}$, yielding the containment in \eqref{eq:LocalXInReducedDLA}.
\end{proof}

Theorem \ref{thm:ParitySeparationImpliesLocalX} is verified using an argument analogous to those in the proofs of
Lemma~\ref{lem:IntersectionFromParity} and
Theorem~\ref{thm:ParitySeparationImpliesLocalX}.

\subsection{Verification of the DLA Containment Condition in~\eqref{eq:ERContainment} for Connected Acyclic Graphs}

We conclude this segment with the verification of the DLA containment condition in~\eqref{eq:ERContainment} for connected acyclic graphs.

\begin{proof}[Proof of Theorem~\ref{thm:ReducedDLAsForTrees}]
Let
\[
M := \max\{ j \ge 0 \mid \mathcal{N}_{v,j} \neq \varnothing \}
\]
be the maximal distance from the distinguished vertex $v$.
Since $\Gamma$ is acyclic, every vertex in $\mathcal{N}_{v,M}$ is a leaf.

By assumption, the parity-degree sequences associated to the unique paths
connecting $v$ to distinct leaves in $\mathcal{N}_{v,M}$ are pairwise distinct.
Therefore, by Theorem~\ref{thm:CoolDLAElements} and
Theorem~\ref{thm:ParitySeparationImpliesLocalX}, we obtain
\[
iX_w \in \mathfrak{g}^{\,v}_{\Gamma,\mathrm{std}}
\qquad
\text{for all } w \in \mathcal{N}_{v,M}.
\]

Each vertex $u \in \mathcal{N}_{v,M-1}$ is adjacent to at least one leaf
$w \in \mathcal{N}_{v,M}$.
Applying Remark~\ref{rmk:ExtensionOfResults}, we conclude that the presence of
$iX_w$ in the reduced dynamical Lie algebra implies
\[
iX_u \in \mathfrak{g}^{\,v}_{\Gamma,\mathrm{std}}
\qquad
\text{for all } u \in \mathcal{N}_{v,M-1}.
\]
Iterating this argument along the unique paths connecting the leaves in
$\mathcal{N}_{v,M}$ to the root $v$, we deduce that
\[
iX_u \in \mathfrak{g}^{\,v}_{\Gamma,\mathrm{std}}
\]
for every vertex $u$ lying on any such path.

Next, consider the set of leaves at the next maximal distance
\[
M_2 := \max\{ j < M \mid \mathcal{N}_{v,j} \text{ contains a leaf} \}.
\]
For each leaf $w \in \mathcal{N}_{v,M_2}$, the hypothesis of the theorem again guarantees that its parity-degree profile along the path to $v$ is distinct from those of other leaves in the same layer.
Thus, Theorem~\ref{thm:ParitySeparationImpliesLocalX} implies
\[
iX_w \in \mathfrak{g}^{\,v}_{\Gamma,\mathrm{std}}
\qquad
\text{for all leaf vertices } w \in \mathcal{N}_{v,M_2}.
\]
As before, Remark~\ref{rmk:ExtensionOfResults} allows us to propagate this
containment upward along the unique paths from these leaves to $v$.

Proceeding inductively over decreasing distances from $v$, and exhausting all
leaf layers of $\Gamma$, we conclude that
\[
iX_u \in \mathfrak{g}^{\,v}_{\Gamma,\mathrm{std}}
\qquad
\text{for every vertex } u \in V(\Gamma).
\]
This establishes the containment of algebras in~\eqref{eq:ERContainment} and
completes the proof.
\end{proof}

\begin{proof}[Proof of Theorem~\ref{thm:FreeReducedFullStructure}]
Under the assumptions on $\Gamma_v$, Theorem~1 of~\cite{KLFCCZ} implies that the associated free DLA on the reduced graph,
\(
\mathfrak{g}_{\Gamma_v,\mathrm{free}},
\)
is isomorphic to
\[
\mathfrak{su}(2^{n-2}) \oplus \mathfrak{su}(2^{n-2}),
\]
acting irreducibly on the two parity sectors
\(
W_{\mathrm{even}}, \; W_{\mathrm{odd}} \subset W_v
\)
introduced in~\eqref{eq:EvenOddSpaces}.

Number the qubits corresponding to the vertices of $\Gamma$ so that the distinguished vertex $v$ corresponds to the $n^{th}$ (last) qubit. With this convention, the above Lie algebra can be characterized as the centralizer in $\mathfrak{su}(W_v)=\mathfrak{su}(2^{n-1})$ of the global parity operator
\[
\tau:=X_1X_2\cdots X_{n-1}
\]
(see Remark \ref{rmk:GlobalParity}).

A convenient basis description of $\mathfrak{g}_{\Gamma_v,\mathrm{free}}$ is obtained in terms of Pauli strings. 
For a Pauli string 
\(
\mathcal{P}=\sigma_1\otimes\cdots\otimes\sigma_{n-1}
\)
with $\sigma_j\in\{I,X,Y,Z\}$, one verifies directly that
\[
\tau\,\mathcal{P}\,\tau = (-1)^{k(\mathcal{P})}\mathcal{P},
\]
where $k(\mathcal{P})$ is the total number of factors equal to $Y$ or $Z$. Hence $\mathcal{P}$ commutes with $\tau$ if and only if $k(\mathcal{P})$ is even. Therefore the reduced algebra is spanned by
\begin{equation}
    \label{eq:PathParity}
    \{\, i\mathcal{P} \mid \mathcal{P}\neq I,\tau;\; k(\mathcal{P})\equiv 0 \!\!\!\pmod 2 \,\},
\end{equation}

that is, by all traceless Pauli strings containing an even total number of $Y$ and $Z$ factors, except for identity and $\tau$ itself (see Table V in \cite{KLFCCZ}).

\medskip

By definition of $\mathfrak{g}^{\,v}_{\Gamma,\mathrm{free}}$, it contains 
\(
iZ_w
\)
for all $w\in\mathcal{N}_{v,1}$.  
Let $w'\in V(\Gamma_v)\setminus \mathcal{N}_{v,1}$ be any vertex not adjacent to $v$ in $\Gamma$, and choose $w\in\mathcal{N}_{v,1}$.  

Observe that the Pauli strings $X_wX_{w'}$ and $Y_wY_{w'}$ satisfy~\eqref{eq:PathParity}.  
Hence
\[
iX_wX_{w'},\quad iY_wY_{w'}
\;\in\;
\mathfrak{g}^{\,v}_{\Gamma,\mathrm{free}}.
\]
Consequently, the element
\[
[iX_wX_{w'},[\,iZ_w,iY_wY_{w'}\,]]
\]
is proportional to $iZ_{w'}$.   Thus we confirm that $iZ_{w'}\in \mathfrak{g}^{\,v}_{\Gamma,\mathrm{free}}$ for every $w'\in V(\Gamma_v)$.  Since, by definition, all $iX_{w'}$ also belong to $\mathfrak{g}^{\,v}_{\Gamma,\mathrm{free}}$, we obtain the full single-site algebra
\[
\mathrm{span}\{\, iX_{w'},iY_{w'},iZ_{w'} \,\}
\cong
\mathfrak{su}(2)
\]
at each vertex $w'\in V(\Gamma_v)$.

\medskip

We now prove that the enlarged algebra equals $\mathfrak{su}(2^{n-1})$.  
Let $\mathcal{Q}\notin\{I,\tau\}$ be an arbitrary Pauli string.

\smallskip
\emph{Case 1:} $k(\mathcal{Q})$ is even.  

Then $i\mathcal{Q}$ already belongs to the free DLA of the reduced subgraph by~\eqref{eq:PathParity}.

\smallskip
\emph{Case 2:} $k(\mathcal{Q})$ is odd.  

Then there exists a site $w'$ such that $\mathcal{Q}_{w'}\in\{Y,Z\}$.  
Define a Pauli string $\mathcal{Q}'$ by
\[
\mathcal{Q}'_{w}
=
\begin{cases}
X, & w = w', \\[4pt]
\mathcal{Q}_{w}, & \text{otherwise},
\end{cases}
\]
that is, $\mathcal{Q}'$ coincides with $\mathcal{Q}$ at every site except $w'$, where the factor is replaced by $X$.  
Then $k(\mathcal{Q}')$ is even, and therefore
\[
i\mathcal{Q}'
\in
\mathfrak{g}^{\,v}_{\Gamma,\mathrm{free}}.
\]

Since $\mathcal{Q}$ and $\mathcal{Q}'$ differ only at the site $w'$, and since we have the full local algebra $\mathfrak{su}(2)$ available at $w'$, suitable commutators with
\[
iX_{w'},\quad iY_{w'},\quad iZ_{w'}
\]
transform $i\mathcal{Q}'$ into $i\mathcal{Q}$ (up to a nonzero scalar multiple).  
Hence $i\mathcal{Q}$ also lies in $\mathfrak{g}^{\,v}_{\Gamma,\mathrm{free}}$.

\medskip

Thus every traceless Pauli string belongs to the free reduced DLA. Since the traceless Pauli strings form a basis of $\mathfrak{su}(2^{n-1})$, we conclude that
\(
\mathfrak{g}^{\,v}_{\Gamma,\mathrm{free}}=\mathfrak{su}(2^{n-1}).
\)
\end{proof}

Corollary~\ref{thm:ReducedDLAIsFullUnitary} follows directly from  Theorem~\ref{thm:FreeReducedFullStructure} together with the assumption that the algebra containment~\eqref{eq:ERContainment} holds. The latter ensures that the standard and free reduced dynamical Lie algebras coincide,
\[
\mathfrak{g}^v_{\Gamma,\mathrm{std}}
=
\mathfrak{g}^v_{\Gamma,\mathrm{free}}.
\]

\section{Proofs of Theorems~\ref{DecompThm} and~\ref{DecompThmReduced}}
\label{sec:ProofsHilbSpaceDecomp}

The proofs of Theorems~\ref{DecompThm} and~\ref{DecompThmReduced} follow the same general strategy. They rely on the fact that, for a connected graph $\Gamma$, the image of the Lie algebra $\mathfrak{g}_{\Gamma,\mathrm{free}}$ inside $\mathfrak{su}(W)$ generates, in its associative envelope,
all two-body Pauli operators $Z_w Z_w'$ for arbitrary pairs of vertices $w,w'\in V$. In the reduced setting, an analogous statement holds for $\mathfrak{g}^v_{\Gamma,\mathrm{free}}$, whose image generates all single-site operators $iZ_w$ in the associative envelope inside
$\mathfrak{su}(W_v)$.

These observations allow us to establish irreducibility of the relevant Hilbert spaces. We present the detailed arguments below.

\begin{proof}[Proof of Theorem \ref{DecompThm}]
    First, note that it follows from the actions of single-gate Pauli operators in the Hadamard basis: \( X(\ket{+}) = \ket{+} \), \( X(\ket{-}) = -\ket{-} \), \( Z(\ket{+}) = \ket{-} \), and \( Z(\ket{-}) = \ket{+} \) that the generators of \(\mathfrak{g}_{\Gamma, \text{free}}\) preserve the parity of the number of \(\ket{-}\) in the Hadamard basis for a standard basis vector. Consequently, the Lie algebra \(\mathfrak{g}_{\Gamma, \text{free}}\), generated by \(\{ X_i \mid i \in V \}\) and \(\{ Z_i Z_j \mid (ij) \in E \}\), preserves the subspaces \( W_{\text{even}} \) and \( W_{\text{odd}} \). 

    Next, we show that both \( W_{\text{even}} \) and \( W_{\text{odd}} \) are irreducible representations of \(\mathfrak{g}_{\Gamma, \text{free}}\). Let \( v \in W_{\text{even}} \) (or \( v \in W_{\text{odd}} \)) be a nonzero vector. Write 
    \begin{equation}
        v = \sum \alpha_h \ket{h},
        \label{Vdecomp}
    \end{equation}
    where each \(\ket{h}\) is a Hadamard basis vector with an even (or odd) number of \(\ket{-}\). We claim that \( \operatorname{span}(X_1, \dots, X_n )(v) \subseteq \operatorname{span}(\mathfrak{g}_{\Gamma, \text{free}})(v) \) contains a pure state vector in the Hadamard basis. 

    To prove this, we argue by induction on \( n \), the number of vertices in \(\Gamma\). For the base case \( n = 1 \), the statement is trivial since both \( W_{\text{even}} = \mathbb{C}\ket{+} \) and \( W_{\text{odd}} = \mathbb{C}\ket{-} \) are one-dimensional. For the inductive step, suppose the claim holds for \( n-1 \). Consider a vector \( v \in W_{\text{even}} \) (or \( W_{\text{odd}} \)). If \( v \) can be decomposed as \( v = \ket{-} \otimes v' \) or \( v = \ket{+} \otimes v' \), the inductive hypothesis applied to \( v' \) ensures the existence of a standard Hadamard basis vector \(\widehat{v}' \in \operatorname{span}(\{ X_2, \dots, X_n \})(v')\), leading to a vector \( \widetilde{v} = \ket{-} \otimes \widehat{v}' \) or \( \widetilde{v} = \ket{+} \otimes \widehat{v}' \) in \(\operatorname{span}(\{ X_1, \dots, X_n \})(v)\).

    In the second case, where \( v \) cannot be decomposed as \( v = \ket{-} \otimes v' \) or \( v = \ket{+} \otimes v' \), there exists a basis vector \(\ket{h}\) in the decomposition \eqref{Vdecomp} with \( h_1 = \ket{+} \). In this scenario, the vector \( v + X_1(v) \) is nonzero and decomposable as \( \ket{+} \otimes v' \), reducing the problem to the first case. This establishes that \(\operatorname{span}(\{ X_1, \dots, X_n \})(v)\) contains a pure Hadamard basis vector.

    Finally, to prove irreducibility, it suffices to show that \(\mathfrak{g}_{\Gamma, \text{free}}(h) \simeq W_{\text{even}}\) (or \(\mathfrak{g}_{\Gamma, \text{free}}(h) \simeq W_{\text{odd}}\)) for a Hadamard basis vector \( h \).
    
     Since $\Gamma$ is connected, for any pair of vertices $w,w'$ there exists a path
    \[
       w = i_0,- i_1- \dots- i_k = w'
    \]
    connecting $w$ to $w'$. As $Z_{i_\ell}Z_{i_{\ell+1}}$ lies in the image of $\mathfrak{g}_{\Gamma,\mathrm{free}}$ for each edge $(i_\ell,i_{\ell+1})\in E$, their product belongs to the associative algebra generated by this image and satisfies
    \begin{equation}
        Z_{i_0}Z_{i_1}Z_{i_1} Z_{i_2} \cdots Z_{i_{k-1}} Z_{i_k} \;=\; Z_wZ_{w'}.
        \label{Zpath}
    \end{equation}

    Finally, the set of operators  
\[
\{ Z_w Z_{w'} \mid (i, j) \in V \times V \}
\]  
acts transitively on the Hadamard basis of \( W_{\text{even}} \) (or equivalently, \( W_{\text{odd}} \)).  
\end{proof}

\begin{proof}[Proof of Theorem \ref{DecompThmReduced}]
    Let $w = \sum \alpha_h \ket{h}$ be a nonzero vector in $W_v$, expressed in the Hadamard basis.  
    By the same argument as in the proof of Theorem~\ref{DecompThm}, the subspace 
    \[
       \operatorname{span}\{ X_i \mid i \neq v \}(w) \;\subseteq\; \operatorname{span}\bigl(\widetilde{\mathfrak{g}}_{\Gamma,\mathrm{free}}\bigr)(w)
    \] 
    contains a pure Hadamard basis vector. It remains to show that the cyclic subspace generated by any single Hadamard basis vector  under the action of \(\mathfrak{g}^{\,v}_{\Gamma,\mathrm{free}}\) contains all the other basis vectors. Equivalently, starting from one such basis vector, we must be able to produce every other basis vector using operators in the associative algebra generated by the elements of \(\mathfrak{g}^{\,v}_{\Gamma,\mathrm{free}}\). 

    To achieve this, it suffices to realize operators that swap $\ket{+}$ and $\ket{-}$ at any chosen qubit, while leaving the other qubits unchanged. We will show that such operators are contained in the associative algebra generated by the elements of $\mathfrak{g}^v_{\Gamma,\mathrm{free}}$.  

    Since $\Gamma$ is connected, for any vertex $w$ of $\Gamma_v$ there exists a path
    \[
       v = i_0,\, i_1,\, \dots,\, i_k = w
    \]
    connecting $v$ to $w$. This allows to obtain the element
    \begin{equation}
        Z_{i_1} Z_{i_2} \cdots Z_{i_{k-1}} Z_{i_k} \;=\; Z_j,
        \label{ZpathReduced}
    \end{equation}
    which provides the necessary single-qubit $Z$ operator at position $w$ inside the associative algebra generated by the image of $\mathfrak{g}^v_{\Gamma,\mathrm{free}}$. 

\end{proof}

\section{Additional Numerical Results}

In this appendix, we include additional numerical results for the variance in the gradient for 11-, 13-, and 15-node asymmetric graphs. Some of these graphs (see e.g., $\Gamma_{5}^{(11)}$, $\Gamma_{7}^{(11)}$, $\Gamma_{3}^{(13)}$, $\Gamma_{5}^{(15)}$) corroborate Observation \ref{leaf_obs}.

\afterpage{\clearpage
    \begin{figure}
        \centering
        \includegraphics[width=2\paperwidth, height=0.8\paperheight, keepaspectratio]{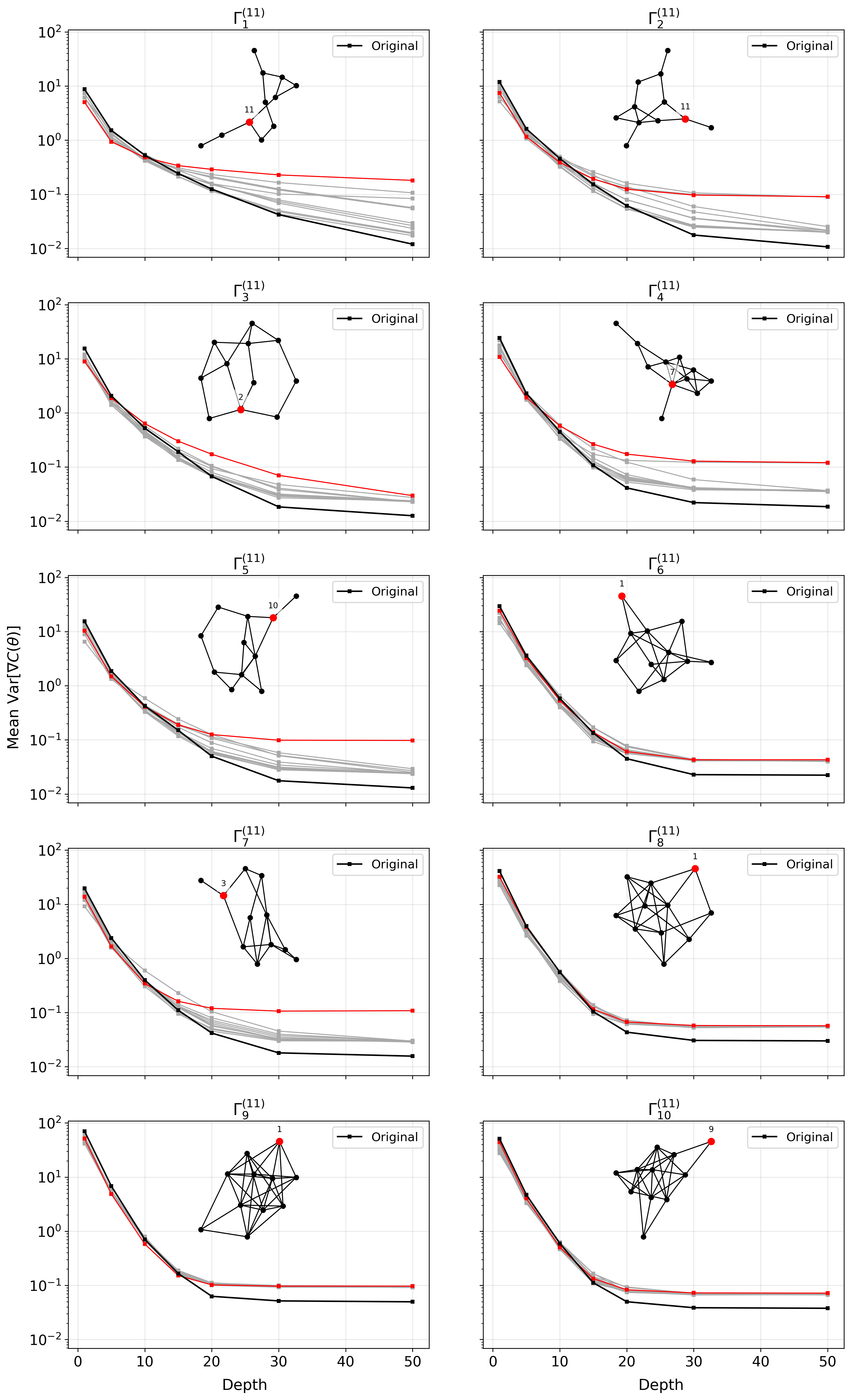}
        \caption{Mean of variance for ten 11-node asymmetric graphs ($\Gamma^{(11)}$). The red curve shows the vertex reduction that results in the largest variance of the gradient at depth 50, with the corresponding vertex highlighted in red and labeled on the graph inset.}
        \label{fig:11_node_variances}
    \end{figure}
    \clearpage
}

\afterpage{\clearpage
    \begin{figure}
        \centering
        \includegraphics[width=2\paperwidth, height=0.8\paperheight, keepaspectratio]{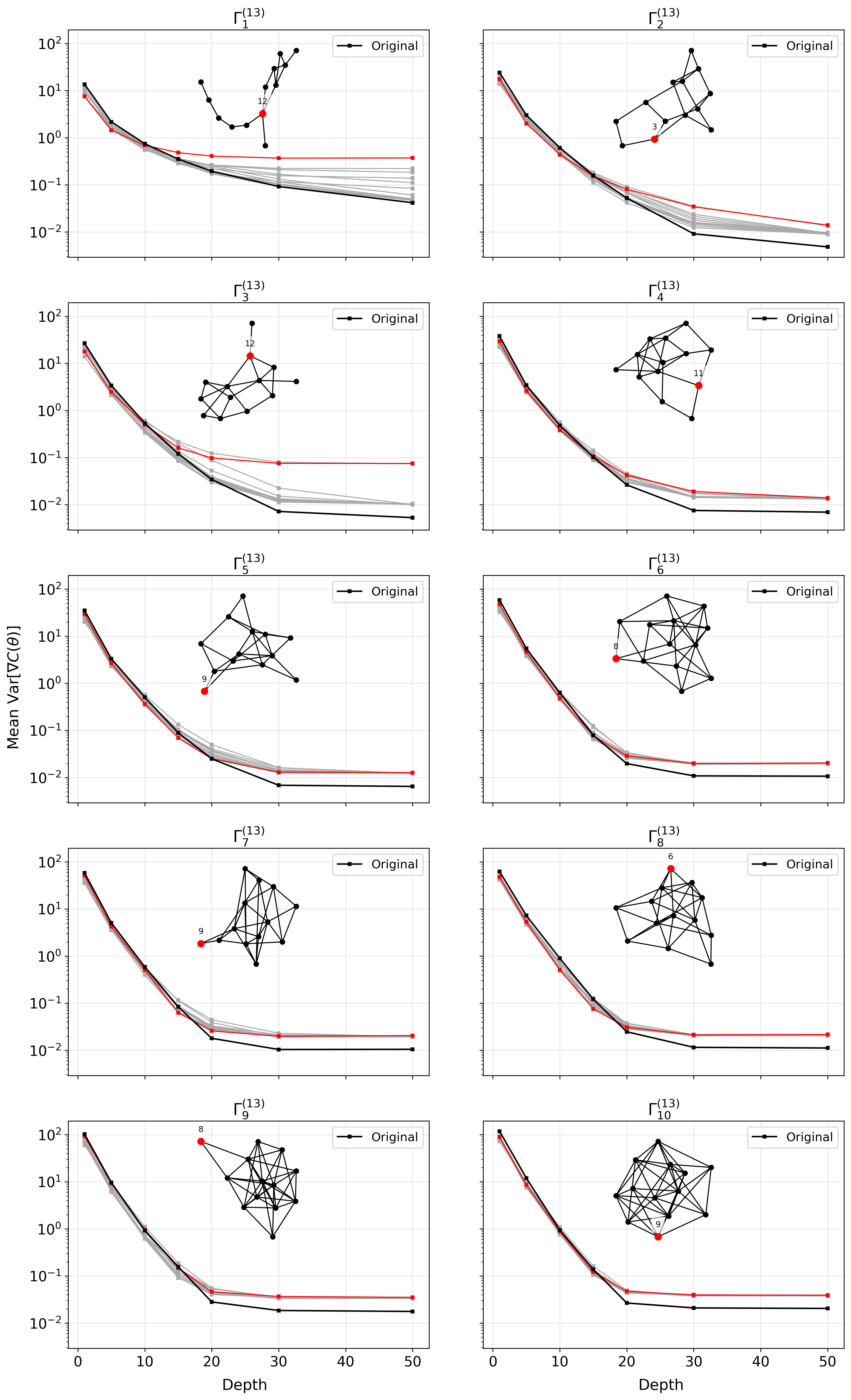}
        \caption{Mean of variance for ten 13-node asymmetric graphs ($\Gamma^{(13)}$). The red curve shows the vertex reduction that results in the largest variance of the gradient at depth 50, with the corresponding vertex highlighted in red and labeled on the graph inset.}
        \label{fig:13_node_variances}
    \end{figure}
    \clearpage
}

\afterpage{\clearpage
    \begin{figure}
        \centering
        \includegraphics[width=2\paperwidth, height=0.8\paperheight, keepaspectratio]{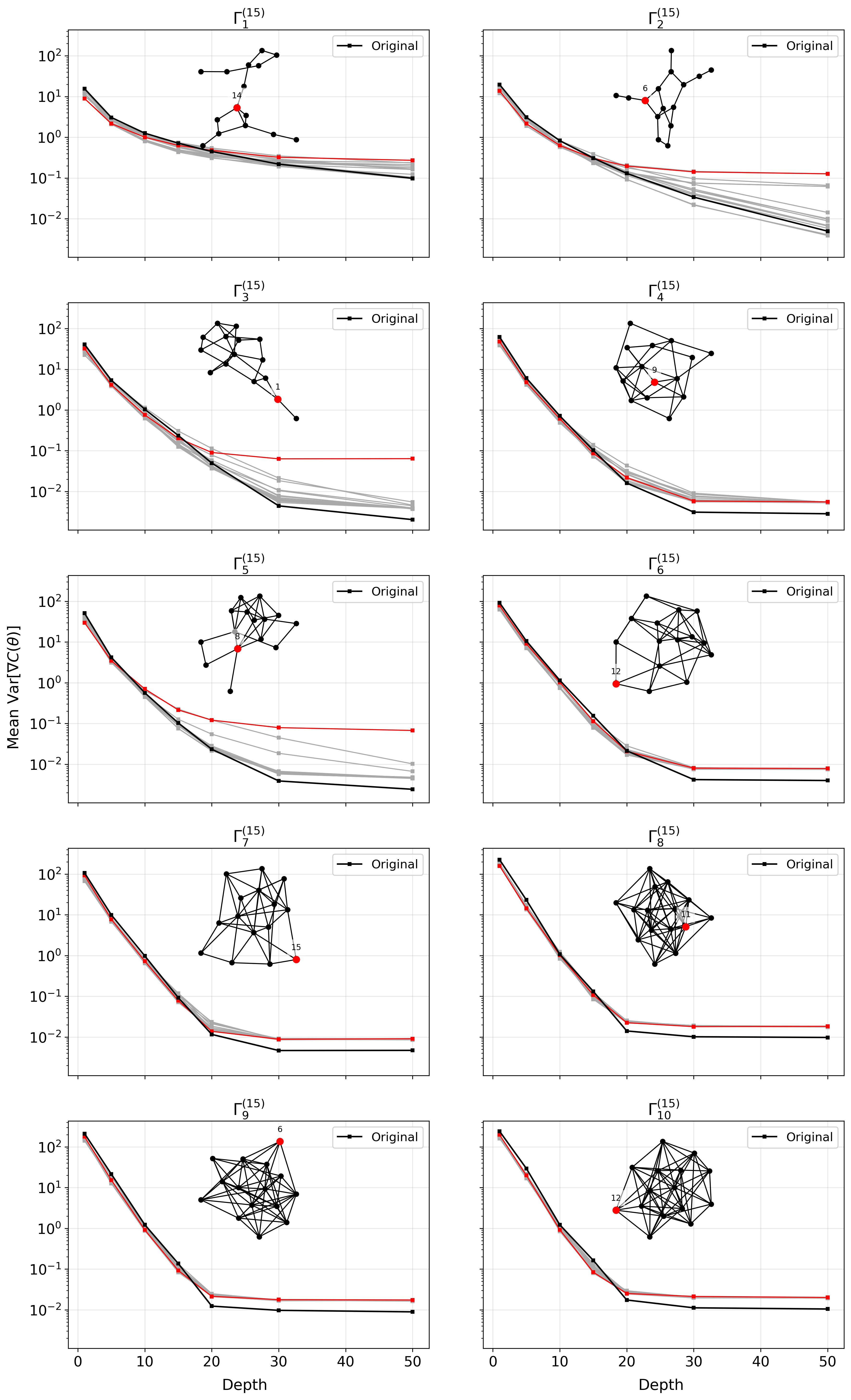}
        \caption{Mean of variance for ten 15-node asymmetric graphs ($\Gamma^{(15)}$). The red curve shows the vertex reduction that results in the largest variance of the gradient at depth 50, with the corresponding vertex highlighted in red and labeled on the graph inset.}
        \label{fig:15_node_variances}
    \end{figure}
    \clearpage
}

\end{document}